\documentclass[11pt]{article}
\usepackage[margin=1.6in]{geometry}
\usepackage[utf8]{inputenc}
\usepackage{amsmath}
\usepackage{amsthm}
\usepackage{subcaption}
\usepackage{tikz}
\usepackage{enumerate}
\usepackage{hyperref}
\usepackage{forest}
\usepackage{authblk}
\usepackage{float}
\usepackage{algorithm2e}

\usetikzlibrary{decorations.pathreplacing}
\usetikzlibrary{patterns}

\newtheorem{theorem}{Theorem}
\newtheorem{proposition}[theorem]{Proposition}
\newtheorem{definition}[theorem]{Definition}
\newtheorem{lemma}[theorem]{Lemma}
\newtheorem{corollary}[theorem]{Corollary}
\newtheorem{claim}{Claim}[theorem]

\newcommand{\gbalgo}{the \texttt{git bisect} algorithm}
\newcommand{\vertex}[1]{\boldsymbol{#1}}
\newcommand{\ancestors}[1]{|#1|}
\newcommand{\faulty}{{faulty commit}}
\newcommand{\gb}{\texttt{git bisect}}
\newcommand{\golden}{\texttt{golden bisect}}
\newcommand{\score}{\mathsf{score}}

\definecolor{ibm_blue}{HTML}{648FFF}
\definecolor{ibm_purple}{HTML}{785EF0}
\definecolor{ibm_magenta}{HTML}{DC267F}
\definecolor{ibm_orange}{HTML}{FE6100}
\definecolor{ibm_yellow}{HTML}{FFB000}
\definecolor{deep_purple}{HTML}{623CEA}
\definecolor{grey}{HTML}{EFF1ED}
\definecolor{light_grey}{HTML}{f4efee}

\title{\textbf{Theoretical analysis of git bisect}}
\author{Julien Courtiel}
\author{Paul Dorbec}
\author{Romain Lecoq}
\affil{\small Normandie Univ, UNICAEN, ENSICAEN, CNRS, GREYC, 14000 Caen, France.}
\date{}

\begin{document}

\maketitle

\begin{abstract}
   In this paper, we consider the problem of finding a regression in a version
   control system (VCS), such as \texttt{git}.
   The set of versions is modelled by a Directed Acyclic Graph (DAG) where
   vertices represent versions of the software, and arcs are the changes between different versions. We
   assume that somewhere in the DAG, a bug was introduced, which persists
   in all of its subsequent versions. It is possible to query a vertex to
   check whether the corresponding version carries the bug. Given a DAG and
   a bugged vertex, the Regression Search Problem consists in finding the
   first vertex containing the bug in a minimum number of queries in the
   worst-case scenario. This problem is known to be NP-complete.

   We study the algorithm used in \texttt{git} to address this problem, known as
   \texttt{git bisect}. We prove that in a general setting, \texttt{git bisect} can use an
   exponentially larger number of queries than an optimal algorithm. We
   also consider the restriction where all vertices have indegree at most
   2 (i.e. where merges are made between at most two branches at a time in the VCS),
   and prove that in this case, \texttt{git bisect} is a
   $\frac{1}{\log_2(3/2)}$-approximation algorithm, and that this bound is
   tight. We also provide a better approximation algorithm for this case.

   Finally, we give an alternative proof of the NP-completeness of the Regression Search Problem, via a variation with bounded indegree.
\end{abstract}

\section{Introduction}

In the context of software development, it is essential to resort to Version Control Systems (VCS, in short), like \texttt{git} or \texttt{mercurial}. VCS enable many developers to work concurrently on the same system of files. Notably, all the \emph{versions} of the project (that is to say the different states of the project over time) are saved by the VCS, as well as the different changes between versions.

Furthermore, many VCS offer the possibility of creating \emph{branches} (i.e. parallel lines of development) and \emph{merging} them, so that individuals can work on their own part of the project, with no risk of interfering with other developers work.
Thereby the overall structure can be seen as a Directed Acyclic Graph (DAG), where the vertices are the versions, also named in this context \emph{commits}, and the arcs model the changes between two versions.

\smallskip

The current paper deals with a problem often occurring in projects of large size: searching the origin of a so-called \emph{regression}.
Even with intensive testing techniques, it seems unavoidable to find out long-standing bugs which have been lying undetected for some time.
Conveniently, one tries to fix this bug by finding the commit in which the bug appeared for the first time.
The idea is that there should be few differences between the code source of the commit that introduced the bug and the one from a previous bug-free commit, which makes it easier to find and fix the bug.

The identification of the faulty commit is possible by performing \emph{queries} on existing commits. A query allows to figure out the status of the commit: whether it is \emph{bugged} or it is \emph{clean}. A single query can be very time-consuming: it may require running tests, manual checks, or the compilation of an entire source code.
In some large projects, performing a query on a single commit can take up to a full day (for example, the Linux kernel project~\cite{gb-website}).
This is why it is essential to find the commit that introduced the bug with as few queries as possible.

The problem of finding an optimal solution in terms of number of queries, known as the Regression Search Problem, was proved to be NP-complete by Carmo, Donadelli, Kohayakawa and Laber in~\cite{CDKY:searching-in-posets}.
However, whenever the DAG is a tree (oriented from the leaves to the root),
the computational complexity of the Regression Search Problem is polynomial \cite{Ben-Asher97optimalsearch,Onak-Parys}, and even linear \cite{MOW:linear}.

To our knowledge, very few papers in the literature deal with the Regression Search Problem in the worst-case scenario, as such. The Decision Tree problem, which is known to be NP-complete \cite{hyafil-rivest} as well as its approximation version \cite{laber:DT-approximation}, somehow generalises the Regression Search Problem, with this difference that the Decision Tree problem aims to minimise the average number of queries instead of the worst-case number of queries.

Many variations of the Regression Search problem exist:
\begin{itemize}
\item the costs of the queries may vary \cite{DKUZ:weighted,EKS:probabilistic-binary-search};
\item the queries return the wrong result (say it is clean while the vertex is bugged or the converse) with a certain probability \cite{EKS:probabilistic-binary-search};
\item one can just try to find a bugged vertex with at least one clean parent \cite{bendk}.
\end{itemize}

The most popular VCS today, namely \texttt{git}, proposes a tool for this problem: an algorithm named \gb. It is a heuristic inspired by binary search that narrows down at each query the range of the possible faulty commits.
This algorithm is widely used and shows excellent experimental results, though to our knowledge, no mathematical study of its performance have been carried out up to now.

In this paper, we fill this gap by providing a careful analysis on
the number of queries that \gb\ uses compared to an optimal strategy.
This paper does not aim to find new approaches for the Regression Search Problem.

First, we show in Section~\ref{sec:general} that, in the general case,
\gb\ may be very inefficient, testing about half the commits
where an optimal logarithmic number of commits can be used to identify exactly the faulty vertex.
But in all the cases where such bad performance occurs, there are large merges between more than two branches,\footnote{According to \href{https://www.destroyallsoftware.com/blog/2017/the-biggest-and-weirdest-commits-in-linux-kernel-git-history}{this blog}, a merge of 66 branches happened in the Linux kernel repository.} also named \emph{octopus merges}.
However, such merges are highly uncommon and inadvisable, so we carry out the study of \gb\ performances with the assumption that the DAG does not contain any octopus merge, that is, every vertex has indegree at most two.
Under such an assumption, we are able to prove in Section~\ref{sec:binary} that \gb\ is an approximation algorithm for the problem, never using more than ${\frac 1 {\log_2(3/2)} \approx 1.71}$ times the optimal number of queries for large enough repositories.
We also provide a family of DAGs for which the number of queries used by \gb\ tends to $\frac 1 {\log_2(3/2)}$ times the optimal number of queries.

This paper also describes in Section~\ref{sec:golden} a new algorithm, which is a refinement of \gb. This new algorithm, which we call \golden,  offers a mathematical guaranteed ratio of $\frac{1}{\log_2(\phi)} \approx 1.44$ for DAGs with indegree at most $2$ where $\phi= \frac{1+\sqrt{5}}{2}$ is the golden ratio. The search of new efficient algorithms for the Regression Search Problem seems to be crucial in software engineering (as evidenced by~\cite{bendk}); \golden\ is an example of progress in this direction.

The good performances of \gb\ and \golden\ in the binary case raise a last question.
Is the problem still NP-complete if the inputs are restricted to binary DAGs?
In Section~\ref{sec:np_complete}, we consider a variation, the Confined Regression Search Problem (CRSP), which is NP-complete even in the binary case.
This variation is equivalent to the Regression Search Problem (RSP) in the general case,
so this gives a new proof of the complexity of this problem.
However, this does not extend to RSP in the binary case.

\subsection{Formal definitions}

Throughout the paper, we refer to VCS repositories as graphs, and more precisely as \emph{Directed Acyclic Graphs} (DAGs), i.e., directed graphs with no cycle. The set $V$ of vertices corresponds to the versions of the software. An arc goes from a vertex $\vertex p$ to another vertex $\vertex v$ if $\vertex v$ is obtained by a modification from $\vertex p$. We then say that $\vertex p$ is a \emph{parent} of $\vertex v$. A vertex may have multiple parents in the case of a merge. An \emph{ancestor} of $\vertex v$ is $\vertex v$ itself or an ancestor of a parent ${\text{of } \vertex v}$.\footnote{Usually, $\vertex v$ is not considered an ancestor of itself. Though, for simplifying the terminology, we use this special convention here.} Equivalently, a vertex is an ancestor of $\vertex v$ if and only if it is co-accessible from $\vertex v$ (i.e., there exists a path from this vertex to $\vertex v$).

We use the convention to write vertices in bold (for example $\vertex v$), and the number of ancestors of a vertex with its name between two vertical bars (for example $\ancestors v$).

In our DAGs, we consider that a bug has been introduced at some vertex, named the \emph{\faulty}. This vertex is unique, and its position is unknown. The \faulty\  is supposed to transmit the bug to each of its \emph{descendants} (that is, its children, its grand-children, and so on). Thus, vertices have two possible statuses: \emph{bugged} or \emph{clean}. A vertex is bugged if and only if it has the \faulty\ as an ancestor. Other vertices are clean. This is illustrated by Figure~\ref{fig:bug_spread}.

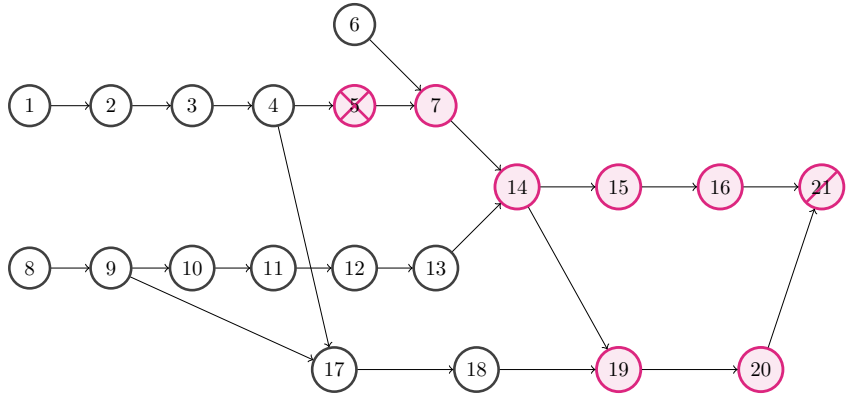
\begin{figure}[ht]
   \centering
   \scalebox{0.9}{
      \begin{tikzpicture}
	[
		scale=.6,
		minimum size=0.75cm,
		every node/.style={circle,draw,scale=0.8,text=black,line width=1.1pt},
		buggedNode/.style={text opacity=1,opacity=0.1,draw opacity=1,fill=ibm_magenta, color=ibm_magenta,text=black, line width=1.2pt},
		normalNode/.style={text opacity=1,opacity=0,draw opacity=1,color=darkgray,text=black, line width=1.1pt},
	]

    \tikzset{
        cross/.pic = {
        \draw[rotate = 45,color=ibm_magenta] (-#1,0) -- (#1,0);
        \draw[rotate = 45,color=ibm_magenta] (0,-#1) -- (0, #1);
        }
    }
	\tikzset{
        halfcross/.pic = {
        \draw[rotate = 45,color=ibm_magenta] (-#1,0) -- (#1,0);
        }
    }

    \draw (8,6) pic[black,opacity=1] {cross=11pt};
    \draw (19.5, 4) pic[black,opacity=1] {halfcross=11pt};

	\node[normalNode] (1) at (0, 6) {$1$};
	\node[normalNode] (2) at (2, 6) {$2$};
	\node[normalNode] (3) at (8,8) {$6$};
	\node[normalNode] (4) at (4,6) {$3$};
	\node[normalNode] (5) at (6,6) {$4$};
	\node[buggedNode] (6) at (8,6) {$5$};
	\node[buggedNode] (7) at (10, 6) {$7$};
	\node[normalNode] (8) at (0, 2) {$8$};
	\node[normalNode] (9) at (2,2) {$9$};
	\node[normalNode] (10) at (4, 2) {$10$};
	\node[normalNode] (11) at (6, 2) {$11$};
	\node[normalNode] (12) at (8, 2) {$12$};
	\node[normalNode] (13) at (10, 2) {$13$};
	\node[buggedNode] (14) at (12, 4) {$14$};
	\node[buggedNode] (15) at (14.5, 4) {$15$};
	\node[buggedNode] (16) at (17, 4) {$16$};
	\node[buggedNode] (17) at (19.5, 4) {$21$};
	\node[normalNode] (18) at (7.5, -0.5) {$17$};
	\node[normalNode] (19) at (11, -0.5) {$18$};
	\node[buggedNode] (20) at (14.5, -0.5) {$19$};
	\node[buggedNode] (21) at (18, -0.5) {$20$};

	\path[->,color=black]
		(1) edge (2)
		(2) edge (4)
		(4) edge (5)
		(5) edge (6)
		(6) edge (7)
		(3) edge (7)
		(7) edge (14)
		
		(8) edge (9)
		(9) edge (10)
		(10) edge (11)
		(11) edge (12)
		(12) edge (13)
		(13) edge (14)
		(14) edge (15)
		(15) edge (16)
		(16) edge (17)
		
		(18) edge (19)
		(19) edge (20)
		(20) edge (21)
		(21) edge (17)
		
		(5) edge (18)
		(9) edge (18)
		(14) edge (20)
	;

\end{tikzpicture}
   }
   \caption{An example of a DAG. The bugged vertices are coloured. The strikeout vertex ($\vertex{21}$) is the marked vertex, known to be bugged. The crossed  vertex ($\vertex 5$) is the \faulty.}
   \label{fig:bug_spread}
\end{figure}

We consider the problem of identifying the \faulty\ in a DAG $D$, where a bugged vertex $\vertex b$ is identified.
It is addressed by performing \emph{queries} on vertices of the graph. Each query states whether the vertex is bugged or clean, and thus whether or not the \faulty\ belongs to its ancestors or not. Once we find a bugged vertex whose parents are all clean, it is the \faulty.

The aim of the Regression Search Problem is to design a strategy for finding the \faulty\ in a minimal number of queries.

Formally, a \emph{strategy} (see for example \cite{cjlm}) for a DAG $D$  is a binary tree $S$ where the nodes are labelled by the vertices of $D$.
Inner nodes of $S$ represent queries. The root of $S$ is the first performed query. If the queried vertex is bugged, then the following strategy is given by the left subtree. If it is clean, the strategy continues on the right subtree.
Whenever the subtree is reduced to a leaf, a single candidate remains. The label of the leaf gives the only possible \faulty.

For example, Figure~\ref{fig:strategy_example} shows a strategy tree for a directed path of size $5$, where the identified bugged vertex is the last one. Suppose that the \faulty\ is $\vertex 4$. In this strategy, we first query $\vertex 2$. Since it is clean, we query next $\vertex 4$, which appears to be bugged. We finally query $\vertex 3$: since it is clean, we infer that the \faulty\ is $\vertex 4$. We have found the \faulty\ with $3$ queries. Remark that if the \faulty\  was $\vertex 1$, $\vertex 2$ or $\vertex 5$, the strategy would use only $2$ queries.

\begin{figure}[ht]
   \centering
   \begin{subfigure}{.45\textwidth}
      \centering
      \hfill
      \scalebox{1}{
         \begin{tikzpicture}
    [scale=.6,auto=left,every node/.style={circle,draw,scale=0.8}]

    \node (1) at (1, 1) {1};
    \node (2) at (2.5, 1) {2};
    \node (3) at (4, 1) {3};
    \node (4) at (5.5, 1) {4};
    \node[forbidden sign] (5) at (7, 1) {5};

    \path[->] 
        (1) edge (2)
        (2) edge (3)
        (3) edge (4)
        (4) edge (5)
    ;

\end{tikzpicture}
      }
   \end{subfigure}
   \hfill
   \begin{subfigure}{.45\textwidth}
      \centering
      \scalebox{1}{
         \begin{forest} 
    for tree={inner sep=2pt,l=10pt,l sep=5pt,circle,draw}
    [2
        [1
            [1,draw=none]
            [2,draw=none]
        ]
        [4
            [3
                [3,draw=none]
                [4,draw=none]
            ]
            [5,draw=none]
        ]
    ]
\end{forest}
      }
   \end{subfigure}
   \caption{\emph{Left.} A directed path on $5$ vertices. \emph{Right.} A possible strategy for the Regression Search Problem on the path on $5$ vertices.}
   \label{fig:strategy_example}
\end{figure}
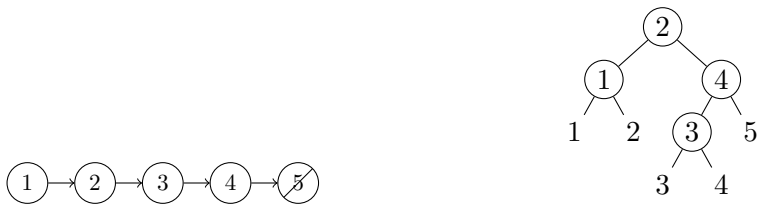

For a given strategy, the \emph{number of queries in the worst-case scenario} corresponds to the height of the tree. In the above example,
this number is 3, occurring when the \faulty\ is $\vertex 3$ or $\vertex 4$.

The Regression Search Problem is formally defined as follows.

\begin{definition} \textbf{Regression Search Problem.}  \\
   \textbf{Input.} A DAG $D$, a marked vertex $\vertex b$ known to be bugged, and an integer $k$.

   \noindent\textbf{Output.} Whether there is a strategy that finds the faulty commit in at most $k$ queries in the worst-case scenario.
   \label{def:regression_search_problem}
\end{definition}

Since the \faulty\ is necessarily an ancestor of $\vertex b$, it is convenient to directly study the induced subgraph on $\vertex b$'s ancestors. In this case, $\vertex b$ is a \emph{sink} (i.e. a vertex with no outgoing edge) accessible from all vertices in the DAG. Thus, when the bugged vertex is not specified, it is assumed to be the only sink of the DAG.

In the following, \emph{optimal strategies} are strategies that use the least number of queries in the worst case scenario.
For example, if the input DAG is a directed path of size $n$, an optimal strategy uses $\lceil \log_2(n) \rceil$ queries in the worst-case scenario. Indeed, a simple binary search enables to remove half of the vertices at each query.

\begin{figure}[ht]
   \centering
   \scalebox{1}{
      \begin{tikzpicture}
    [scale=.6,auto=left,every node/.style={circle,draw,scale=0.8}]

    \node (1) at (0,3) {1};
    \node (2) at (0,2) {2};
    \node (3) at (0,1) {3};
    \node (4) at (0,0) {4};
    \node (5) at (0,-1) {5};
    \node[forbidden sign] (6) at (2,1) {6};

    \path[->] 
        (1) edge (6)
        (2) edge (6)
        (3) edge (6)
        (4) edge (6)
        (5) edge (6)
    ;

\end{tikzpicture}
   }
   \caption{An octopus of size $6$.}
   \label{fig:octopus_graph}
\end{figure}
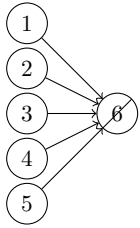

A second interesting example is what we refer to as an \emph{octopus}. In this digraph, there is a single sink and all other vertices are parent of the sink (see Figure~\ref{fig:octopus_graph}).
When the \faulty\ is the sink, we must query all other vertices to make sure that the sink is faulty, regardless of the strategy. Thus, any optimal strategy uses $n-1$ queries in the worst-case scenario.

These two examples actually constitute extreme cases for the Regression Search Problem, as shown by the following proposition.

\begin{proposition}
   For any DAG $D$ where the marked bugged vertex has $n$ ancestors, an optimal strategy that finds the \faulty\ uses at least $\lceil \log_{2}(n)\rceil$ queries, and at most $n-1$ queries.
   \label{prop:bounding_nb_queries}
\end{proposition}

\begin{proof}
   Remember that a strategy is a binary tree with at least $n$ leaves, and the number of queries in the worst-case scenario corresponds to the height of the tree. But the height of such a binary tree is necessarily at least $\lceil \log_{2}(n)\rceil$, which proves the lower bound.

   As for the upper bound, it is sufficient to query the $n-1$ ancestors of the marked bugged vertex to identify the \faulty.
\end{proof}

From a complexity point of view, the Regression Search Problem is hard: Carmo, Donadelli, Kohayakawa and Laber proved in \cite{CDKY:searching-in-posets} that the Regression Search Problem is NP-complete.\footnote{In reality, the problem they studied has an extra restriction: a query cannot be performed on a vertex which was eliminated from the set of candidates for the \faulty\ (which occurs for example when an ancestor is known to be bugged).
However, the gadget they used in the proof of NP-completeness also works for our problem where we do not necessarily forbid such queries.}
We also provide in Section \ref{sec:np_complete} an alternative proof of its \mbox{NP-completeness} (see Corollary~\ref{cor:rsp_crsp_crspbin_npcomplete}).

\subsection{Description of \gb}

As said in the introduction, some VCS provide a tool for the Regression Search Problem. The most known tool is \gb, but an equivalent exists in \texttt{mercurial} (\texttt{hg bisect}~\cite{hg-bisect}).

The algorithm \gb\ is a greedy algorithm based on the classical binary search.
At each step, it keeps only the subgraph where the \faulty\ lies and queries the vertex that split the digraph in the most balanced way.

To be more precise, let us define the notion of \emph{score}.

\begin{definition}[Score]
   Given a DAG $D$ with $n$ vertices, the score of a vertex $\vertex x$ of $D$ is  designated by $\score_D(\vertex x)$ or just $\score(\vertex x)$ if there is no ambiguity. Its value is
   \[\min(\ancestors x,n-\ancestors x),\] where $\ancestors x$ is the number of ancestors of $\vertex x$ (recall that $\vertex x$ is an ancestor of itself).
   \label{def:score}
\end{definition}

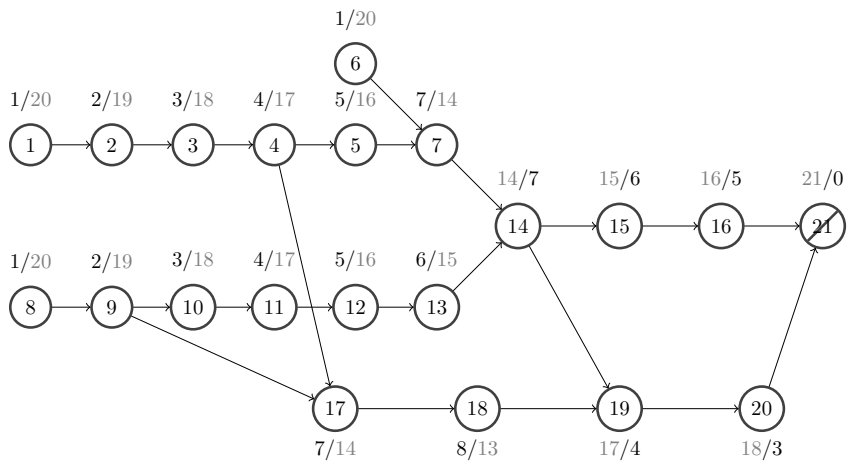
\begin{figure}[ht]
   \centering
   \scalebox{0.9}{
      \begin{tikzpicture}
	[
		scale=.6,
		minimum size=0.75cm,
		every node/.style={circle,draw,scale=0.8,text=black,line width=1.1pt},
		buggedNode/.style={text opacity=1,opacity=0.1,draw opacity=1,fill=ibmMagenta, color=ibmMagenta,text=black, line width=1.2pt},
		normalNode/.style={text opacity=1,opacity=0,draw opacity=1,color=darkgray,text=black, line width=1.1pt},
	]

	\tikzset{
        halfcross/.pic = {
        \draw[rotate = 45,color=darkgray] (-#1,0) -- (#1,0);
        }
    }

    \draw (19.5, 4) pic[black,opacity=1] {halfcross=11pt};

	\node[normalNode,label={[label distance=-0.2cm]90:1/\textcolor{gray}{20}}] (1) at (0, 6) {$1$};
	\node[normalNode,label={[label distance=-0.2cm]90:2/\textcolor{gray}{19}}] (2) at (2, 6) {$2$};
	\node[normalNode,label={[label distance=-0.2cm]90:1/\textcolor{gray}{20}}] (3) at (8,8) {$6$};
	\node[normalNode,label={[label distance=-0.2cm]90:3/\textcolor{gray}{18}}] (4) at (4,6) {$3$};
	\node[normalNode,label={[label distance=-0.2cm]90:4/\textcolor{gray}{17}}] (5) at (6,6) {$4$};
	\node[normalNode,label={[label distance=-0.2cm]90:5/\textcolor{gray}{16}}] (6) at (8,6) {$5$};
	\node[normalNode,label={[label distance=-0.2cm]90:7/\textcolor{gray}{14}}] (7) at (10, 6) {$7$};
	\node[normalNode,label={[label distance=-0.2cm]90:1/\textcolor{gray}{20}}] (8) at (0, 2) {$8$};
	\node[normalNode,label={[label distance=-0.2cm]90:2/\textcolor{gray}{19}}] (9) at (2,2) {$9$};
	\node[normalNode,label={[label distance=-0.2cm]90:3/\textcolor{gray}{18}}] (10) at (4, 2) {$10$};
	\node[normalNode,label={[label distance=-0.2cm]90:4/\textcolor{gray}{17}}] (11) at (6, 2) {$11$};
	\node[normalNode,label={[label distance=-0.2cm]90:5/\textcolor{gray}{16}}] (12) at (8, 2) {$12$};
	\node[normalNode,label={[label distance=-0.2cm]90:6/\textcolor{gray}{15}}] (13) at (10, 2) {$13$};
	\node[normalNode,label={[label distance=-0.2cm]90:\textcolor{gray}{14}/7}] (14) at (12, 4) {$14$};
	\node[normalNode,label={[label distance=-0.2cm]90:\textcolor{gray}{15}/6}] (15) at (14.5, 4) {$15$};
	\node[normalNode,label={[label distance=-0.2cm]90:\textcolor{gray}{16}/5}] (16) at (17, 4) {$16$};
	\node[normalNode,label={[label distance=-0.2cm]90:\textcolor{gray}{21}/0}] (17) at (19.5, 4) {$21$};
	\node[normalNode,label={[label distance=-0.3cm]-90:7/\textcolor{gray}{14}}] (18) at (7.5, -0.5) {$17$};
	\node[normalNode,label={[label distance=-0.3cm]-90:8/\textcolor{gray}{13}}] (19) at (11, -0.5) {$18$};
	\node[normalNode,label={[label distance=-0.3cm]-90:\textcolor{gray}{17}/4}] (20) at (14.5, -0.5) {$19$};
	\node[normalNode,label={[label distance=-0.3cm]-90:\textcolor{gray}{18}/3}] (21) at (18, -0.5) {$20$};

	\path[->,color=black]
		(1) edge (2)
		(2) edge (4)
		(4) edge (5)
		(5) edge (6)
		(6) edge (7)
		(3) edge (7)
		(7) edge (14)
		
		(8) edge (9)
		(9) edge (10)
		(10) edge (11)
		(11) edge (12)
		(12) edge (13)
		(13) edge (14)
		(14) edge (15)
		(15) edge (16)
		(16) edge (17)
		
		(18) edge (19)
		(19) edge (20)
		(20) edge (21)
		(21) edge (17)
		
		(5) edge (18)
		(9) edge (18)
		(14) edge (20)
	;

\end{tikzpicture}
   }
   \caption{The notation $a/b$ along each vertex indicates that $a$ is the number of ancestors of the vertex, and $b$ is the number of non-ancestors. The score (see Definition~\ref{def:score}) is displayed in black.}
   \label{fig:score}
\end{figure}

For example, let us refer to Figure~\ref{fig:score}: vertex $\vertex 5$ has $5$ ancestors ($\vertex 1$, $\vertex 2$, $\vertex 3$, $\vertex 4$ and $\vertex 5$). So $\score(\vertex 5)=\min(5,21-5) = 5$.

If vertex $\vertex x$ is queried and appears to be bugged, then there remain $\ancestors x$ candidates for the \faulty: the ancestors of $\vertex x$. If the query of $\vertex x$ reveals on the contrary that it is clean, then the number of candidates for the \faulty\ is $n-\ancestors x$, which is the number of non-ancestors. This is why the score of $\vertex x$ can be interpreted as the least number of vertices to be eliminated from the set of possible candidates for the \faulty, when $\vertex x$ is queried. For a DAG, each vertex has a score and the maximum score is the score with the maximum value among all.

A detailed description of \gb\ is given by Algorithm~\ref{algo:gitbisect}.

\begin{algorithm}[ht]
    \caption{\gb}
    \textbf{Input.} A DAG $D$ and a bugged vertex $\vertex b$. \\
    \textbf{Output.} The \faulty\ of $D$. \\
    \textbf{Steps:}
    \begin{enumerate}
        \item Remove from $D$ all non-ancestors of $\vertex b$.
        \item If $D$ has only one vertex, return this vertex. \label{item:debut-gb}
        \item Compute the score for each vertex of $D$.
        \item Query the vertex with the maximum score. If there are several vertices which have the maximum score, select any one then query it.
        \item If the queried vertex is bugged, remove from $D$ all non-ancestors of the queried vertex. Otherwise, remove from $D$ all ancestors of the queried vertex.
        \item Go to Step~\ref{item:debut-gb}.
    \end{enumerate}
    \label{algo:gitbisect}
\end{algorithm}

As an example of an execution, consider the DAG from Figure~\ref{fig:score}. Vertex $\vertex{18}$ has the maximum score ($\score(\vertex{18})=8$) so constitutes the first vertex to be queried. If we assume that the \faulty\ is $\vertex{5}$, then the query reveals that $\vertex{18}$ is clean. So all ancestors of $\vertex{18}$ are removed (that are $\vertex 1, \vertex 2, \vertex 3, \vertex 4, \vertex 8, \vertex 9, \vertex {17}, \vertex{18}$). Vertex $\vertex{14}$ is then queried because it has the new maximum score $5$, and so on.

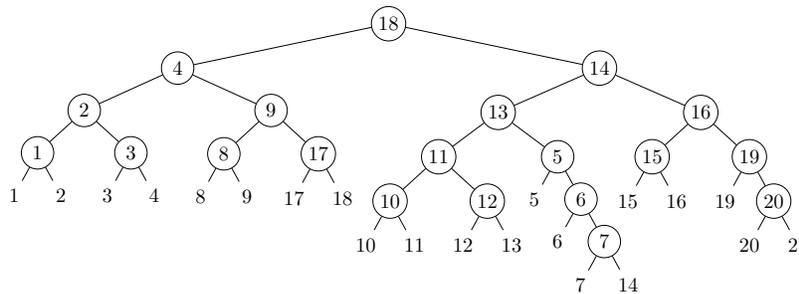
\begin{figure}[ht]
  \centering
  \scalebox{.7}{
     \begin{forest} 
    for tree={inner sep=2pt,l=10pt,l sep=5pt,circle,draw,minimum size=1.6em}
    [18
    [4
        [2
            [1
                [1,draw=none]
                [2,draw=none]
            ]
            [3
                [3,draw=none]
                [4,draw=none]
            ]
        ]
        [9
            [8
                [8,draw=none]
                [9,draw=none]
            ]
            [17
                [17,draw=none]
                [18,draw=none]
            ]
        ]
    ]
    [14
        [13
            [11
                [10
                    [10,draw=none]
                    [11,draw=none]
                ]
                [12
                    [12,draw=none]
                    [13,draw=none]
                ]
            ]
            [5
                [5,draw=none]
                [6
                    [6,draw=none]
                    [7
                        [7,draw=none]
                        [14,draw=none]
                    ]
                ]
            ]
        ]
        [16
            [15
                [15,draw=none]
                [16,draw=none]
            ]
            [19
                [19,draw=none]
                [20
                    [20,draw=none]
                    [21,draw=none]
                ]
            ]
        ]
    ]
    ]
\end{forest}
  }
  \caption{The \gb\ strategy corresponding to the graph of Figure \ref{fig:score}. In case of score equality, the convention we choose consists in querying the vertex with the smallest label.}
  \label{fig:complete_strategy_example}
\end{figure}

The whole \gb\ strategy tree is shown in Figure~\ref{fig:complete_strategy_example}.

Notice that for this DAG, \gbalgo\ uses $6$ queries in the worst-case scenario, which is not optimal for this \emph{well-chosen} example as we are going to see later.

The greedy idea behind \gb\ (choosing the query which partitions the commits as evenly as possible) is quite widespread in the literature. For example, it was used to find a $(\log(n)+1)$-approximation for the Decision Tree Problem \cite{AH:approximating-DT}, in particular within the framework of geometric models \cite{AMMRS:geometric}.

\section{Worst-case number of queries}
\label{sec:general}

This section addresses the complexity analysis of \gb\ in the worst-case scenario.

In Sections \ref{sec:general}, \ref{sec:binary} and \ref{sec:golden}, we consider algorithms that prune all non-ancestors of the marked vertex $\vertex b$.
Therefore, all results are stated for DAGs that have one sink, which is the marked vertex, and for which the number of vertices $n$ is also the number of candidates for the \faulty.

\subsection{The comb construction}

We describe in this subsection a way to enhance any DAG in such a way the Regression Search Problem can always be solved in a logarithmic number of queries.

\begin{definition}[Comb addition]
   Let $D$ be a Directed Acyclic Graph with $n$ vertices. Let $\vertex{v_1}<\vertex{v_2}<\ldots<\vertex{v_n}$ be a topological ordering of $D$, that is a linear ordering of the vertices such that if $\vertex{v_i} \to \vertex{v_j}$ is an arc, then $\vertex{v_i}<\vertex{v_j}$.

   We say that we add a comb to $D$ if we add to $D$:
   \begin{itemize}
      \item $n$ new vertices $\vertex{u_1},\ldots,\vertex{u_n}$;
      \item the arcs $\vertex{v_i} \to \vertex{u_i}$ for $i \in \{1,\ldots,n\}$;
      \item the arcs $\vertex{u_i} \to \vertex{u_{i+1}}$ for $i \in \{1,\ldots,n-1\}$.
   \end{itemize}
   The resulting graph is designated by $comb(D)$. The new identified bugged vertex of $comb(D)$ is $\vertex{u_n}$.
   \label{def:comb}
\end{definition}

An example of comb addition is shown by Figure~\ref{fig:comb_addition}.

The comb addition depends on the initial topological ordering, but the latter will not have any impact on the following results. This is why we take the liberty of writing $comb(D)$ without any mention of the topological ordering.

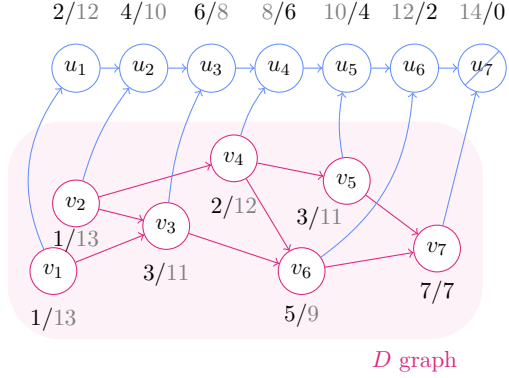
\begin{figure}[ht]
   \centering
   \begin{tikzpicture}
	[
		scale=.6,
		every node/.style={circle,draw,scale=0.8},
		uNode/.style={color=ibm_blue,text=black},
		vNode/.style={color=ibm_magenta,text=black,fill=white},
		buggedNode/.style={forbidden sign}
	]

	\draw[fill,opacity=0.06,color=ibm_magenta,rounded corners=20pt] (-0.5,-1) rectangle (10,3.8); 
	\node[rectangle,draw=none,color=ibm_magenta] at (8.5,-1.5) {$D$ graph};

	\node[vNode,label={[label distance=-0.2cm]-90:1/\textcolor{gray}{13}}] (v1) at (0.5, 0.5) {$v_1$};
	\node[vNode,label={[label distance=-0.4cm]-90:1/\textcolor{gray}{13}}] (v2) at (1, 2) {$v_2$};
	\node[vNode,label={[label distance=-0.2cm]-90:3/\textcolor{gray}{11}}] (v3) at (3.0, 1.5) {$v_3$};
	\node[vNode,label={[label distance=-0.2cm]-90:2/\textcolor{gray}{12}}] (v4) at (4.5, 3.0) {$v_4$};
	\node[vNode,label={[label distance=-0.2cm]-95:3/\textcolor{gray}{11}}] (v5) at (7.0, 2.5) {$v_5$};
	\node[vNode,label={[label distance=-0.2cm]-90:5/\textcolor{gray}{9}}] (v6) at (6.0, 0.5) {$v_6$};
	\node[vNode,label={[label distance=-0.2cm]-90:7/7}] (v7) at (9.0, 1.0) {$v_7$};
	
	\path[->,color=ibm_magenta] 
		(v1) edge (v3)
		(v2) edge (v3)
		(v3) edge (v6)
		(v2) edge (v4)
		(v4) edge (v6)
		(v4) edge (v5)
		(v5) edge (v7)
		(v6) edge (v7)
	;
	
	\node[uNode,label={[label distance=-0.1cm]90:2/\textcolor{gray}{12}}] (u1) at (1, 5.0) {$u_1$};
	\node[uNode,label={[label distance=-0.1cm]90:4/\textcolor{gray}{10}}] (u2) at (2.5, 5.0) {$u_2$};
	\node[uNode,label={[label distance=-0.02cm]90:6/\textcolor{gray}{8}}] (u3) at (4, 5.0) {$u_3$};
	\node[uNode,label={[label distance=-0.02cm]90:\textcolor{gray}{8}/6}] (u4) at (5.5, 5.0) {$u_4$};
	\node[uNode,label={[label distance=-0.1cm]90:\textcolor{gray}{10}/4}] (u5) at (7, 5.0) {$u_5$};
	\node[uNode,label={[label distance=-0.1cm]90:\textcolor{gray}{12}/2}] (u6) at (8.5, 5.0) {$u_6$};
	\node[buggedNode,uNode,label={[label distance=-0.1cm]90:\textcolor{gray}{14}/0}] (u7) at (10, 5.0) {$u_7$};

	\path[->,color=ibm_blue] 
		(u1) edge (u2)
		(u2) edge (u3)
		(u3) edge (u4)
		(u4) edge (u5)
		(u5) edge (u6)
		(u6) edge (u7)
	;

	\path[->, color=ibm_blue] 
		(v1) edge[bend left=30] (u1)
		(v2) edge[bend left=10] (u2)
		(v3) edge[bend left=10] (u3)
		(v4) edge[bend left=10] (u4)
		(v5) edge[bend left=10] (u5)
		(v6) edge[bend right=25] (u6)
		(v7) edge (u7)
	;

\end{tikzpicture}
   \caption{Illustration of the comb addition. The initial digraph is highlighted in pink.}
   \label{fig:comb_addition}
\end{figure}

\begin{theorem}
   Let $D$ be a Directed Acyclic Graph with $n$ vertices and such that the number of queries used by \gbalgo\ is $x$. If we add a comb to $D$, then the resulting DAG $comb(D)$ is such that:
   \begin{itemize}
      \item the optimal strategy uses only $\lceil \log_2(2n) \rceil$ queries;
      \item when $n$ is odd, \gbalgo\ uses $x+1$ queries.
   \end{itemize}
   \label{theo:comb}
\end{theorem}

\noindent\textit{\textbf{Proof idea}}
On one hand, the optimal strategy for $comb(D)$ can be naturally achieved with a binary search on the $\vertex{u_i}$ vertices.
On the other hand, with the assumption that $n$ is odd, \gb\ will necessarily query $\vertex{v_n}$ first since its score is $n$, and the scores of the $\vertex{u_i}$ vertices are all even. This explains why \gbalgo\ uses $x+1$ queries.

\begin{proof}[Detailed proof]
   We keep the same notation as Definition~\ref{def:comb}. For a DAG $D$ and a subset of vertices $X\subseteq V$, the \emph{induced subgraph} of $D$ on $X$, denoted by $D[X]$, is the digraph with vertex set $X$, and with an arc from vertex $\vertex u$ to vertex $\vertex v$ if and only if the corresponding arc is in $D$.

   \begin{claim}
      For all $i$, $\vertex{u_i}$ has $2i$ ancestors, which are all the vertices $\vertex{u_j}$ and $\vertex{v_j}$ with $j\le i$. The ancestors of $\vertex{v_i}$ do not change. \label{claim:comb_ancestors}
   \end{claim}

   Observe first that no $\vertex{v_i}$ is the head of an arc added in $comb(D)$. Inductively, we infer that the ancestors of $\vertex{v_i}$ do not change.

   As for $\vertex{u_i}$, we prove the claim by induction. Indeed, vertex $\vertex{u_i}$ has two parents which are $\vertex{u_{i-1}}$ and $\vertex{v_i}$. By induction hypothesis, we can see that all the vertices $\vertex{u_j}$ and $\vertex{v_j}$ with $j<i$ are ancestors of $\vertex{u_i}$ since they are the ancestors of $\vertex{u_{i-1}}$. Moreover all ancestors $\vertex{v_j}$ of $\vertex{v_i}$ satisfy $j \le i$ (by topological ordering). Consequently $\vertex{u_i}$ has $2i$ ancestors: itself, $\vertex{v_i}$ and all the ancestors of $\vertex{u_{i-1}}$.

   \begin{claim}
      The optimal number of queries is $\lceil \log_2(2n) \rceil$ for $comb(D)$.
   \end{claim}

	Let us prove this claim for every digraph $D$ by induction on the number $n$ of vertices of $D$.

	The case $n=1$ is obvious: if $D$ has only $1$ vertex, we query $\vertex{v_1}$ to know whether $\vertex{u_1}$ or $\vertex{v_1}$ is the \faulty. The number of queries is then $\lceil \log_2(2 \times 1) \rceil = 1$.

	Now fix $n>1$ and let us assume that the claim holds for every digraph $D$ of size smaller than $n$. We choose as the first query the vertex $\vertex{u_i}$ where $i = \lceil \frac n 2 \rceil$.

   Depending on whether $\vertex{u_i}$ is bugged or clean, the digraph after this query is
   either $comb(D)[\vertex{u_1},\dots,\vertex{u_i},\vertex{v_1},\dots,\vertex{v_i}]$ or $comb(D)[\vertex{u_{i+1}},\dots,\vertex{u_n},\vertex{v_{i+1}},\dots,\vertex{v_n}]$.

   Notice that in any case, the resulting digraph is of the form $comb(D')$. Indeed, we just have to choose $D' := D[\vertex{v_1},\dots,\vertex{v_i}]$ or $D' := D[\vertex{v_{i+1}},\dots,\vertex{v_n}]$, and keep the same topological ordering.

   Now we can use the induction hypothesis on $comb(D')$, which has at most $\left\lceil \frac n 2 \right\rceil$ vertices: we can find a strategy in at most $\lceil \log_2(2 \left \lceil \frac n 2 \right \rceil ) \rceil$ queries to find the \faulty\  in $comb(D')$.

   The overall number of queries for $comb(D)$ with this strategy is then at most $1 + \left\lceil \log_2\left( 2 \lceil \frac n 2 \rceil  \right) \right\rceil$, which is equal to $\lceil \log_2(2n) \rceil$ whenever $n \geq 1$. By Proposition~\ref{prop:bounding_nb_queries}, a strategy with this number of queries must be optimal.

   \begin{claim}
      If $n$ is odd, \gbalgo\ necessarily uses $x+1$ queries.
   \end{claim}

   By Claim~\ref{claim:comb_ancestors}, $\vertex{v_n}$ has $n$ ancestors, and digraph $comb(D)$ has $2n$ vertices. So $\score(\vertex{v_n})=n$ (hence maximal).

   Vertex $\vertex{v_n}$ is the only one to have a maximal score. Indeed, on the one hand, any vertex of the form $\vertex{v_i}$ with $i < n$ has fewer than $n$ ancestors. On the other hand, $\vertex{u_i}$ having $2i$ ancestors, its score must be even, and therefore cannot be maximal if $n$ is odd.

   Thus \gbalgo\ is going to choose $\vertex{v_n}$ as first query. If this vertex turns out to be clean, it remains a directed path of length $n$, inducing $\lceil \log_2(n) \rceil$ \gb\ queries. If $\vertex{v_n}$ is bugged, then the resulting graph is $D$, for which the worst-case number of \gb\ queries is $x$. Therefore, since $x \geq \lceil \log_2(n) \rceil$ by Proposition~\ref{prop:bounding_nb_queries}, the number of \gb\ queries for $comb(D)$ in the worst-case scenario is $x+1$.
\end{proof}

If the initial number of vertices $n$ is even, there is no guarantee that \gb\  will perform $x+1$ queries on $comb(D)$ -- it depends on whether the first queried vertex is $\vertex{v_n}$ or $\vertex{u_{n/2}}$.

However a referee rightly mentioned that the odd hypothesis could be (almost) removed by tweaking the comb construction whenever $n$ is even. Indeed, by deleting the edge from $\vertex{ v_{n/2} }$ to $\vertex{ u_{n/2}}$, \gb\ is forced to use $x+1$ queries in the worst-case scenario, while
the following strategy uses $\lceil \log_2(2n) + 1 \rceil$ queries : run a binary search on the path formed by vertices $\vertex{u_1}, \ldots, \vertex{u_n}$, then query all remaining parents of the identified vertex $\vertex{u_i}$ (that is possibly zero, two or one parents depending on whether the identified vertex is respectively $\vertex{u_{n/2}}$, $\vertex{u_{n/2+1}}$ or any other $\vertex{u_i}$) .

\subsection{A pathological example for \gb}

The following corollary shows the existence of digraphs for which \gbalgo\ totally fails. The optimal number of queries is linear, while \gbalgo\ effectively uses an exponential number of queries.

\begin{theorem}
   For any integer $k > 2$, there exists a DAG such that the optimal number of queries is $k$, while \gbalgo\ uses $2^{k-1}-1$ queries in the worst-case scenario.
   \label{theo:pathological}
\end{theorem}

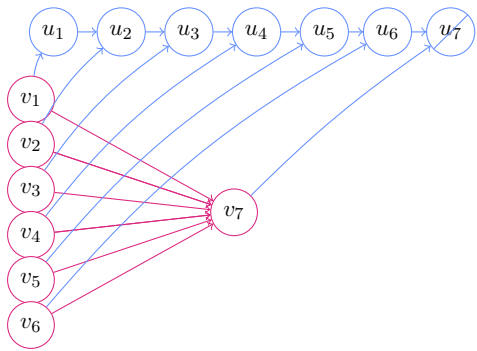
\begin{figure}[ht]
   \centering
   \begin{tikzpicture}
	[
		scale=.6,
		every node/.style={circle,draw,scale=0.8},
		uNode/.style={color=ibm_blue,text=black},
		vNode/.style={color=ibm_magenta,text=black,fill=white},
		buggedNode/.style={forbidden sign}
	]

	\node[vNode] (v1) at (0,0) {$v_1$};
	\node[vNode] (v2) at (0,-1) {$v_2$};
	\node[vNode] (v3) at (0,-2) {$v_3$};
	\node[vNode] (v4) at (0,-3) {$v_4$};
	\node[vNode] (v5) at (0,-4) {$v_5$};
	\node[vNode] (v6) at (0,-5) {$v_6$};
	\node[vNode] (v7) at (4.5,-2.5) {$v_7$};

	\path[->,color=ibm_magenta] 
		(v1) edge (v7)
		(v2) edge (v7)
		(v3) edge (v7)
		(v2) edge (v7)
		(v4) edge (v7)
		(v4) edge (v7)
		(v5) edge (v7)
		(v6) edge (v7)
	;

	\node[uNode] (u1) at (0.5, 1.5) {$u_1$};
	\node[uNode] (u2) at (2, 1.5) {$u_2$};
	\node[uNode] (u3) at (3.5, 1.5) {$u_3$};
	\node[uNode] (u4) at (5, 1.5) {$u_4$};
	\node[uNode] (u5) at (6.5, 1.5) {$u_5$};
	\node[uNode] (u6) at (7.9, 1.5) {$u_6$};
	\node[buggedNode,uNode] (u7) at (9.3, 1.5) {$u_7$};

	\path[->,color=ibm_blue] 
		(u1) edge (u2)
		(u2) edge (u3)
		(u3) edge (u4)
		(u4) edge (u5)
		(u5) edge (u6)
		(u6) edge (u7)
	;

	\path[->,color=ibm_blue] 
		(v1) edge[bend left=11] (u1)
		(v2) edge[bend left=10] (u2)
		(v3) edge[bend left=10] (u3)
		(v4) edge[bend left=10] (u4)
		(v5) edge[bend left=10] (u5)
		(v6) edge[bend left=10] (u6)
		(v7) edge[bend left=5] (u7)
	;

\end{tikzpicture}
   \caption{$Comb(D)$ graph where $D$ is an octopus of size $7$.}
   \label{fig:octopus_with_comb}
\end{figure}

\begin{proof}
   Choose $D$ as an octopus with $2^{k-1} - 1$ vertices. The number of \gb\ queries in $D$ is $2^{k-1} - 2$ in the worst-case scenario.
   The wanted digraph is then $comb(D)$ (see Figure~\ref{fig:octopus_with_comb} for an illustration). Indeed, by Theorem~\ref{theo:comb}, \gbalgo\ uses $2^{k-1} - 1$ \gb\ queries to find the \faulty\ in $comb(D)$, while an optimal strategy uses $\left\lceil \log_2\left( 2^k - 2\right)\right\rceil=k$ queries.
\end{proof}

This also shows that \gbalgo\ is not a $C$-approximation algorithm for the Regression Search Problem, for any constant $C$.

\section{Approximation ratio for binary DAGs}
\label{sec:binary}

\subsection{Results}

The pathological input for \gbalgo\ has a very particular shape (see Figure~\ref{fig:octopus_with_comb}): it involves a vertex with a gigantic indegree. However, in the context of VCS, this structure is quite rare. It means that many branches have been merged at the same time (the famous \textit{octopus merge}). Such an operation is strongly discouraged, in addition to the fact that we just showed that \gb\ becomes inefficient in this situation.

This motivates to define a new family of DAGs, closer to reality:

\begin{definition}[Binary digraph]
   A digraph is \emph{binary} if each vertex has indegree (that is, the number of ingoing edges) at most $2$.
\end{definition}

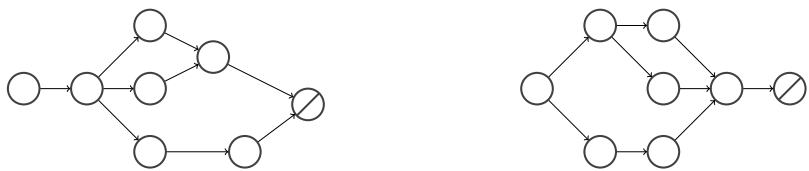
\begin{figure}[ht]
   \centering
   \begin{subfigure}[t]{.45\textwidth}
      \centering
      \scalebox{0.7}{
         \begin{tikzpicture}
	[
		scale=.6,
		every node/.style={circle,draw,scale=0.8},
		minimum size=0.75cm,
		normalNode/.style={text opacity=0,opacity=0,draw opacity=1,color=darkgray,text=black, line width=1.1pt}
	]

	\node[normalNode] (1) at (4, 6) {};
	\node[normalNode] (2) at (0, 4) {};
	\node[normalNode] (3) at (2, 4) {};
	\node[normalNode] (4) at (4, 4) {};
	\node[normalNode] (5) at (6, 5) {};
	\node[normalNode,forbidden sign] (6) at (9, 3.5) {};
	\node[normalNode] (7) at (4, 2) {};
	\node[normalNode] (8) at (7, 2) {};

	\path[->,color=black]
		(1) edge (5)
		(2) edge (3)
		(4) edge (5)
		(5) edge (6)
		(3) edge (1)
		(3) edge (4)
		(3) edge (7)
		(7) edge (8)
		(8) edge (6)
	;

\end{tikzpicture}
      }
   \end{subfigure}
   \hspace{0.3cm}
   \begin{subfigure}[t]{.45\textwidth}
      \centering
      \scalebox{0.7}{
         \begin{tikzpicture}
	[
		scale=.6,
		every node/.style={circle,draw,scale=0.8},
		minimum size=0.75cm,
		normalNode/.style={text opacity=0,opacity=0,draw opacity=1,color=darkgray,text=black, line width=1.1pt}
	]

	\node[normalNode] (1) at (2, 6) {};
	\node[normalNode] (2) at (4, 6) {};
	\node[normalNode] (3) at (0, 4) {};
	\node[normalNode] (4) at (4, 4) {};
	\node[normalNode] (5) at (6, 4) {};
	\node[normalNode,forbidden sign] (6) at (8, 4) {};
	\node[normalNode] (7) at (2, 2) {};
	\node[normalNode] (8) at (4, 2) {};

	\path[->,color=black]
		(1) edge (2)
		(2) edge (5)
		(4) edge (5)
		(8) edge (5)
		(5) edge (6)
		(3) edge (1)
		(3) edge (7)
		(7) edge (8)
		(1) edge (4)
	;

\end{tikzpicture}
      }
   \end{subfigure}
   \caption{\emph{Left.} A binary DAG. \emph{Right.} A non-binary DAG.}
   \label{fig:binary_dag_and_non_binary_dag}
\end{figure}

Figure~\ref{fig:binary_dag_and_non_binary_dag} illustrates this definition.  If we restrict the DAG to be binary, \gb\ proves to be efficient.

\begin{theorem}
   On any binary DAG with $n$ vertices, the number of queries of \gbalgo\ is at most $\frac{\log_2(n)}{\log_2(3/2)}$.
   \label{theo:bound_gb_binary_dag}
\end{theorem}

\begin{corollary}
   The algorithm \gb\ is a  $\frac{1}{\log_2(3/2)} \approx 1.71$-approximation algorithm on binary DAGs.
\end{corollary}

\subsection{Bounding the number of queries}

The key ingredient of the proof lies in the next lemma, which exhibits a core property of binary DAGs. It states that if the DAG is binary, there must be a vertex with a ``good'' score, i.e., that removes at least approximately one third of the remaining vertices at each query. The overall number of queries is then equal to $\log_{3/2} (n)$.

\begin{lemma}
   In every binary DAG with $n$ vertices, there exists a vertex $\vertex{v}$ whose number of ancestors, $|v|$, satisfies the double inequality $\frac{n}{3} \leq \ancestors v \le \frac{2n+1}{3}$.
   \label{lem:good_score_vertex}
\end{lemma}

The reader can look at Figure~\ref{fig:good_score_vertex} for an illustrative example.

\begin{figure}[ht]
   \centering
   \scalebox{0.8}{
      \begin{tikzpicture}
	[
		scale=.6,
		every node/.style={circle,draw,scale=0.8},
		vNode/.style={color=black,text=black},
	]

	\node[vNode,label={[label distance=-0.2cm]90:1/\textcolor{gray}{7}},text opacity=0] (v1) at (0, 0) {$v$};
	\node[vNode,label={[label distance=-0.3cm]-90:1/\textcolor{gray}{7}},text opacity=0] (v2) at (0, -3.5) {$v$};
	\node[vNode,label={[label distance=-0.2cm]90:2/\textcolor{gray}{6}},text opacity=0] (v3) at (2, 0) {$v$};
	\node[vNode,label={[label distance=-0.2cm]-90:2/\textcolor{gray}{6}},text opacity=0] (v4) at (2, -2) {$v$};
	\node[vNode,label={[label distance=-0.2cm]90:\textcolor{gray}{5}/3},ultra thick,fill,opacity=0.08,text opacity=1,draw opacity=1] (v5) at (4.5, -1) {$\boldsymbol{v}$};
	\node[vNode,label={[label distance=-0.2cm]-90:2/\textcolor{gray}{6}},text opacity=0] (v6) at (4, -4) {$v$};
	\node[vNode,label={[label distance=-0.2cm]90:\textcolor{gray}{7}/1},text opacity=1] (v7) at (7, -2) {$u$};
	\node[vNode,label={[label distance=-0.2cm]90:\textcolor{gray}{8}/0},text opacity=0,forbidden sign] (v8) at (9, -2) {$v$};
	
	\path[->,color=darkgray] 
	(v1) edge (v3)
	(v2) edge (v4)
	(v3) edge (v5)
	(v4) edge (v5)
	(v5) edge (v7)
	(v7) edge (v8)
	(v2) edge (v6)
	(v6) edge (v7)
	;

\end{tikzpicture}
   }
   \caption{The highlighted vertex $\vertex{v}$ is the only one to have its number of ancestors in $[\frac{n}{3},\frac{2n+1}{3}]$, where $n=8$. }
   \label{fig:good_score_vertex}
\end{figure}
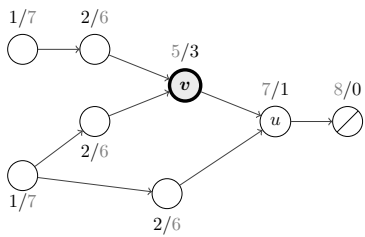

\begin{proof}\hspace{-.09cm}\footnote{The authors wish to thank the referee who suggested this more condensed proof.}
   If $n = 1$, the only vertex $\vertex{v}$ of the DAG has $\ancestors v = 1$, which satisfies the bound.
   Then, if $n \geq 2$, let $\vertex{u}$ be a vertex such that $\ancestors u \geq (2n+2)/3$, chosen so that $\ancestors u$ is as small as possible.
   Since $\ancestors u \geq 2$, the vertex $\vertex{u}$ has one or two parents.
   Let $\vertex{v}$ be the parent of $\vertex{u}$ with the most ancestors.

   Since $\ancestors v < \ancestors u$ and, by minimality of $\ancestors u$, we have $\ancestors v \leq (2n+1)/3$.
   Furthermore, at least half of the $\ancestors{u}-1$ strict ancestors of $\vertex{u}$ must be ancestors of $\vertex{v}$.
   It follows that $\ancestors{v} \ge \frac{(\ancestors{u}-1)}{2} \ge \frac{(2n-1)}{6} > \frac{(n-1)}{3}$, i.e., that $\ancestors{v} \geq \frac{n}{3}$.
\end{proof}

This lemma is sufficient to prove the logarithmic upper bound  for the number of \gb\ queries.

\begin{proof}[Proof of Theorem~\ref{theo:bound_gb_binary_dag}]
   Let $D$ be a DAG with $n$ vertices, and $D_k$ the digraph obtained from $D$ after $k$ \gb\ queries. Let $n_k$ be the number of vertices in $D_k$.
   After each query, \gbalgo\ chooses the vertex $\vertex v$ given by Lemma~\ref{lem:good_score_vertex} or a vertex with a better score. In any case, the score of the chosen vertex in $D_k$ is greater or equal than $\dfrac{n_k - 1} 3$.  This is why
   \begin{equation}
      n_{k+1} \leq \frac{2 n_{k}+1}{3}.
      \label{eq:rec_nk}
   \end{equation}
   We can then show by induction that
   \begin{equation}
      n_k \leq 1 +     \left( \frac 2 3 \right)^k (n-1).
      \label{eq:ineq_nk}
   \end{equation}

   We distinguish two cases from here.

   \textbf{1. Case $\boldsymbol{n \leq 8}$.} We can check the small cases by repetitively using Inequality~\eqref{eq:rec_nk} and keeping the integral part of the right member of the inequality (since we work with integers), thereby obtaining an upper bound $F(n)$ on the least integer $k$ for which $n_k = 1$. For example, if $n=n_0 = 5$, we see that $n_1 \leq \frac {11} 3$, hence $n_1 \leq 3$. Then $n_2 \leq \frac 7 3$ and so $n_2 \leq 2$, and finally $n_3 \leq 1$, which means that the number of queries for a DAG of size $5$ is at most $F(5)=3$. The first values of this upper bound are listed in Table~\ref{table:Fn}.   We remark that this is consistent with the $\log_{3/2}(n)$ bound of Theorem~\ref{theo:bound_gb_binary_dag}.

   \begin{table}[ht!]
      \noindent
      \resizebox{\textwidth}{!}{
      \begin{tabular}{|c|ccccccccccccc|}
         \hline
         $n$ & $1$ & $2$ & $3$ & $4$ & $5$ & $6$ & $7$ & $8$ & $9$ & $10$ & $11$ & $12$ & $13$ \\
         \hline  
         Upper bound $F(n)$ & $0$ & $1$ & $2$ & $3$ & $3$ & $4$ & $4$ & $4$ & $5$ & $5$ & $5$ & $5$ & $6$ \\ 
         \hline
         \shortstack{ \\ Approximation \\ for $\log_{\frac{3}{2}}(n)$  } & $0$ & $1.71$ & $2.71$ & $3.42$ & $3.97$ & $4.42$ & $4.80$ & $5.13$ & $5.42$ & $5.68$ & $5.91$ & $6.12$ & $6.33$ \\ 
         \hline
         \shortstack{ \\ Approximation \\ for $\log_{\phi}(n) + 1$  } & $1$ & $2.44$ & $3.28$ & $3.88$ & $4.34$ & $4.72$ & $5.04$ & $5.32$ & $5.56$ & $5.78$ & $5.98$ & $6.16$ & $6.33$ \\
         \hline
      \end{tabular}
   }
   \caption{Checking Theorem~\ref{theo:bound_gb_binary_dag} and  Theorem~\ref{theo:bound_golden_binary_dag} for small sizes}\label{table:Fn}
   \end{table}

   \textbf{2. Case $\boldsymbol{n \geq 9}$.} Let $x$ be the largest number of queries such that $n_x \geq 9$. This means that after $x+1$ queries, the DAG will have at most $8$ vertices and by Table~\ref{table:Fn}, we see that a maximum of $4$ extra queries can be required to find the \faulty\ from this point. Therefore, the number of \gb\  queries for $D$ is at most $x+5$. 
      
   Setting $k=x$ in Inequality~\eqref{eq:ineq_nk} shows that 
   $8 \leq  \left( \frac 2 3 \right)^x (n-1)$, hence $x + \log_{\frac 3 2}(8) \leq \log_{\frac 3 2}(n-1)$. Since $\log_{3/2}\left(8\right) \simeq 5.13$, we see that the number of \gb\ queries of $D$ is indeed bounded by $x+5 \leq \log_{\frac 3 2}(n-1)$.
\end{proof}

\subsection{Tight case}

The upper bound of Theorem~\ref{theo:bound_gb_binary_dag} is asymptotically sharp, as stated by the following proposition.

\begin{proposition}
   For any integer $k \ge 1$, there exists a \emph{binary} DAG $J_k$ such that
   \begin{itemize}
      \item the number of \gb\ queries on $comb(J_k)$ is $k + \left\lceil \log_2(k+1) \right\rceil +3$;
      \item an optimal strategy for $comb(J_k)$ uses at most $\log_2(\frac 3 2 )\,k + \log_2(3k+7) + 2$ queries.
   \end{itemize}
   (Remember that the $comb$ operation is described by Definition~\ref{def:comb}.)
   \label{prop:tight_case}
\end{proposition}

Figure \ref{fig:J33} shows what $J_k$ looks like for $k=3$. For this example, \gb\ uses 7 queries in the worst-case scenario (which occurs for example when $\vertex c$ is bugged).

\begin{proof}
   We first describe a family $J_k$ of graphs that fulfil the properties of
   Proposition~\ref{prop:tight_case}.

   We start by defining $J_k^0$, the \emph{backbone} of $J_k$. It is formed by taking three directed paths on $k+1$ vertices $\vertex{x_1} \rightarrow \vertex{x_2} \rightarrow \dots \rightarrow \vertex{x_{k+1}}$, $\vertex{y_1} \rightarrow  \dots \rightarrow \vertex{y_{k+1}}$ and $\vertex{z_0} \rightarrow \dots \rightarrow \vertex{z_{k}}$ and merging the three vertices $\vertex{x_{k+1}}$, $\vertex{y_{k+1}}$ and $\vertex{z_0}$ into a vertex $\vertex c$ (see Figure~\ref{fig:J30} for an example with $k=3$).

   We construct our final graph $J_k$ from its backbone through $k+1$ successive digraphs: $J_k^0, J_k^1, \dots, J_k^k$. For each $d$ starting from $1$ to $k$, let us define
   \[\ell_d = \left\lfloor \frac{n_{d-1}} 6 + 1 \right\rfloor \]
   where  $n_{d-1}$ stands for the number of vertices in $J_k^{d-1}$.
   Add a directed path on $\ell_d$ vertices towards each backbone vertex at distance $d$ from $\vertex c$, namely $\vertex{x_{k+1-d}}$, $\vertex{y_{k+1-d}}$, and $\vertex{z_{d}}$. Then, add edges from the new parents of $\vertex{x_{k+1-d}}$ and of $\vertex{y_{k+1-d}}$ to the first vertex of the path newly attached to $\vertex{z_d}$. Also, the new parent of $\vertex{z_d}$ is denoted by $\vertex{z'_d}$. The reader can refer to Figure~\ref{fig:generic_binary_dag_J} for a better understanding of the notation.

   \begin{figure}[ht]
      \centering
      \scalebox{0.7}{
         \begin{tikzpicture}
	[
		scale=.6,
		every node/.style={
			circle,
			draw,
			scale=0.8,
			text=black,
			line width=1.2pt,
			opacity=.2,
			text opacity=1,
			draw opacity=1
		},
		d0Node/.style={fill=ibm_blue,color=ibm_blue,text=black,line width=1.2pt},
		d1Node/.style={fill=ibm_purple,color=ibm_purple,text=black,line width=1.2pt},
		d2Node/.style={fill=ibm_magenta,color=ibm_magenta,text=black,line width=1.2pt},
		d3Node/.style={fill=ibm_orange,color=ibm_orange,text=black,line width=1.2pt},
		d4Node/.style={fill=ibm_yellow,color=ibm_yellow,text=black,line width=1.2pt},
		minimum size=0.75cm
	]

	% J0
	\node[color=gray,text=black] (xe) at (-10.5, 2) {$x_1$};
	\node[d4Node] (x0) at (-8.5, 2) {$x_{k-d}$};
	\node[d3Node] (x1) at (-6, 2) {$x_{k+1-d}$};
	\node[d2Node] (x2) at (-4, 2) {};
	\node[d1Node,text opacity=1] (x3) at (-2, 2) {$x_k$};

	\node[color=gray,text=black] (ye) at (-10.5, -2) {$y_1$};
	\node[d4Node] (y0) at (-8.5, -2) {$y_{k-d}$};
	\node[d3Node] (y1) at (-6, -2) {$y_{k+1-d}$};
	\node[d2Node] (y2) at (-4, -2) {};
	\node[d1Node,text opacity=1] (y3) at (-2, -2) {$y_k$};

	\node[d0Node,text opacity=1] (c) at (0, 0) {c};

	\node[d1Node] (z1) at (2, 0) {$z_1$};
	\node[d2Node] (z2) at (4, 0) {};
	\node[d3Node,text opacity=1] (z3) at (6, 0) {$z_d$};
	\node[d4Node] (z4) at (8, 0) {$z_{d+1}$};
	\node[color=gray,forbidden sign, text=black] (ze) at (10, 0) {$z_k$};

	% J1
	\node[d1Node] (x31) at (-2, 4) {};
	\node[d1Node] (x32) at (-2, 6) {};

	\node[d1Node] (y31) at (-2, -4) {};
	\node[d1Node] (y32) at (-2, -6) {};

	\node[d1Node] (z11) at (2, 2) {$z'_1$};
	\node[d1Node] (z12) at (2, 4) {};

	% J2
	\node[d2Node] (x21) at (-4, 4) {};
	\node[d2Node] (x22) at (-4, 6) {};

	\node[d2Node] (y21) at (-4, -4) {};
	\node[d2Node] (y22) at (-4, -6) {};

	\node[d2Node] (z21) at (4, 2) {};
	\node[d2Node] (z22) at (4, 4) {};

	% J3
	\node[d3Node] (x11) at (-6, 4) {};
	\node[d3Node] (x12) at (-6, 6) {};
	\node[d3Node] (x13) at (-6, 8) {};

	\node[d3Node] (y11) at (-6, -4) {};
	\node[d3Node] (y12) at (-6, -6) {};
	\node[d3Node] (y13) at (-6, -8) {};

	\node[d3Node,text opacity=1] (z31) at (6, 2) {$z'_d$};
	\node[d3Node] (z32) at (6, 4) {};
	\node[d3Node] (z33) at (6, 6) {};

	\path[->,color=black] 
		% J0 
		(x1) edge[dashed] (x2)
		(x2) edge[dashed] (x3)
		(y1) edge[dashed] (y2)
		(y2) edge[dashed] (y3)
		(x3) edge (c)
		(y3) edge (c)
		(c) edge (z1)
		(z1) edge[dashed] (z2)
		(z2) edge[dashed] (z3)
		
		(xe) edge[dashed] (x0)
		(x0) edge (x1)
		(ye) edge[dashed] (y0)
		(y0) edge (y1)
		(z3) edge (z4)
		(z4) edge[dashed] (ze)

		% J1
		(x31) edge (x3)
		(x32) edge[dashed] (x31)
		(y31) edge (y3)
		(y32) edge[dashed] (y31)
		(z11) edge (z1)
		(z12) edge[dashed] (z11)

		% J2
		(x21) edge (x2)
		(x22) edge[dashed] (x21)
		(y21) edge (y2)
		(y22) edge[dashed] (y21)
		(z21) edge (z2)
		(z22) edge[dashed] (z21)

		% J3
		(x11) edge (x1)
		(x12) edge (x11)
		(x13) edge[dashed] (x12)
		(y11) edge (y1)
		(y12) edge (y11)
		(y13) edge[dashed] (y12)
		(z31) edge (z3)
		(z32) edge (z31)
		(z33) edge[dashed] (z32)
	;

	\draw [decorate, decoration={brace,mirror,amplitude=7}] 
    ([yshift=-9mm]x1.west) --node[below=4mm, draw=none, text opacity=1]{$d$} ([yshift=-9mm]x3.east);

	\draw [decorate, decoration={brace,amplitude=7}] 
    ([yshift=9mm]y1.west) --node[above=3mm, draw=none, text opacity=1]{} ([yshift=9mm]y3.east);

	\draw [decorate, decoration={brace,mirror,amplitude=7}] 
    ([yshift=-5mm]z1.west) --node[below=3mm, draw=none, text opacity=1]{$d$} ([yshift=-5mm]z3.east);

	\draw[decorate,decoration={brace,mirror,amplitude=3pt,raise=5pt}]
    ([xshift=1mm]y11.north west) -- ([xshift=1mm]y13.south west) node [draw=none,midway,xshift=-0.2cm,left, text opacity=1] {$\ell_d$};
	
	\draw[decorate,decoration={brace,mirror,amplitude=3pt,raise=5pt}]
    ([xshift=1mm]x13.north west) -- ([xshift=1mm]x11.south west) node [draw=none,midway,xshift=-0.2cm,left, text opacity=1] {$\ell_d$};
	
	\draw[decorate,decoration={brace,mirror,amplitude=3pt,raise=5pt}]
    ([xshift=1mm]z33.north west) -- ([xshift=1mm]z31.south west) node [draw=none,midway,xshift=-0.2cm,left, text opacity=1] {$\ell_d$};

	% J1

	\draw[->, color=ibm_purple] (x31) to[out=290,in=180] (-1.5,3) 
										to[out=0,in=270] (-1.1,5) 
										to[out=90,in=-90] (-1.2,7)
										to[out=90,in=120] (z12);
	\draw[->, color=ibm_purple] (y31) to[out=65,in=-100] (-0.9,-2) 
										to[out=80,in=-90] (-0.7,3)
										to[out=90,in=-90] (-0.8,6.5)
										to[out=90,in=140] (z12);
	
	% J2

	\draw[->, color=ibm_magenta] (x21) to[out=290,in=180] (-3.5,3) 
										to[out=0,in=270] (-3.1,5) 
										to[out=90,in=-90] (-3.3,8)
										to[out=90,in=120] (z22);
	\draw[->, color=ibm_magenta] (y21) to[out=65,in=-100] (-3.1,-2) 
										to[out=80,in=-90] (-2.7,3)
										to[out=90,in=-90] (-2.7,7.9)
										to[out=90,in=140] (z22);

	% J3
	\draw[->, color=ibm_orange] (x11) to[out=290,in=180] (-5.5,3) 
										to[out=0,in=270] (-5.1,5) 
										to[out=90,in=-90] (-5.3,9)
										to[out=90,in=120] (z33);

	\draw[->, color=ibm_orange] (y11) to[out=65,in=-100] (-5.1,-2) 
										to[out=80,in=-90] (-4.7,3)
										to[out=90,in=-90] (-4.7,8.9)
										to[out=90,in=140] (z33);

\end{tikzpicture}
      }
      \caption{the $d$-th step in the construction of $J_k$.}
      \label{fig:generic_binary_dag_J}
   \end{figure}
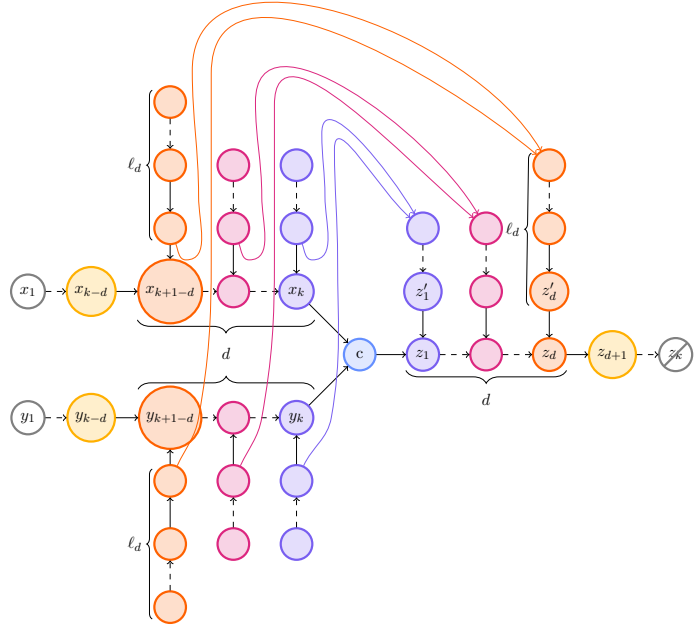

   We wish the number of vertices in the final graph $J_k$ to be odd in order to use Theorem~\ref{theo:comb}. If $J_k^k$ has an odd number of vertices, then we keep the digraph as such. If this number turns to be even, we just replace $\ell_k$ by $\ell_k + 1$ in the last step, which increases the number of vertices by $3$, and so makes it odd. The resulting digraph is denoted by $J_k$ .

   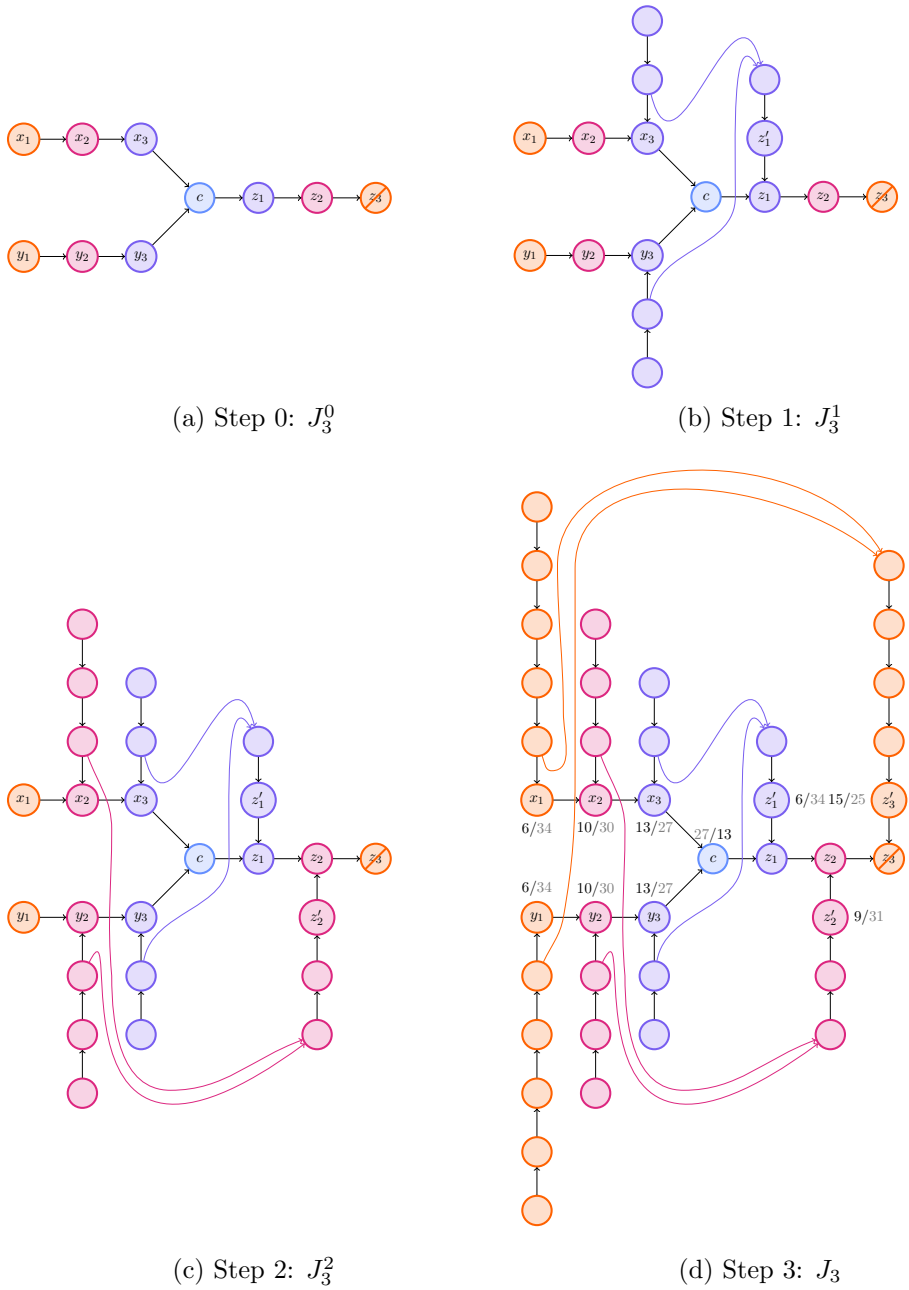
\begin{figure}
      \centering
      \begin{subfigure}[t]{.5\textwidth}
         \scalebox{0.65}{
            \begin{tikzpicture}
	[
		scale=.6,
		every node/.style={
			circle,
			draw,
			scale=0.8,
			text=black,
			line width=1.2pt,
			opacity=.2,
			text opacity=1,
			draw opacity=1
		},
		d0Node/.style={fill=ibm_blue,color=ibm_blue,text=black,line width=1.2pt},
		d1Node/.style={fill=ibm_purple,color=ibm_purple,text=black,line width=1.2pt},
		d2Node/.style={fill=ibm_magenta,color=ibm_magenta,text=black,line width=1.2pt},
		d3Node/.style={fill=ibm_orange,color=ibm_orange,text=black,line width=1.2pt},
		minimum size=0.75cm
	]

	% J0
	\node[d3Node] (x1) at (-6, 2) {$x_1$};
	\node[d2Node] (x2) at (-4, 2) {$x_2$};
	\node[d1Node] (x3) at (-2, 2) {$x_3$};
	\node[d3Node] (y1) at (-6, -2) {$y_1$};
	\node[d2Node] (y2) at (-4, -2) {$y_2$};
	\node[d1Node] (y3) at (-2, -2) {$y_3$};
	\node[d0Node] (c) at (0, 0) {$c$};
	\node[d1Node] (z1) at (2, 0) {$z_1$};
	\node[d2Node] (z2) at (4, 0) {$z_2$};
	\node[d3Node,forbidden sign] (z3) at (6, 0) {$z_3$};

	\node[draw=none] (centering) at (-2, -6) {};
	
	\path[->,color=black]
		% J0 
		(x1) edge (x2)
		(x2) edge (x3)
		(y1) edge (y2)
		(y2) edge (y3)
		(x3) edge (c)
		(y3) edge (c)
		(c) edge (z1)
		(z1) edge (z2)
		(z2) edge (z3)
	;

\end{tikzpicture}
         }
         \caption{Step 0: $J_3^0$}
         \label{fig:J30}
      \end{subfigure}%
      \begin{subfigure}[t]{.5\textwidth}
         \scalebox{0.65}{
            \begin{tikzpicture}
	[
		scale=.6,
		every node/.style={
			circle,
			draw,
			scale=0.8,
			text=black,
			line width=1.2pt,
			opacity=.2,
			text opacity=0,
			draw opacity=1
		},
		d0Node/.style={fill=ibm_blue,color=ibm_blue,text=black,line width=1.2pt},
		d1Node/.style={fill=ibm_purple,color=ibm_purple,text=black,line width=1.2pt},
		d2Node/.style={fill=ibm_magenta,color=ibm_magenta,text=black,line width=1.2pt},
		d3Node/.style={fill=ibm_orange,color=ibm_orange,text=black,line width=1.2pt},
		minimum size=0.75cm,
	]

	% J0
	\node[d3Node,text opacity=1] (x1) at (-6, 2) {$x_1$};
	\node[d2Node,text opacity=1] (x2) at (-4, 2) {$x_2$};
	\node[d1Node,text opacity=1] (x3) at (-2, 2) {$x_3$};
	\node[d3Node,text opacity=1] (y1) at (-6, -2) {$y_1$};
	\node[d2Node,text opacity=1] (y2) at (-4, -2) {$y_2$};
	\node[d1Node,text opacity=1] (y3) at (-2, -2) {$y_3$};
	\node[d0Node,text opacity=1] (c) at (0, 0) {$c$};
	\node[d1Node,text opacity=1] (z1) at (2, 0) {$z_1$};
	\node[d2Node,text opacity=1] (z2) at (4, 0) {$z_2$};
	\node[d3Node,text opacity=1,forbidden sign] (z3) at (6, 0) {$z_3$};

	% J1
	\node[d1Node] (x31) at (-2, 4) {};
	\node[d1Node] (x32) at (-2, 6) {};

	\node[d1Node] (y31) at (-2, -4) {};
	\node[d1Node] (y32) at (-2, -6) {};

	\node[d1Node,text opacity=1] (z11) at (2, 2) {$z'_1$};
	\node[d1Node] (z12) at (2, 4) {};

	\path[->,color=black] 
		% J0 
		(x1) edge (x2)
		(x2) edge (x3)
		(y1) edge (y2)
		(y2) edge (y3)
		(x3) edge (c)
		(y3) edge (c)
		(c) edge (z1)
		(z1) edge (z2)
		(z2) edge (z3)

		% J1
		(x31) edge (x3)
		(x32) edge (x31)
		(y31) edge (y3)
		(y32) edge (y31)
		(z11) edge (z1)
		(z12) edge (z11)
	;

	% J1
	\draw[->,color=ibm_purple] (x31) .. controls (-1,0) and (1.,8)  .. (z12);
	\draw[->,color=ibm_purple] (y31) to[out=80,in=-90] (1,0) to[out=90,in=120] (z12);

\end{tikzpicture}
         }
         \caption{Step 1: $J_3^1$}
         \label{fig:J31}
      \end{subfigure}
      \begin{subfigure}[t]{.5\textwidth}
         \scalebox{0.65}{
            \begin{tikzpicture}
	[
		scale=.6,
		every node/.style={
			circle,
			draw,
			scale=0.8,
			text=black,
			line width=1.2pt,
			opacity=.2,
			text opacity=0,
			draw opacity=1
		},
		d0Node/.style={fill=ibm_blue,color=ibm_blue,text=black,line width=1.2pt},
		d1Node/.style={fill=ibm_purple,color=ibm_purple,text=black,line width=1.2pt},
		d2Node/.style={fill=ibm_magenta,color=ibm_magenta,text=black,line width=1.2pt},
		d3Node/.style={fill=ibm_orange,color=ibm_orange,text=black,line width=1.2pt},
		minimum size=0.75cm,
	]

	% J0
	\node[d3Node,text opacity=1] (x1) at (-6, 2) {$x_1$};
	\node[d2Node,text opacity=1] (x2) at (-4, 2) {$x_2$};
	\node[d1Node,text opacity=1] (x3) at (-2, 2) {$x_3$};
	\node[d3Node,text opacity=1] (y1) at (-6, -2) {$y_1$};
	\node[d2Node,text opacity=1] (y2) at (-4, -2) {$y_2$};
	\node[d1Node,text opacity=1] (y3) at (-2, -2) {$y_3$};
	\node[d0Node,text opacity=1] (c) at (0, 0) {$c$};
	\node[d1Node,text opacity=1] (z1) at (2, 0) {$z_1$};
	\node[d2Node,text opacity=1] (z2) at (4, 0) {$z_2$};
	\node[d3Node,text opacity=1,forbidden sign] (z3) at (6, 0) {$z_3$};

	% J1
	\node[d1Node] (x31) at (-2, 4) {};
	\node[d1Node] (x32) at (-2, 6) {};

	\node[d1Node] (y31) at (-2, -4) {};
	\node[d1Node] (y32) at (-2, -6) {};

	\node[d1Node,text opacity=1] (z11) at (2, 2) {$z'_1$};
	\node[d1Node] (z12) at (2, 4) {};

	% J2
	\node[d2Node] (x21) at (-4, 4) {};
	\node[d2Node] (x22) at (-4, 6) {};
	\node[d2Node] (x23) at (-4, 8) {};

	\node[d2Node] (y21) at (-4, -4) {};
	\node[d2Node] (y22) at (-4, -6) {};
	\node[d2Node] (y23) at (-4, -8) {};

	\node[d2Node,text opacity=1] (z21) at (4, -2) {$z'_2$};
	\node[d2Node] (z22) at (4, -4) {};
	\node[d2Node] (z23) at (4, -6) {};

	\path[->,color=black] 
		% J0 
		(x1) edge (x2)
		(x2) edge (x3)
		(y1) edge (y2)
		(y2) edge (y3)
		(x3) edge (c)
		(y3) edge (c)
		(c) edge (z1)
		(z1) edge (z2)
		(z2) edge (z3)

		% J1
		(x31) edge (x3)
		(x32) edge (x31)
		(y31) edge (y3)
		(y32) edge (y31)
		(z11) edge (z1)
		(z12) edge (z11)

		% J2
		(x21) edge (x2)
		(x22) edge (x21)
		(x23) edge (x22)
		(y21) edge (y2)
		(y22) edge (y21)
		(y23) edge (y22)
		(z21) edge (z2)
		(z22) edge (z21)
		(z23) edge (z22)
	;

	% J1
	\draw[->,color=ibm_purple] (x31) .. controls (-1,0) and (1.,8)  .. (z12);
	\draw[->,color=ibm_purple] (y31) to[out=80,in=-90] (1,0) to[out=90,in=120] (z12);
	
	% J2
	\draw[->,color=ibm_magenta] (y21) .. controls (-2.5,-1) and (-5.5,-13) .. (z23);
	\draw[->, color=ibm_magenta] (x21) to[out=290,in=90] (-3,1) 
										to[out=-90,in=95] (-3,-4)
										to[out=-90,in=180] (-1,-7.9)
										to[out=0,in=200] (z23);

\end{tikzpicture}
         }
         \caption{Step 2: $J_3^2$}
         \label{fig:J32}
      \end{subfigure}%
      \begin{subfigure}[t]{.5\textwidth}
         \scalebox{0.65}{
            \begin{tikzpicture}
	[
		scale=.6,
		every node/.style={
			circle,
			draw,
			scale=0.8,
			text=black,
			line width=1.2pt,
			opacity=.2,
			text opacity=1,
			draw opacity=1
		},
		d0Node/.style={fill=ibm_blue,color=ibm_blue,text=black,line width=1.2pt},
		d1Node/.style={fill=ibm_purple,color=ibm_purple,text=black,line width=1.2pt},
		d2Node/.style={fill=ibm_magenta,color=ibm_magenta,text=black,line width=1.2pt},
		d3Node/.style={fill=ibm_orange,color=ibm_orange,text=black,line width=1.2pt},
		minimum size=0.75cm
	]

	% J0
	\node[d3Node,text opacity=1,label={[label distance=-0.3cm]-90:6/\textcolor{gray}{34}}] (x1) at (-6, 2) {$x_1$};
	\node[d2Node,text opacity=1,label={[label distance=-0.4cm]-90:10/\textcolor{gray}{30}}] (x2) at (-4, 2) {$x_2$};
	\node[d1Node,text opacity=1,label={[label distance=-0.4cm]-90:13/\textcolor{gray}{27}}] (x3) at (-2, 2) {$x_3$};
	\node[d3Node,text opacity=1,label={[label distance=-0.3cm]90:6/\textcolor{gray}{34}}] (y1) at (-6, -2) {$y_1$};
	\node[d2Node,text opacity=1,label={[label distance=-0.4cm]90:10/\textcolor{gray}{30}}] (y2) at (-4, -2) {$y_2$};
	\node[d1Node,text opacity=1,label={[label distance=-0.4cm]90:13/\textcolor{gray}{27}}] (y3) at (-2, -2) {$y_3$};
	\node[d0Node,text opacity=1,label={[label distance=-0.5cm]90:\textcolor{gray}{27}/13}] (c) at (0, 0) {$c$};
	\node[d1Node,text opacity=1] (z1) at (2, 0) {$z_1$};
	\node[d2Node,text opacity=1] (z2) at (4, 0) {$z_2$};
	\node[d3Node,text opacity=1,forbidden sign] (z3) at (6, 0) {$z_3$};

	% J1
	\node[d1Node] (x31) at (-2, 4) {};
	\node[d1Node] (x32) at (-2, 6) {};

	\node[d1Node] (y31) at (-2, -4) {};
	\node[d1Node] (y32) at (-2, -6) {};

	\node[d1Node,text opacity=1,label={[label distance=-0.1cm]0:6/\textcolor{gray}{34}}] (z11) at (2, 2) {$z'_1$};
	\node[d1Node] (z12) at (2, 4) {};

	% J2
	\node[d2Node] (x21) at (-4, 4) {};
	\node[d2Node] (x22) at (-4, 6) {};
	\node[d2Node] (x23) at (-4, 8) {};

	\node[d2Node] (y21) at (-4, -4) {};
	\node[d2Node] (y22) at (-4, -6) {};
	\node[d2Node] (y23) at (-4, -8) {};

	\node[d2Node,text opacity=1,label={[label distance=-0.1cm]0:9/\textcolor{gray}{31}}] (z21) at (4, -2) {$z'_2$};
	\node[d2Node] (z22) at (4, -4) {};
	\node[d2Node] (z23) at (4, -6) {};

	% J3
	\node[d3Node] (x11) at (-6, 4) {};
	\node[d3Node] (x12) at (-6, 6) {};
	\node[d3Node] (x13) at (-6, 8) {};
	\node[d3Node] (x14) at (-6, 10) {};
	\node[d3Node] (x15) at (-6, 12) {};

	\node[d3Node] (y11) at (-6, -4) {};
	\node[d3Node] (y12) at (-6, -6) {};
	\node[d3Node] (y13) at (-6, -8) {};
	\node[d3Node] (y14) at (-6, -10) {};
	\node[d3Node] (y15) at (-6, -12) {};

	\node[d3Node,text opacity=1,label={[label distance=-0.1cm]180:15/\textcolor{gray}{25}}] (z31) at (6, 2) {$z'_3$};
	\node[d3Node] (z32) at (6, 4) {};
	\node[d3Node] (z33) at (6, 6) {};
	\node[d3Node] (z34) at (6, 8) {};
	\node[d3Node] (z35) at (6, 10) {};

	\path[->,color=black] 
		% J0 
		(x1) edge (x2)
		(x2) edge (x3)
		(y1) edge (y2)
		(y2) edge (y3)
		(x3) edge (c)
		(y3) edge (c)
		(c) edge (z1)
		(z1) edge (z2)
		(z2) edge (z3)

		% J1
		(x31) edge (x3)
		(x32) edge (x31)
		(y31) edge (y3)
		(y32) edge (y31)
		(z11) edge (z1)
		(z12) edge (z11)

		% J2
		(x21) edge (x2)
		(x22) edge (x21)
		(x23) edge (x22)
		(y21) edge (y2)
		(y22) edge (y21)
		(y23) edge (y22)
		(z21) edge (z2)
		(z22) edge (z21)
		(z23) edge (z22)

		% J3
		(x11) edge (x1)
		(x12) edge (x11)
		(x13) edge (x12)
		(x14) edge (x13)
		(x15) edge (x14)
		(y11) edge (y1)
		(y12) edge (y11)
		(y13) edge (y12)
		(y14) edge (y13)
		(y15) edge (y14)
		(z31) edge (z3)
		(z32) edge (z31)
		(z33) edge (z32)
		(z34) edge (z33)
		(z35) edge (z34)
	;

	% J1
	\draw[->,color=ibm_purple] (x31) .. controls (-1,0) and (1.,8)  .. (z12);
	\draw[->,color=ibm_purple] (y31) to[out=80,in=-90] (1,0) to[out=90,in=120] (z12);
	
	% J2
	\draw[->,color=ibm_magenta] (y21) .. controls (-2.5,-1) and (-5.5,-13) .. (z23);
	\draw[->, color=ibm_magenta] (x21) to[out=290,in=90] (-3,1) 
										to[out=-90,in=95] (-3,-4)
										to[out=-90,in=180] (-1,-7.9)
										to[out=0,in=200] (z23);

	% J3
	\draw[->, color=ibm_orange] (x11) to[out=290,in=180] (-5.5,3) 
										to[out=0,in=270] (-5.1,5) 
										to[out=90,in=-90] (-5.3,10)
										to[out=90,in=120] (z35);

	\draw[->, color=ibm_orange] (y11) to[out=65,in=-100] (-5.1,-2) 
										to[out=80,in=-90] (-4.7,3)
										to[out=90,in=-90] (-4.7,9.9)
										to[out=90,in=140] (z35);

\end{tikzpicture}
         }
         \caption{Step 3: $J_3$}
         \label{fig:J33}
      \end{subfigure}
      \caption{Construction of $J_3$.}
      \label{fig:construction_of_J3}
   \end{figure}

   The construction for $k=3$ is shown in Figure~\ref{fig:construction_of_J3}.

   \noindent\textbf{Number of vertices in the final digraph.}
   For each $d > 1$, the number of vertices $n_d$ satisfies the inequality
   \[n_d = n_{d-1} + 3 \ell_d \leq \dfrac 3 2 n_{d-1} + 3.\]
   A quick induction shows that
   \[ n_d \leq \left( \dfrac 3 2 \right)^d (3k+7) -6.\]
   Remember that, if $n_k$ is even, we have added $3$ vertices in the final digraph. In any case, the number of vertices in $J_k$ is bounded by $\left( \dfrac 3 2 \right)^d (3k+7)$.
   
   \begin{claim}
      The sequence $(\ell_d)_{d \in \{1,\dots,k\}}$ is strictly increasing.
      \label{claim:increasing}
   \end{claim}
   
   The sequence is reduced to one element for $k=1$, so we can assume that $k \geq 2$.
   Let $ 1 \leq d < k$. By definition of $\ell_d$, we have \( \ell_{d+1} > \frac{n_d} 6 \) and \( \ell_{d} \leq \frac{n_{d-1}} 6 + 1\). So 
   \[ \ell_{d+1}-\ell_d > \frac {n_d} 6 - \frac{n_{d-1}} 6 - 1 = \frac{\ell_d} 2 - 1.\]
   But \[\ell_d > \frac{n_{d-1}} 6 \geq \frac{n_{0}} 6 = \frac{3k+1}{6} \geq 2,\]
   whenever $k\geq 2$. We conclude from the former that $\ell_{d+1} - \ell_d > 0$.

   \begin{claim}
      When $\vertex c$ is the \faulty, \gbalgo\ uses $k + \left\lceil \log_2(k+1) \right\rceil + 2$ queries on $J_k$.
      \label{claim:gb_for_Jk}
   \end{claim}

   We are going to show that the resulting digraph just after the $i$-th step of \gbalgo\,is $J^{k-i}_k$, for $i \in \{0,\dots,k\}$.
   In other words, after $k$ \gb\ queries, we end up with the backbone $J_k^0$. After we show this fact, the claim is easily proved. Indeed, two extra queries from $J_k^0$  lead to a binary search on a directed path with $k+1$ vertices, for which \gb\ uses $\lceil \log_2(k+1) \rceil$ queries to find the \faulty.
   This explains why the number of \gb\ queries is $k + 2 + \lceil \log_2(k+1) \rceil$.
   To do so, we prove by induction on $d$ the construction invariants in $J_k^d$:
   \begin{itemize}
      \item $\score(\vertex{x_k}) = |x_k| = \frac{n_d - 1}{3}$
      \item $\score(\vertex{y_k}) = |y_k| = \frac{n_d - 1}{3}$
      \item $\score(\vertex{c}) = n_d - |c| = \frac{n_d - 1}{3}$
      \item for all $1 \leq i \leq d$, $\score(\vertex{z'_i}) = |z'_i| = 3\ell_i$
   \end{itemize}

   All these properties clearly hold for $d=0$.

   Let us assume now the induction hypotheses for $d-1$.
   By construction, the number of ancestors of $\vertex{x_k}$ in $J_k^{d}$ increases by $\ell_d$ in comparison with its number in $J_k^{d-1}$, while its number of non-ancestors
   increases by $2\ell_d$. It is the same for $\vertex{y_k}$. As for $\vertex{c}$, its number of ancestors increases by $2\ell_d$, while its number  of non-ancestors increases by $\ell_d$. From these observations, we inductively infer the first three invariants. As for the last item, it is obvious that $|z'_i| = 3 \ell_i$ by construction. The number of non-ancestors of $\vertex{z'_i}$ is $n_d - 3 \ell_i$, which is at least $n_d - 3 \ell_d=n_{d-1}$ by Claim~\ref{claim:increasing}. For $k=1$, we have $n_0 = 4 > 3 = 3 \ell_1$.  For $k \geq 2$, we saw in the proof of Claim~\ref{claim:increasing} that $\ell_i \geq 2$, hence \[n_{d-1} \geq 6 (\ell_d - 1) \geq 6 (\ell_i - 1) \geq 3 \ell_i.\]
The score of $\vertex{z'_i}$ is thus  $|z'_i| = 3 \ell_i$.

   Now, let us suppose that the digraph just before the $i$-th step is $J_k^{m}$, where $m = k - i + 1$, and let us show that after the $i$-th step, the digraph becomes $J_k^{m-1}$.
   To do so, we have to investigate the scores of all vertices in $J_k^{m}$.
   By construction, each vertex is either an ancestor of $\vertex{x_k}$,
   an ancestor of $\vertex{y_k}$, a descendant of $\vertex{c}$,
   or an ancestor of a vertex $\vertex{z'_j}$ with $j \in \{1,\dots,m\}$.
   The vertex $\vertex{x_k}$ having fewer ancestors than non-ancestors,
   the ancestors of $\vertex{x_k}$ different from $\vertex{x_k}$ have a smaller score than $\vertex{x_k}$.
   Thus \gbalgo\ never queries an ancestor of $\vertex{x_k}$ different from $\vertex{x_k}$.
   Similarly, we can eliminate every other vertex, excepted $\vertex{x_k}$,
   $\vertex{y_k}$, $\vertex{c}$ and $\vertex{z'_j}$ with $j \in \{1,\dots,m\}$.

   We already saw that the scores of $\vertex{x_k}$, $\vertex{y_k}$ and $\vertex{c}$  are the same and bounded by $\frac {n_m} {3}$.
   As for the vertex $\vertex{z'_j}$, its score is equal to $3  \ell_j$.
   Since $\ell_j$ is strictly increasing by Claim~\ref{claim:increasing}, we can eliminate every vertex $\vertex{z'_j}$ for $j < m$.
   It remains to compute the score of $\vertex{z'_m}$. Remark that $6 \ell_d > n_{d-1}$ by the definition of $\ell_d$.
   We deduce that \[ n_m = n_{m-1} + 3 \ell_m < 9 \ell_m.\]
   But $\score(\vertex{z'_m}) = |z'_m| = 3 \ell_m$, which is bigger than $\dfrac{n_m} 3$ by the above inequality.

   So $\vertex{z'_m}$ is the only vertex with a maximal score; \gbalgo\ will query this vertex. Since $\vertex c$ is not an ancestor of $\vertex{z'_m}$, \gb\ will remove every ancestor of $\vertex{z'_m}$: we recover $J_k^{m-1}$.

   \medskip

   \noindent\textbf{Conclusion.} By Theorem~\ref{theo:comb} and Claim~\ref{claim:gb_for_Jk}, \gbalgo\ uses $k + \lceil \log_2(k+1) \rceil + 2 + 1$ queries on $Comb(J_k)$. Also by Theorem~\ref{theo:comb}, since the number of vertices in $J_k$ is odd and is bounded by $\left( \dfrac 3 2 \right)^k (3k+7)$, the optimal strategy uses at most $\left\lceil \log_2\left(2\left(\dfrac 3 2 \right)^k (3k+7)\right) \right\rceil \le \log_2\left(\frac 3 2 \right)\,k + \log_2(3k+7) + 2$ queries.
   This concludes the proof of Proposition~\ref{prop:tight_case}.
\end{proof}

By Proposition~\ref{prop:tight_case}, we cannot find a better approximation ratio than $\frac{1}{\log_2(3/2)}$ for \gb.

\begin{corollary}
    For any $\varepsilon > 0$, \gbalgo\ is not a $\left( \frac 1 {\log_2(3/2)} - \varepsilon \right)$-approximation algorithm for binary DAGs.
\end{corollary}

\subsection{Generalisation for \texorpdfstring{$\Delta$}{Delta}-ary DAGs}

For any $\Delta \geq 1$, a DAG is said to be \emph{$\Delta$-ary} if each of its vertices has indegree at most equal to $\Delta$. It is worth noting that the results for binary DAGs can be naturally extended to $\Delta$-ary DAGs.

Indeed, Lemma~\ref{lem:good_score_vertex}, which is of paramount importance to understand the structure of binary DAGs, can be generalised as follows.

\begin{lemma}
   In every $\Delta$-ary DAG with $n$ vertices, there exists a vertex $\vertex{v}$ such that $\ancestors v$, its number of ancestors, satisfies $\frac{n-1}{\Delta+1} < \ancestors v \le \frac{\Delta n+1}{\Delta+1}$.
\end{lemma}

This leads to the following theorem.

\begin{theorem}
   On any $\Delta$-ary DAG with $n$ vertices, the number of queries of \gbalgo\ is at most $\frac{\log_2(n)}{\log_2\left(\frac{\Delta+1}{\Delta}\right)}$.

   Consequently, \gbalgo\ is a $\frac{1}{\log_2\left(\frac{\Delta+1}{\Delta}\right)}$-approximation algorithm on $\Delta$-ary DAGs.
\end{theorem}

Note that the bound above is tight.
Indeed, the previous construction of graphs $J_k^d$ can be extended by merging $\Delta+1$ paths, $\Delta$ of which are analogous to $\vertex{x_1} \to \vertex{x_2} \to \cdots \to \vertex{x_{k+1}}$.
We get that $n_d = n_{d-1} + (\Delta+1) \ell_d$, and $\vertex c$ as well as each vertex $\vertex{x_k}$ has score $\frac{n_{d-1}}{\Delta+1}$. The \gb\ algorithm selects the vertex $\vertex{z'_d}$, which has $(\Delta+1) \ell_d$ ancestors.

Choosing an appropriate value for $\ell_d$ (namely $\lfloor n_{d-1}/\Delta(\Delta+1) \rfloor + 1$),
we end up with a graph of order $(1+1/\Delta)^k(\Delta k + O(1))$.
Thus, the optimal strategy on the comb requires $k \log_2((\Delta+1)/\Delta) + o(k)$ queries.

\section{A new algorithm with a better approximation ratio for binary DAGs}
\label{sec:golden}

In this section, we describe a new algorithm improving the number of queries in the worst-case scenario compared to \gb.

\subsection{Description of \golden}

We design a new algorithm for the Regression Search Problem, named \emph{\golden}, which is a slight modification of \gb.
It is so called because it is based on the \emph{golden ratio}, defined as $\phi = \frac{1+\sqrt{5}}{2}$.

The difference of \golden\ with respect to \gb\ is that it may not query a vertex with the maximum score if the maximum score is too small.
Let us give some preliminary definitions.

\begin{definition}[Subsets $B^\ge$ and $B^<$]
   Let $D$ be a DAG. We define $V^\ge$ as the set of vertices which have more ancestors than non-ancestors. Let $B^\ge$ (for ``Best'' or ``Boundary'') denote the subset of vertices $\vertex{v}$ of $V^\ge$ such that no parent of $\vertex{v}$ belongs to $V^\ge$ and let $B^<$ be the set of parents of vertices of $B^\ge$.
\end{definition}

The reader can look at Figure~\ref{fig:sets_of_vertices} for an illustrative example.
Note that the score of a vertex $\vertex{v}$ with $\ancestors{v}$ ancestors is $n-\ancestors{v}$ if $\vertex{v} \in V^{\geq}$, or $\ancestors{v}$ if $\vertex{v} \notin V^{\geq}$.

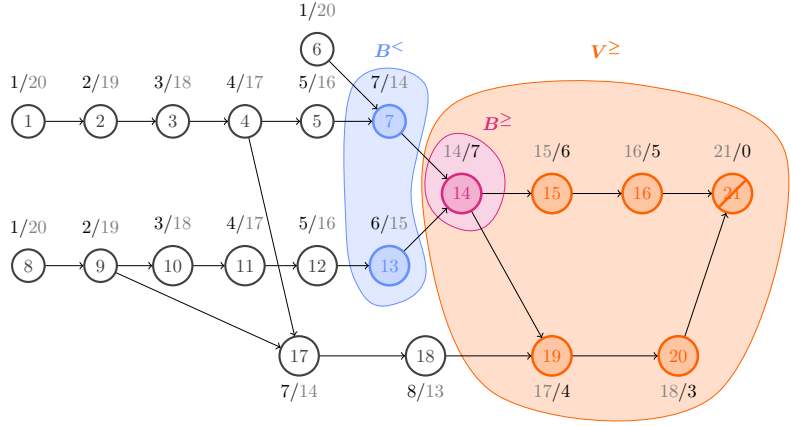
\begin{figure}[ht]
   \centering
   \scalebox{0.8}{
      \begin{tikzpicture}
	[
		scale=.6,
		every node/.style={
			circle,
			draw,
			scale=0.8,
			text=black,
			line width=1.2pt,
			opacity=.2,
			text opacity=1,
			draw opacity=1
		},
		vPlusNode/.style={fill=ibm_orange,color=ibm_orange,text=ibm_orange,line width=1.2pt},
		bPlusNode/.style={fill=ibm_magenta,color=ibm_magenta,text=ibm_magenta,line width=1.2pt},
		bMinusNode/.style={fill=ibm_blue,color=ibm_blue,text=ibm_blue,line width=1.2pt},
		normalNode/.style={text opacity=1,opacity=0,draw opacity=1,color=darkgray,text=darkgray,line width=1.pt}
	]

	\draw[color=ibm_orange,fill=ibm_orange!20]  plot[smooth cycle, tension=.6] coordinates {(17,7) (21,5) (19,-1.6)  (13,-1.6) (11,3) (11.8,6.5)};
	\draw[color=ibm_magenta,fill=ibm_magenta!20]  plot[smooth cycle, tension=.8] coordinates {(12,3) (11,4) (11.6,5.6) (13,5) (13,3.5)};
	\draw[color=ibm_blue,fill=ibm_blue!20]  plot[smooth cycle, tension=.65] coordinates {(9,7) (11,7) (10.6,4) (11,1.7) (10,0.9) (9.2,1.3) (8.8,3)};

	\node[draw=none,color=ibm_orange] at (16,8) {$\boldsymbol{V^\ge}$};
	\node[draw=none,color=ibm_magenta] at (13,6) {$\boldsymbol{B^\ge}$};
	\node[draw=none,color=ibm_blue] at (10,8) {$\boldsymbol{B^<}$};

	\node[normalNode,label={[label distance=-0.2cm]90:1/\textcolor{gray}{20}}] (1) at (0, 6) {$1$};
	\node[normalNode,label={[label distance=-0.2cm]90:2/\textcolor{gray}{19}}] (2) at (2, 6) {$2$};
	\node[normalNode,label={[label distance=-0.2cm]90:1/\textcolor{gray}{20}}] (3) at (8,8) {$6$};
	\node[normalNode,label={[label distance=-0.2cm]90:3/\textcolor{gray}{18}}] (4) at (4,6) {$3$};
	\node[normalNode,label={[label distance=-0.2cm]90:4/\textcolor{gray}{17}}] (5) at (6,6) {$4$};
	\node[normalNode,label={[label distance=-0.2cm]90:5/\textcolor{gray}{16}}] (6) at (8,6) {$5$};
	\node[bMinusNode,label={[label distance=-0.2cm]90:7/\textcolor{gray}{14}}] (7) at (10, 6) {$7$};
	\node[normalNode,label={[label distance=-0.2cm]90:1/\textcolor{gray}{20}}] (8) at (0, 2) {$8$};
	\node[normalNode,label={[label distance=-0.2cm]90:2/\textcolor{gray}{19}}] (9) at (2,2) {$9$};
	\node[normalNode,label={[label distance=-0.2cm]90:3/\textcolor{gray}{18}}] (10) at (4, 2) {$10$};
	\node[normalNode,label={[label distance=-0.2cm]90:4/\textcolor{gray}{17}}] (11) at (6, 2) {$11$};
	\node[normalNode,label={[label distance=-0.2cm]90:5/\textcolor{gray}{16}}] (12) at (8, 2) {$12$};
	\node[bMinusNode,label={[label distance=-0.2cm]90:6/\textcolor{gray}{15}}] (13) at (10, 2) {$13$};
	\node[bPlusNode,label={[label distance=-0.2cm]90:\textcolor{gray}{14}/7}] (14) at (12, 4) {$14$};
	\node[vPlusNode,label={[label distance=-0.2cm]90:\textcolor{gray}{15}/6}] (15) at (14.5, 4) {$15$};
	\node[vPlusNode,label={[label distance=-0.2cm]90:\textcolor{gray}{16}/5}] (16) at (17, 4) {$16$};
	\node[vPlusNode,label={[label distance=-0.2cm]90:\textcolor{gray}{21}/0},forbidden sign] (17) at (19.5, 4) {$21$};
	\node[normalNode,label={[label distance=-0.3cm]-90:7/\textcolor{gray}{14}}] (18) at (7.5, -0.5) {$17$};
	\node[normalNode,label={[label distance=-0.3cm]-90:8/\textcolor{gray}{13}}] (19) at (11, -0.5) {$18$};
	\node[vPlusNode,label={[label distance=-0.3cm]-90:\textcolor{gray}{17}/4}] (20) at (14.5, -0.5) {$19$};
	\node[vPlusNode,label={[label distance=-0.3cm]-90:\textcolor{gray}{18}/3}] (21) at (18, -0.5) {$20$};

	\path[->,color=black]
		(1) edge (2)
		(2) edge (4)
		(4) edge (5)
		(5) edge (6)
		(6) edge (7)
		(3) edge (7)
		(7) edge (14)
		
		(8) edge (9)
		(9) edge (10)
		(10) edge (11)
		(11) edge (12)
		(12) edge (13)
		(13) edge (14)
		(14) edge (15)
		(15) edge (16)
		(16) edge (17)
		
		(18) edge (19)
		(19) edge (20)
		(20) edge (21)
		(21) edge (17)
		
		(5) edge (18)
		(9) edge (18)
		(14) edge (20)
	;
\end{tikzpicture}
   }
   \caption{A binary DAG with the $3$ sets of vertices $V^\ge$, $B^\ge$ and $B^<$.}
   \label{fig:sets_of_vertices}
\end{figure}

Using the sets defined above, we propose a refinement of Lemma~\ref{lem:good_score_vertex}.

\begin{lemma}
   Given any DAG with $n$ vertices, there exists a vertex $\vertex v \in B^\ge \cup B^<$ such that $\score(\vertex v) \ge \dfrac{n-1} 3$.
   \label{lem:score_B_sets}
\end{lemma}

\begin{proof}
   The lemma is obvious whenever $n \leq 3$ (choose a vertex with no parent).

   Let us choose any $\vertex{v}$ in $B^\ge$.
   Since the graph is binary, $\vertex{v}$ has 1 or 2 parents. Let us study both cases separately.

   \textbf{Vertex $\vertex{v}$ has only one parent $\vertex{p}$.} Thus, $\vertex{p}$ has exactly $\ancestors v-1$ ancestors  and $n - \ancestors v + 1$ non-ancestors.
   Since $\vertex{p}$ is not in $V^\ge$, $\ancestors v-1 < n- \ancestors v +1$, and thus $\ancestors v < \frac{n}{2}+1$.

   Also, since $\vertex v \in V^\ge$, $\score(\vertex v)=n-\ancestors v \geq \frac{n}{2}-1$, which satisfies the lemma whenever $n \geq 4$.

   \textbf{Vertex $\vertex{v}$ has two parents $\vertex{x}$ and $\vertex{y}$,} belonging to $B^<$, respectively having $\ancestors{x}$ and $\ancestors{y}$ ancestors.
   If any of $\vertex{x}$ or $\vertex{y}$ has $\dfrac{n-1}{3}$ or more ancestors, then the lemma holds for $\vertex{x}$ or $\vertex{y}$, since $\score(\vertex x)=\ancestors{x}$ and $\score(\vertex y)=\ancestors{y}$.

   Let us assume the contrary, that is $\ancestors{x} < \dfrac{n-1}{3}$ and $\ancestors{y} < \dfrac{n-1}{3}$.
   But aside itself, every ancestor of $\vertex{v}$ must be an ancestor of $\vertex{x}$ or $\vertex{y}$.
   Hence
   \[\ancestors{v} \le \ancestors{x}+\ancestors{y}+1 < \dfrac{n-1}{3} + \dfrac{n-1}{3}+1 = \dfrac{2n+1}{3}.\]
   Since $\vertex v$ is in $V^\geq$, $\score(\vertex v) \ge n-\dfrac{2n+1}{3} = \dfrac{n-1}{3}$. Thus $\vertex v$ satisfies the condition of the lemma.
\end{proof}

A description of \golden\ is given by Algorithm~\ref{algo:golden}.

\begin{algorithm}[ht]
   \caption{\golden}
   \textbf{Input.} A DAG $D$ and a bugged vertex $\vertex{b}$. \\
   \textbf{Output.} The \faulty\ of $D$. \\
   \textbf{Steps:}
   \begin{enumerate}
      \item Remove from $D$ all non-ancestors of $\vertex{b}$.
      \item If $D$ has only one vertex, return this vertex. \label{item:start_golden}
      \item Compute the score for each vertex of $D$.
      \textbf{\item If the maximum score is at least $\frac n {\phi^2} \approx 38.2 \% \times  n$ (where $\phi = \frac{1+\sqrt{5}}{2}$), query a vertex with the maximum score. \label{item:one_move}
      \item Otherwise, query a vertex of $B^{\geq} \cup B^<$ which has the maximum score among vertices of $B^{\geq} \cup B^<$, even though it may not be the overall maximum score. \label{item:two_moves} }
      \item If the queried vertex is bugged, remove from $D$ all non-ancestors of the queried vertex. Otherwise, remove from $D$ all ancestors of the queried vertex.
      \item Go to Step~\ref{item:start_golden}.
   \end{enumerate}
   (The differences with \gb\ are displayed in \textbf{bold}.)
   \label{algo:golden}
\end{algorithm}

Now, let us describe the behavior of the \golden\ algorithm on the example of Figure \ref{fig:sets_of_vertices}. We have $21/\phi^2 \approx 8.02$. The maximum score $8$ is smaller than this number, so we run Step \ref{item:two_moves} instead of Step~\ref{item:one_move}.
Thus as its first query, \golden\ chooses $\vertex{7}$, which belongs to $B^<$ and has score $7$.
Another possible first query is to choose $\vertex{14}$, which has the same score as $\vertex{7}$, but belongs to $B^\ge$.
In both cases, \golden\ uses 5 queries in the worst-case scenario (see Figure~\ref{fig:complete_golden_strategy_example} for a possible strategy tree whenever $\vertex{7}$ is queried). 
This diverges from \gb, which would pick $\vertex{18}$ (with a score of $8$) with 6 queries in the worst-case scenario.

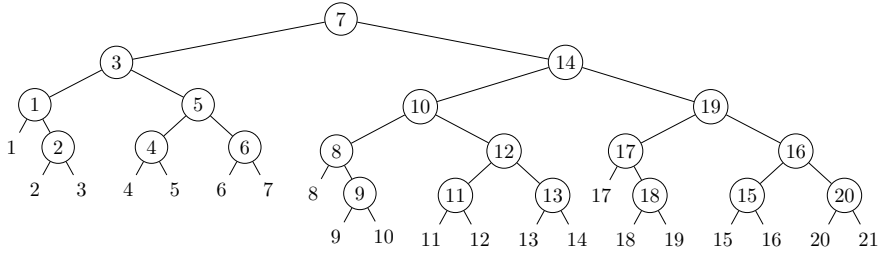
\begin{figure}[ht]
   \centering
   \scalebox{0.7}{
      \begin{forest} 
    for tree={inner sep=2pt,l=10pt,l sep=5pt,circle,draw,minimum size=1.6em}
    [7
    [3
        [1
            [1,draw=none]
            [2
                [2,draw=none]
                [3,draw=none]
            ]
        ]
        [5
            [4
                [4,draw=none]
                [5,draw=none]
            ]
            [6
                [6,draw=none]
                [7,draw=none]
            ]
        ]
    ]
    [14
        [10
            [8
                [8,draw=none]
                [9
                    [9,draw=none]
                    [10,draw=none]
                ]
            ]
            [12
                [11
                    [11,draw=none]
                    [12,draw=none]
                ]
                [13
                    [13,draw=none]
                    [14,draw=none]
                ]
            ]
        ]
        [19
            [17
                [17,draw=none]
                [18
                    [18,draw=none]
                    [19,draw=none]
                ]
            ]
            [16
                [15
                    [15,draw=none]
                    [16,draw=none]
                ]
                [20
                    [20,draw=none]
                    [21,draw=none]
                ]
            ]
        ]
    ]
    ]
\end{forest}
   }
   \caption{A \golden\ strategy tree for the digraph of Figure \ref{fig:sets_of_vertices}. In case of equality of score, the vertex with the smallest label is chosen.}
   \label{fig:complete_golden_strategy_example}
\end{figure}

The \golden\ strategy is not always better than \gb: Figure~\ref{fig:gb_better_than_golden} shows an example of binary DAG where \gb\ is thriftier than \golden. 

In this figure, vertex $\vertex a$ is queried first by \golden\ contrary to \gb\ which starts with vertex $\vertex b$. Each grey rectangle represents a directed path and the number inside is the number of vertices.

To understand why \golden\ is less efficient than \gb, both strategies are also presented under the form of trees where the nodes represent the induced subgraphs. 
The directed paths are not developed here and the colored squares represent the number of additional queries used to find the \faulty\  in the worst-case scenario.

\begin{figure}[h!]
   \centering
   \begin{subfigure}[c]{\textwidth}
      \centering
      (a)\hspace{0.5cm}\vspace{.2cm}
      \scalebox{0.6}{
         $\begin{array}{l}
            \begin{tikzpicture}
	[
		scale=.6, every node/.style={circle,draw,scale=0.8, text=black, line width=1.2pt, text opacity=1, draw opacity=1},
		minimum size=0.45cm,
		normalNode/.style={opacity=1,draw opacity=1,color=darkgray,text=black, line width=1.1pt,fill=white}, 
		labelNode/.style={opacity=0,draw opacity=1,color=darkgray,text=black, line width=1.1pt,draw=none},
		goldenNode/.style={opacity=0.1,draw opacity=1,fill=ibm_orange, color=ibm_orange,text=ibm_orange, line width=1.2pt,minimum size=0.7cm},
		gbNode/.style={opacity=0.1,draw opacity=1,fill=ibm_purple, color=ibm_purple,text=ibm_purple, line width=1.2pt, minimum size=0.7cm},
	]

	\draw[rounded corners, densely dashed, line width=1.5pt, color=darkgray, fill=grey] (-4,1.3) rectangle (4,2.70);
	\draw[rounded corners, densely dashed, line width=1.5pt, color=darkgray, fill=grey] (0,-4.7) rectangle (6,-3.3);
	\draw[rounded corners, densely dashed, line width=1.5pt, color=darkgray, fill=grey] (13,-1.7) rectangle (21,-0.3);

	\node[normalNode] (1) at (-6, 2) {};
	\node[normalNode] (2) at (-4, 2) {};
	\node[normalNode] (3) at (4, 2) {};
	\node[labelNode] (33) at (0, 2) {\Large 14};
	\node[normalNode] (4) at (6, 2) {};
	\node[goldenNode,label={[label distance=-0.2cm]90:\large 17/\textcolor{gray}{32}}] (5) at (8, 2) {\large \textbf{a}};
	\node[normalNode,label={[label distance=-0.2cm]90:\large \textcolor{gray}{33}/16}] (6) at (11, -1) {};
	\node[normalNode,label={[label distance=-0.2cm]90:\large 15/\textcolor{gray}{34}}] (7) at (8, -4) {};
	\node[normalNode] (8) at (6, -4) {};
	\node[normalNode] (9) at (0, -4) {};
	\node[labelNode] (99) at (3, -4) {\Large 11};
	\node[normalNode] (10) at (-2, -4) {};
	\node[normalNode] (11) at (-4, -3) {};
	\node[normalNode] (12) at (-4, -5) {};
	\node[normalNode] (13) at (13, -1) {};
	\node[normalNode] (14) at (21, -1) {};
	\node[labelNode] (144) at (17, -1) {\Large 14};
	\node[normalNode] (15) at (23, -1) {};
	\node[gbNode,label={[label distance=-1.9cm]90:\large 18/\textcolor{gray}{31}}] (16) at (11, -6) {\large \textbf{b}};

	\path[->,color=black]
		(1) edge (2)
		(3) edge (4)
		(4) edge (5)
		(5) edge (6)
		(7) edge (6)
		(8) edge (7)
		(10) edge (9)
		(11) edge (10)
		(12) edge (10)
		(6) edge (13)
		(14) edge (15)
		(4) edge[bend left=5] (16)
		(12) edge[bend right=5] (16)
		(16) edge[bend right=20] (15)
	;

\end{tikzpicture}
         \end{array}$
      }
   \end{subfigure}

   \begin{subfigure}[c]{\textwidth}
      \centering
      (b)\hspace{0.5cm} \vspace{-.0cm}
      \scalebox{0.5}{
         $\begin{array}{l}
            \begin{tikzpicture}
	[
		scale=.6, every node/.style={circle,draw,scale=0.8, text=black, line width=1.2pt, text opacity=1, draw opacity=1},
		minimum size=0.45cm,
		normalNode/.style={opacity=1,draw opacity=1,color=darkgray,text=black, line width=1.1pt,fill=white}, 
		labelNode/.style={opacity=0,draw opacity=1,color=darkgray,text=black, line width=1.1pt,draw=none},
		gbNode/.style={fill,opacity=1,color=ibm_purple!15,draw=ibm_purple,text=ibm_purple, line width=1.2pt},
	]

	\newcommand*{\dipath}[3]{
		\pgfmathsetmacro{\x}{#1}
		\pgfmathsetmacro{\y}{#2}
		\pgfmathsetmacro{\yRectangleTop}{#2+0.5}
		\pgfmathsetmacro{\yRectangleBottom}{#2-0.5}
		\pgfmathsetmacro{\xRectangleBottom}{#1-0.8}

		\draw[fill,opacity=0.5,color=light_grey,rounded corners=5pt] (\xRectangleBottom,\the\numexpr\y-0.8) rectangle (\the\numexpr\x+2.8,\the\numexpr\y+0.8); 

		\draw[rounded corners, densely dashed, line width=1.5pt, color=darkgray, fill=grey] (\the\numexpr\x,\the\numexpr\yRectangleBottom) rectangle (\the\numexpr\x+2,\the\numexpr\yRectangleTop);

		\node[normalNode] (1) at (\the\numexpr\x, \the\numexpr\y) {};
		\node[normalNode] (2) at (\the\numexpr\x+2, \the\numexpr\y) {};
		\node[labelNode] (3) at (\the\numexpr\x+1, \the\numexpr\y) {\large #3};
	}
	\newcommand*{\dipathsimple}[2]{
		\pgfmathsetmacro{\x}{#1}
		\pgfmathsetmacro{\y}{#2}
		\pgfmathsetmacro{\xRectangleBottom}{#1-0.5}

		\draw[fill,opacity=0.5,color=light_grey,rounded corners=5pt] (\xRectangleBottom,\the\numexpr\y-0.5) rectangle (\the\numexpr\x+2,\the\numexpr\y+0.5); 

		\node[normalNode] (1) at (\the\numexpr\x, \the\numexpr\y) {};
		\node[normalNode] (2) at (\the\numexpr\x+1.5, \the\numexpr\y) {};

		\path[->,color=black]
			(1) edge (2)
		;
	}
	\newcommand*{\dipathmerge}[4]{
		\pgfmathsetmacro{\x}{#1}
		\pgfmathsetmacro{\y}{#2}
		\pgfmathsetmacro{\yRectangleTop}{#2+0.5}
		\pgfmathsetmacro{\yRectangleBottom}{#2-0.5}

		\pgfmathsetmacro{\xRectangleBottom}{#1-0.8}

		\draw[fill,opacity=0.5,color=light_grey,rounded corners=5pt] (\xRectangleBottom,\the\numexpr\y-0.8) rectangle (\the\numexpr\x+7.2,\the\numexpr\y+2.0); 

		\draw[rounded corners, densely dashed, line width=1.5pt, color=darkgray, fill=grey] (\the\numexpr\x,\the\numexpr\yRectangleBottom) rectangle (\the\numexpr\x+2,\the\numexpr\yRectangleTop);
		\draw[rounded corners, densely dashed, line width=1.5pt, color=darkgray, fill=grey] (\the\numexpr\x+3,\the\numexpr\yRectangleBottom) rectangle (\the\numexpr\x+5,\the\numexpr\yRectangleTop);

		\node[normalNode] (1) at (\the\numexpr\x, \the\numexpr\y) {};
		\node[gbNode] (2) at (\the\numexpr\x+2, \the\numexpr\y) {};
		\node[labelNode] (8) at (\the\numexpr\x+1, \the\numexpr\y) {\large #3};

		\node[normalNode] (3) at (\the\numexpr\x+3, \the\numexpr\y) {};
		\node[normalNode] (4) at (\the\numexpr\x+5, \the\numexpr\y) {};
		\node[labelNode] (5) at (\the\numexpr\x+4, \the\numexpr\y) {\large #4};

		\node[normalNode] (6) at (\the\numexpr\x+6.5, \the\numexpr\y) {};
		\node[normalNode] (7) at (\the\numexpr\x+5, \the\numexpr\y+1.5) {};

		\path[->,color=black]
			(2) edge (3)
			(4) edge (6)
			(7) edge[bend left=25] (6)
		;
	}
	\newcommand*{\dipathmergehalfsimple}[4]{
		\pgfmathsetmacro{\x}{#1}
		\pgfmathsetmacro{\y}{#2}
		\pgfmathsetmacro{\yRectangleTop}{#2+0.5}
		\pgfmathsetmacro{\yRectangleBottom}{#2-0.5}

		\pgfmathsetmacro{\xRectangleBottom}{#1-0.8}

		\draw[fill,opacity=0.5,color=light_grey,rounded corners=5pt] (\xRectangleBottom,\the\numexpr\y-0.8) rectangle (\the\numexpr\x+6.8,\the\numexpr\y+2.0); 

		\draw[rounded corners, densely dashed, line width=1.5pt, color=darkgray, fill=grey] (\the\numexpr\x,\the\numexpr\yRectangleBottom) rectangle (\the\numexpr\x+2,\the\numexpr\yRectangleTop);

		\node[normalNode] (1) at (\the\numexpr\x, \the\numexpr\y) {};
		\node[gbNode] (2) at (\the\numexpr\x+2, \the\numexpr\y) {};
		\node[labelNode] (8) at (\the\numexpr\x+1, \the\numexpr\y) {\large #3};

		\node[normalNode] (3) at (\the\numexpr\x+3, \the\numexpr\y) {};
		\node[normalNode] (4) at (\the\numexpr\x+4.5, \the\numexpr\y) {};

		\node[normalNode] (6) at (\the\numexpr\x+6, \the\numexpr\y) {};
		\node[normalNode] (7) at (\the\numexpr\x+4.5, \the\numexpr\y+1.5) {};

		\path[->,color=black]
			(2) edge (3)
			(3) edge (4)
			(4) edge (6)
			(7) edge[bend left=25] (6)
		;
	}
	\newcommand*{\dipathmergesimple}[2]{
		\pgfmathsetmacro{\x}{#1}
		\pgfmathsetmacro{\y}{#2}
		\pgfmathsetmacro{\xRectangleBottom}{#1-0.8}

		\draw[fill,opacity=0.5,color=light_grey,rounded corners=5pt] (\xRectangleBottom,\the\numexpr\y-0.6) rectangle (\the\numexpr\x+3.9,\the\numexpr\y+2.0); 

		\node[normalNode] (1) at (\the\numexpr\x, \the\numexpr\y) {};
		\node[gbNode] (2) at (\the\numexpr\x+1.5, \the\numexpr\y) {};

		\node[normalNode] (4) at (\the\numexpr\x+3, \the\numexpr\y) {};
		\node[normalNode] (5) at (\the\numexpr\x+1.5, \the\numexpr\y+1.5) {};

		\path[->,color=black]
			(1) edge (2)
			(2) edge (4)
			(5) edge[bend left=25] (4)
		;
	}
	\newcommand*{\firstgbright}[2]{
		\pgfmathsetmacro{\x}{#1}
		\pgfmathsetmacro{\y}{#2}
		\pgfmathsetmacro{\yRectangleTop}{#2+0.5}
		\pgfmathsetmacro{\yRectangleBottom}{#2-0.5}
		\pgfmathsetmacro{\xRectangleBottom}{#1-0.8}

		\draw[fill,opacity=0.5,color=light_grey,rounded corners=5pt] (\xRectangleBottom,\the\numexpr\y-0.8) rectangle (\the\numexpr\x+8,\the\numexpr\y+2.0); 

		\draw[rounded corners, densely dashed, line width=1.5pt, color=darkgray, fill=grey] (\the\numexpr\x,\the\numexpr\yRectangleBottom) rectangle (\the\numexpr\x+2,\the\numexpr\yRectangleTop);
		\draw[rounded corners, densely dashed, line width=1.5pt, color=darkgray, fill=grey] (\the\numexpr\x+5,\the\numexpr\yRectangleBottom) rectangle (\the\numexpr\x+7,\the\numexpr\yRectangleTop);

		\node[normalNode] (1) at (\the\numexpr\x, \the\numexpr\y) {};
		\node[normalNode] (2) at (\the\numexpr\x+2, \the\numexpr\y) {};
		\node[labelNode] (3) at (\the\numexpr\x+1, \the\numexpr\y) {\large 14};

		\node[gbNode] (4) at (\the\numexpr\x+3.5, \the\numexpr\y) {};
		\node[normalNode] (5) at (\the\numexpr\x+2, \the\numexpr\y+1.5) {};

		\node[normalNode] (6) at (\the\numexpr\x+5, \the\numexpr\y) {};
		\node[normalNode] (7) at (\the\numexpr\x+7, \the\numexpr\y) {};
		\node[labelNode] (8) at (\the\numexpr\x+6, \the\numexpr\y) {\large 15};

		\path[->,color=black]
			(2) edge (4)
			(5) edge[bend left=25] (4)
			(4) edge (6)
		;
	}
	\newcommand*{\querybox}[4]{
		\pgfmathsetmacro{\x}{#1}
		\pgfmathsetmacro{\y}{#2}
		\pgfmathsetmacro{\additional}{#3}
		\pgfmathsetmacro{\depth}{#4}

		\node[color=ibm_purple,rectangle] (b) at (\x,\the\numexpr\y-0.4) {\textbf{+\additional\  = \depth}};
	}
	\newcommand*{\normalbranch}[4]{
		\draw [-,darkgray,line cap=round](#1,#2) -- (#3,#4);
	}

	\draw[fill,opacity=0.5,color=light_grey,rounded corners=5pt] (1,2) rectangle (5,4); 
	\node[draw=none,normalNode,opacity=0,text opacity=1] (b) at (3,3) {\Large $G$};

	\dipathmerge{-10.5}{-2}{9}{7}
	\dipath{-11}{-5}{9}
	\dipathmergehalfsimple{-6.5}{-6}{5}{2}
	\dipath{-8}{-9}{5}
	\dipathmergesimple{-3}{-10}
	\firstgbright{10.5}{-2}
	\dipath{17}{-5}{15}
	\dipathmerge{7}{-6}{8}{6}
	\dipath{6.5}{-9}{8}
	\dipathmergehalfsimple{11}{-10}{4}{2}
	\dipath{9.5}{-13}{4}
	\dipathmergesimple{14.5}{-14}

	\normalbranch{3}{2}{-7}{0.3};
	\normalbranch{3}{2}{13.5}{0.3};

	\normalbranch{-7.5}{-3}{-5.5}{-4};
	\normalbranch{-7.5}{-3}{-9.5}{-4};

	\normalbranch{-4.5}{-7}{-2.5}{-8};
	\normalbranch{-4.5}{-7}{-6.5}{-8};

	\normalbranch{14}{-3}{11}{-4};
	\normalbranch{14}{-3}{17}{-4};

	\normalbranch{10}{-7}{12}{-8};
	\normalbranch{10}{-7}{8}{-8};

	\normalbranch{13}{-11}{15}{-12};
	\normalbranch{13}{-11}{11}{-12};

	\querybox{-10}{-6}{4}{6}
	\querybox{-7}{-10}{3}{6}
	\querybox{-1.5}{-11}{2}{5}
	\querybox{7.5}{-10}{3}{6}
	\querybox{10.5}{-14}{2}{6}
	\querybox{16}{-15}{2}{6}
	\querybox{18}{-6}{4}{6}

\end{tikzpicture}
         \end{array}$
      }
   \end{subfigure}

   \begin{subfigure}[c]{\textwidth}
      \centering
      (c)\hspace{0.5cm}
      \scalebox{0.5}{
         $\begin{array}{l}
            \input{figures/golden_graph_tree.tex}
         \end{array}$
      }
   \end{subfigure}
   \caption{(a) A graph of size 49 where \golden\ uses one more query than \gb. Dashed boxes represent directed paths.  (b) The strategy tree for \gb. (c) The strategy tree for \golden.}
   \label{fig:gb_better_than_golden}
\end{figure}

\subsection{Results for \golden\ on binary DAGs}

This subsection lists the main results about the complexity analysis of \golden. First, note that Theorem~\ref{theo:pathological} also holds for \golden, so the general case (i.e, whenever the DAGs are not necessarily binary) is as bad as \gb.

As for binary DAGs, we establish that the \golden\ algorithm has a better upper bound for the number of queries, in comparison with \gb.

\begin{theorem}
   On any binary DAG with $n$ vertices, the number of \golden\ queries is at most $\log_\phi(n) + 1 = \frac{\log_2(n)}{\log_2(\phi)}+1$, where $\phi$ is the golden ratio.
   \label{theo:bound_golden_binary_dag}
\end{theorem}

\noindent\textit{\textbf{Proof idea}} The \golden\ algorithm has the remarkable following property:
starting from a graph with $n$ vertices, either the subgraph remaining after one query is of size at most $\frac{n}{\phi}$, or the subgraph obtained after two queries is of size at most $\frac{n}{\phi^2}$. If we admit this point, the proof of Theorem~\ref{theo:bound_golden_binary_dag} has no difficulty.

The reason why we have such a guarantee on the size of the remaining graph after one or two queries comes from the choices of the sets $B^{\geq}$ and $B^<$. If \golden\ first queries a bugged vertex of $B^{\geq}$ with a ``bad" score, then the parents of this vertex must have a ``really good" score in the new resulting graph.

Let us take a critical example: \golden\ queries a bugged vertex of $B^{\geq}$, let us say $\vertex{q}$, with $\score(\vertex q)=(n-1)/3$ --- which is the worst possible score for such vertices, by Lemma~\ref{lem:score_B_sets}. In this case, each of the two parents of $\vertex q$ has the really good score of $(n-1)/3$ in the new graph, which is approximately half of its size. So, even if the first query just removes one third of the vertices, the size of the graph after two queries is more or less $n/3$ (which is smaller than $n / \phi^2$).
Similar arguments hold whenever the first query concerns a vertex of $B^<$.

The ratio $1/\phi$ appears in fact whenever we try to balance what could go wrong after one query and what could go wrong after two queries.

\medskip

As first corollary, \golden\ is a better approximation algorithm than \gb\ (in the binary case):

\begin{corollary}
   For every $\varepsilon > 0$, \golden\ is a $\left(\frac{1}{\log_2(\phi)} + \varepsilon\right)$-approximation algorithm on binary DAGs with a sufficiently large size.
   \label{cor:approx_golden}
\end{corollary}

This also gives an upper bound for the optimal number of queries in the worst-case scenario, using the fact that no power of $\phi$ is an integer and thus that $\lfloor \log_\phi(n) + 1 \rfloor = \lceil \log_\phi(n) \rceil$.

\begin{corollary}
   For any binary DAG $D$ with $n$ vertices, the optimal number $\textnormal{opt}$ of queries in the worst-case scenario satisfies
   \[ \lceil \log_2(n) \rceil \leq \textnormal{opt} \leq   \lceil \log_\phi(n) \rceil .\]
   \label{cor:bounds_opt}
\end{corollary}

\vspace*{-.7cm}

Note that the latter corollary is an analogue of Proposition~\ref{prop:bounding_nb_queries}, but for binary DAGs. The lower bound is satisfied for a large variety of DAGs, the most obvious ones being the directed paths. As for the upper bound, there is a 4-vertices graph, commonly named claw (see Figure
\ref{fig:claw_graph}), that uses $\lceil\log_{\phi}(4)\rceil = 3$ queries in the worst-case scenario.

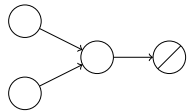
\begin{figure}[ht]
   \centering
   \scalebox{0.8}{
      \begin{tikzpicture}
    [scale=.6,auto=left,every node/.style={circle,draw,scale=0.8},text opacity=0]

    \node (1) at (0,2) {1};
    \node (2) at (0,0) {2};
    \node (3) at (2,1) {3};
    \node[forbidden sign] (4) at (4,1) {4};

    \path[->] 
        (1) edge (3)
        (2) edge (3)
        (3) edge (4)
    ;

\end{tikzpicture}
   }
   \caption{Claw graph.}
   \label{fig:claw_graph}
\end{figure}

\subsection{Proof of the upper bound}

We prove here the upper bound for the number of queries used by \golden\ for binary graphs.

Recall that $\phi = \dfrac{1+\sqrt{5}}{2}$ is the golden ratio.
We also have $1+\phi - \phi^2 = 0$, and thus $n - \frac{n}{\phi} = \frac{n}{\phi^2}$.

\begin{lemma}
   For any binary DAG with $n \geq 14$ vertices,
   \begin{enumerate}[(i)]
      \item either the \golden\ reduces the searching area to at most $\frac{n}{\phi}$ in one query, \label{item_1_lem_up}
      \item or it reduces the searching area to at most $\frac{n}{\phi^2}$ in two queries.
      \label{item_2_lem_up}
   \end{enumerate}
   \label{lem:upper_bound_golden}
\end{lemma}

Note that the lemma does not hold for $n=13$, as shown by Figure~\ref{fig:backbone}. Here,  the digraph after $1$ \golden\ step has $9$ vertices, which is larger than $\frac{13}{\phi} \approx 8.03$, and after $2$ \golden\ steps, it has $5$ vertices, which is larger than $\frac{13}{\phi^2} \approx 4.96$.

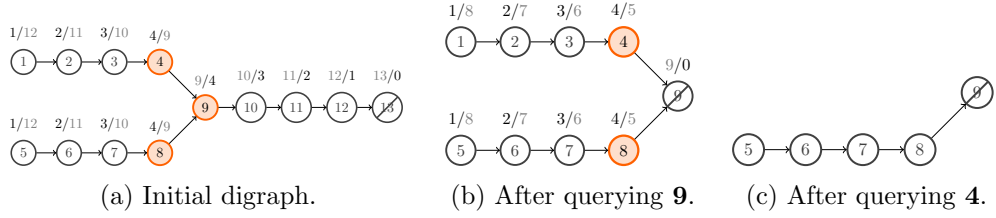
\begin{figure}[ht]
   \centering
   \begin{subfigure}[t]{.42\textwidth}
      \centering
      \resizebox{\textwidth}{!}{
         \begin{tikzpicture}
	[
		scale=.5,
		every node/.style={
			circle,
			draw,
			scale=0.8,
			text=black,
			line width=1.2pt,
			opacity=.2,
			text opacity=1,
			draw opacity=1
		},
		queryNode/.style={fill=ibm_orange, color=ibm_orange,text=black, line width=1.2pt},
		normalNode/.style={text opacity=1,opacity=0,draw opacity=1,color=darkgray,text=darkgray, line width=1.pt}
	]

	\node[normalNode,label={[label distance=-0.2cm]90:1/\textcolor{gray}{12}}] (x0) at (-8, 2) {$1$};
	\node[normalNode,label={[label distance=-0.2cm]90:2/\textcolor{gray}{11}}] (x1) at (-6, 2) {$2$};
	\node[normalNode,label={[label distance=-0.2cm]90:3/\textcolor{gray}{10}}] (x2) at (-4, 2) {$3$};
	\node[queryNode,label={[label distance=-0.2cm]90:4/\textcolor{gray}{9}}] (x3) at (-2, 2) {$4$};
	\node[normalNode,label={[label distance=-0.2cm]90:1/\textcolor{gray}{12}}] (y0) at (-8, -2) {$5$};
	\node[normalNode,label={[label distance=-0.2cm]90:2/\textcolor{gray}{11}}] (y1) at (-6, -2) {$6$};
	\node[normalNode,label={[label distance=-0.2cm]90:3/\textcolor{gray}{10}}] (y2) at (-4, -2) {$7$};
	\node[queryNode,label={[label distance=-0.2cm]90:4/\textcolor{gray}{9}}] (y3) at (-2, -2) {$8$};
	\node[queryNode,label={[label distance=-0.2cm]90:\textcolor{gray}{9}/4}] (c) at (0, 0) {$9$};
	\node[normalNode,label={[label distance=-0.2cm]90:\textcolor{gray}{10}/3}] (z1) at (2, 0) {$10$};
	\node[normalNode,label={[label distance=-0.2cm]90:\textcolor{gray}{11}/2}] (z2) at (4, 0) {$11$};
	\node[normalNode,label={[label distance=-0.2cm]90:\textcolor{gray}{12}/1}] (z3) at (6, 0) {$12$};
	\node[normalNode,label={[label distance=-0.2cm]90:\textcolor{gray}{13}/0},forbidden sign] (z4) at (8, 0) {$13$};

	\path[->,color=black]
		(x0) edge (x1)
		(x1) edge (x2)
		(x2) edge (x3)
		(y0) edge (y1)
		(y1) edge (y2)
		(y2) edge (y3)
		(x3) edge (c)
		(y3) edge (c)
		(c) edge (z1)
		(z1) edge (z2)
		(z2) edge (z3)
		(z3) edge (z4)
	;

\end{tikzpicture}
      }
      \caption{Initial digraph.}
   \end{subfigure}
   \hfill
   \begin{subfigure}[t]{.28\textwidth}
      \centering
      \resizebox{\textwidth}{!}{
         \begin{tikzpicture}
	[
		scale=.5,
		every node/.style={
			circle,
			draw,
			scale=0.8,
			text=black,
			line width=1.2pt,
			opacity=.2,
			text opacity=1,
			draw opacity=1
		},
		queryNode/.style={fill=ibm_orange, color=ibm_orange,text=black, line width=1.2pt},
		normalNode/.style={text opacity=1,opacity=0,draw opacity=1,color=darkgray,text=darkgray, line width=1.pt}
	]

	\node[normalNode,label={[label distance=-0.2cm]90:1/\textcolor{gray}{8}}] (x0) at (-8, 2) {$1$};
	\node[normalNode,label={[label distance=-0.2cm]90:2/\textcolor{gray}{7}}] (x1) at (-6, 2) {$2$};
	\node[normalNode,label={[label distance=-0.2cm]90:3/\textcolor{gray}{6}}] (x2) at (-4, 2) {$3$};
	\node[queryNode,label={[label distance=-0.2cm]90:4/\textcolor{gray}{5}}] (x3) at (-2, 2) {$4$};
	\node[normalNode,label={[label distance=-0.2cm]90:1/\textcolor{gray}{8}}] (y0) at (-8, -2) {$5$};
	\node[normalNode,label={[label distance=-0.2cm]90:2/\textcolor{gray}{7}}] (y1) at (-6, -2) {$6$};
	\node[normalNode,label={[label distance=-0.2cm]90:3/\textcolor{gray}{6}}] (y2) at (-4, -2) {$7$};
	\node[queryNode,label={[label distance=-0.2cm]90:4/\textcolor{gray}{5}}] (y3) at (-2, -2) {$8$};
	\node[normalNode,label={[label distance=-0.2cm]90:\textcolor{gray}{9}/0},forbidden sign] (c) at (0, 0) {$9$};

	\path[->,color=black]
		(x0) edge (x1)
		(x1) edge (x2)
		(x2) edge (x3)
		(y0) edge (y1)
		(y1) edge (y2)
		(y2) edge (y3)
		(x3) edge (c)
		(y3) edge (c)
	;

\end{tikzpicture}
      }
      \caption{After querying $\vertex{9}$.}
   \end{subfigure}
   \hfill
   \begin{subfigure}[t]{.28\textwidth}
      \centering
      \resizebox{\textwidth}{!}{
         \begin{tikzpicture}
	[
		scale=.5,
		every node/.style={
			circle,
			draw,
			scale=0.8,
			text=black,
			line width=1.2pt,
			opacity=.2,
			text opacity=1,
			draw opacity=1},
		normalNode/.style={text opacity=1,opacity=0,draw opacity=1,color=darkgray,text=darkgray, line width=1.pt}
	]

	\node[normalNode] (y0) at (-8, -2) {$5$};
	\node[normalNode] (y1) at (-6, -2) {$6$};
	\node[normalNode] (y2) at (-4, -2) {$7$};
	\node[normalNode] (y3) at (-2, -2) {$8$};
	\node[normalNode,forbidden sign] (c) at (0, 0) {$9$};

	\path[->,color=black]
		(y0) edge (y1)
		(y1) edge (y2)
		(y2) edge (y3)
		(y3) edge (c)
	;

\end{tikzpicture}
      }
      \caption{After querying $\vertex 4$.}
   \end{subfigure}
   \caption{First two steps of \golden\  for a DAG of size $13$.}
   \label{fig:backbone}
\end{figure}

\begin{proof}
   If there exists a vertex $\vertex v$ in $D$ with $\score (\vertex v) \geq \frac n {\phi^2}$, then item~\eqref{item_1_lem_up} holds since \golden\ will query such a vertex via Step~\ref{item:one_move}, and then there will remain at most $n - \frac n {\phi^2} = \frac n \phi$ vertices.

   Thus we can assume that all vertices have a score smaller than $\frac n {\phi^2}$. Note that under this assumption, every vertex of $B^\geq$ has exactly two parents.
   Indeed, if a vertex $\vertex z \in B^\geq$ has only one parent, say $\vertex x$, then
   \[\score(\vertex x) = \ancestors x = \ancestors z - 1 \geq \frac n 2 - 1 \geq \frac n {\phi^2}\]
   (the last inequality is true whenever $n \geq 9$).

   Let us start by proving a useful claim.

   \begin{claim} 
      Suppose that after the first query in $D$, the resulting digraph, say $D'$, satisfies the following two properties :
      \begin{itemize}
         \item $|D'|$, the number of vertices of $D'$, is greater than $\frac {n} {\phi}$;
         \item there exists a vertex $\vertex v$ in $D'$ with no more than $\frac{n} {\phi^2}$ ancestors in $D'$, and no more than $\frac{n} {\phi^2}$ non-ancestors in $D'$.
      \end{itemize}
      Then the score of $\vertex v$ in $D'$ is greater or equal than $\frac{|D'|}{\phi^2}$, and item~\ref{item_2_lem_up} of the lemma holds.
      \label{claim:defi}
   \end{claim}

   \begin{proof} 
      Let $\#a$ and $\#na$ be respectively the number of ancestors and non-ancestors of $\vertex v$ in $D'$, with $\#a \leq \frac{n}{\phi^2}$ and $\#na \leq \frac{n}{\phi^2}$.

      First notice that the score of $\vertex v$ in $D'$ is the minimum between $\#a$ and $\#na$. So, if we show that both $\#a$ and $\#na$ are no less than $\frac{|D'|}{\phi^2}$, the first part of the claim is proved.

      Thus
      \[\#a =  |D'| - \#na \geq |D'| - \frac{n}{\phi^2} = \frac{|D'|}{\phi^2} + \frac 1 {\phi} \left(  |D'| - \frac n \phi \right) \]
      (the last equality can be derived from the identity $1 = \frac 1 \phi + \frac 1 {\phi^2}$). But since $|D'| > \frac n \phi$  by hypothesis, we  deduce that
      $\#a \geq \frac{|D'|}{\phi^2}$.
      The numbers $\#a$ and $\#na$ play symmetric roles in this claim, so we can similarly infer that $\#na \geq  \frac{|D'|}{\phi^2}$. Thereby we have proved that $\score_{D'}(\vertex v)\geq  \frac{|D'|}{\phi^2}$.

      As for the second part of the claim, \golden\ will run Step~\ref{item:one_move} and query $\vertex v$ or a vertex with a larger score. The searching area is thus reduced to at most $\frac{n}{\phi^2}$ vertices.
   \end{proof}

   Let $\vertex z$ be the first vertex queried by \golden. Note that $\vertex z$ belongs to $B^\ge$ or $B^<$. In either case, we are going to show that if (\ref{item_1_lem_up}) fails, then the hypotheses of the above claim are satisfied,
   and consequently (\ref{item_2_lem_up}) holds.

   \begin{figure}[ht]
      \centering
      \begin{subfigure}[t]{.45\textwidth}
         \centering
         \scalebox{0.75}{
            \begin{tikzpicture}
	[scale=.6, every node/.style={circle,draw,scale=0.8, text=black, text opacity=1, draw opacity=1},
	minimum size=0.75cm,
	]

	\pgfdeclarepatternformonly{south west lines}{\pgfqpoint{-0pt}{-0pt}}{\pgfqpoint{3pt}{3pt}}{\pgfqpoint{3pt}{3pt}}{
		\pgfsetlinewidth{0.4pt}
		\pgfpathmoveto{\pgfqpoint{0pt}{0pt}}
		\pgfpathlineto{\pgfqpoint{3pt}{3pt}}
		\pgfpathmoveto{\pgfqpoint{2.8pt}{-.2pt}}
		\pgfpathlineto{\pgfqpoint{3.2pt}{.2pt}}
		\pgfpathmoveto{\pgfqpoint{-.2pt}{2.8pt}}
		\pgfpathlineto{\pgfqpoint{.2pt}{3.2pt}}
		\pgfusepath{stroke}
	}

	% 2
	\draw[pattern=south west lines, pattern color=gray] (4.8,1.5) ellipse (2.3cm and 1.8cm);
	\draw[pattern=south west lines, pattern color=gray] (4.8,-1.5) ellipse (2.3cm and 1.8cm);
	
	\node[] (z) at (9,0) {\LARGE$z$};
	\node[fill=white] (x) at (7, 1.5) {\LARGE$x$};
	\node[fill=white] (y) at (7, -1.5) {\LARGE$y$};
	
	\path[->,color=black]
	(x) edge (z)
	(y) edge (z)
	;

\end{tikzpicture}
         }
         \caption{First queried vertex $\vertex z$ is in $B^\geq$, and $\vertex z$ is bugged.}
      \end{subfigure}
      \begin{subfigure}[t]{.45\textwidth}
         \centering
         \scalebox{0.7}{
            \begin{tikzpicture}
	[scale=.6, every node/.style={circle,draw,scale=0.8, text=black, text opacity=1, draw opacity=1},
	minimum size=1cm,
	]

	\begin{scope}
		\clip (-7,-10.5)rectangle(-2,-7);
		\draw[pattern=south west lines, pattern color=gray] (-4.5,-8.5) ellipse (2.5cm and 1.7cm);
	\end{scope}
	\draw[pattern=south west lines, pattern color=gray,opacity=0.5] (-4.5,-5.5) ellipse (2.5cm and 1.7cm);
	\draw[pattern=south west lines, pattern color=gray] (2.5,-7) ellipse (2.5cm and 1.7cm);
	
	\node[fill=white] (z) at (0,-7) {\LARGE$c$};
	\node[fill=white,draw opacity=0.5,text opacity=0.5] (x) at (-2, -5.5) {\LARGE$z$};
	\node[fill=white] (y) at (-2, -8.5) {\LARGE$z'$};

	\draw[color=black,opacity=0.7] (-5.67,-7) arc
    [
        start angle=242,
        end angle=298,
        x radius=2.5cm,
        y radius =1.7cm
    ] ;

	\path[->,color=black]
		(x) edge (z)
		(y) edge (z)
	;

\end{tikzpicture}
         }
         \captionsetup{width=.9\linewidth}
         \caption{First queried vertex $\vertex z$ is in $B^<$, and $\vertex z$ is clean.}
      \end{subfigure}
      \caption{Possible digraphs after the first  query of \golden.}
      \label{fig:golden_cases}
   \end{figure}
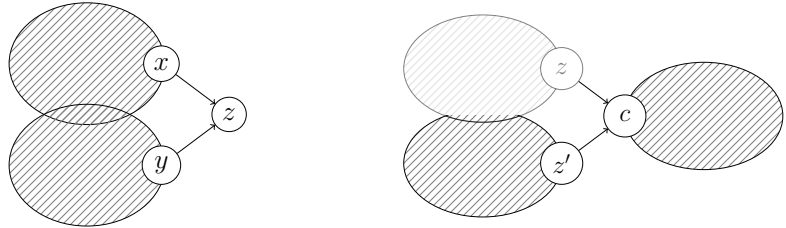

   \noindent\textbf{Case 1: $\vertex z \in B^\ge$.}
   Since $\vertex z \in V^\ge$ by hypothesis, its score corresponds to the number of non-ancestors, thus $\score_D(\vertex z) = n - \ancestors z \le \frac n {\phi^2}$.
   If $\vertex z$ is clean, only the non-ancestors of $\vertex z$ remain after one step of \golden, which is fewer than $\frac n {\phi^2}$ vertices, and \eqref{item_1_lem_up} holds.

   Suppose now that $\vertex z$ is bugged. Let $D'$ denote the DAG obtained from $D$ after querying $\vertex z$ (constisting in only the ancestors of $\vertex z$, which becomes the new marked sink, see Figure~\ref{fig:golden_cases}(a)).
Note that $D'$ has $\ancestors z= n - \score_D(\vertex z)$ vertices, which is greater than $n - \frac n {\phi^2} =  \frac n \phi$. 
   
   Let $\vertex x$ and $\vertex y$ be the parents of $\vertex z$, and assume that $\ancestors{x} \geq \ancestors{y}$.
   Since $\ancestors x = \score_D(\vertex x) \leq \frac{n}{\phi^2}$, vertex $\vertex x$ has at most $\frac{n}{\phi^2}$ ancestors in $D'$. Moreover, vertex $\vertex x$ has $\ancestors z - \ancestors x$ non-ancestors in $D'$.
   But since a non-ancestor of $\vertex x$ in $D'$ is either $\vertex z$ or an ancestor of $\vertex y$, we have
   \[\ancestors z - \ancestors x \leq \ancestors y + 1 \]
   and because
   $\ancestors y \leq  \ancestors x \leq \score_D(\vertex z) = n - \ancestors z$, we deduce
   \[\ancestors z - \ancestors x \leq \ancestors x + 1 \leq n - \ancestors z + 1,\]
   hence
   \[3 (\ancestors z - \ancestors x)  \leq (\ancestors z - \ancestors x) + (\ancestors x + 1) + (n - \ancestors z + 1) = n + 2. \]
   Thus $\vertex x$ has at most $\frac{n + 2} 3$ non-ancestors, which is smaller than $\frac n {\phi^2}$ whenever $n \geq 14$. The hypotheses of Claim~\ref{claim:defi} hold,
   which concludes Case~1.

   \noindent\textbf{Case 2: $\vertex z \in B^<$.}
   Since $\vertex z$ belongs to $B^<$, it has a child in $B^\ge$, denoted by $\vertex c$.
   By assumption, $\vertex c$ has $2$ parents, let $\vertex{z'}$ be the other parent of $\vertex c$ (also in $B^<$).

   If $\vertex z$ is bugged, then there remain at most $\score(\vertex z)=\ancestors{z} < \frac n {\phi^2}$ vertices, which makes \eqref{item_1_lem_up} true.
   So assume that the queried vertex $\vertex z$ is clean, and after one step of $\golden$, we end up with a new DAG $D'$, obtained from $D$ by removing all ancestors of $\vertex z$ (see Figure~\ref{fig:golden_cases}(b)). Note that $D'$ has $n-\ancestors{z}$ vertices, which is greater than $n - \score_D(\vertex z) >\frac n \phi$.

   Note that the non ancestors of $\vertex c$ in $D'$ are exactly the non ancestors of $\vertex c$ in $D$, and thus $\vertex c$ has no more than $\frac{n}{\phi^2}$ non ancestors in $D'$.
   Thus if $\vertex {c}$ has no more than $\frac n {\phi^2}$ ancestors in $D'$, $\vertex c$
   satisfies the hypotheses of Claim~\ref{claim:defi} and the result holds.

   Assume that $\vertex{c}$ has more than $\frac n {\phi^2}$ ancestors in $D'$, which means that
   \[ \ancestors{c} - \ancestors z \geq \frac n {\phi^2} \geq n - \ancestors{c}.\]
   Since  $\frac n {\phi^2}$ cannot be an integer, $\ancestors{c} - \ancestors z \geq  n - \ancestors{c} + 1$.

   Moreover because $n - \ancestors{c} =  \score_D(\vertex  c)  \leq  \score_D(\vertex z)  = \ancestors z $, we have
   \[ 3 (n - \ancestors{c} + 1)  \leq  (n - \ancestors{c} + 1) +  (\ancestors z + 1)  + (\ancestors{c} - \ancestors z)  = n + 2.\]

   So the number of non-ancestors of $\vertex {z'}$ in $D'$, that is $n - \ancestors{c} + 1$, is less than $\frac{n+2} 3 < \frac n {\phi^2}$ when $n \geq 14$.
   Moreover, the number of ancestor of $\vertex{z'}$ in $D'$ is no more than its number of ancestors in $D$, which satisfies $|z'| = \score_D(\vertex{z'}) \le \frac{n}{\phi^2}$.
   So $\vertex{z'}$ satisfies the hypotheses of Claim~\ref{claim:defi} and the result holds.
\end{proof}

We can now establish the upper bound for the overall number of \golden\ queries.

\begin{proof}[Proof of Theorem~\ref{theo:bound_golden_binary_dag}]
   We prove by induction on $n$ that for any binary DAG with $n$ vertices, the number of \golden\ queries is at most $\log_\phi(n) + 1$.

   The base case contains all graphs up to  $n = 13$. To prove it, we use Lemma~\ref{lem:score_B_sets} to show that \golden\ eliminates at least $\frac {n-1} 3$ vertices at the first step. So the maximal number of queries for  size $n$ is bounded by one plus the maximal number of queries for size $n - \left\lceil \frac{n - 1} 3\right\rceil$. The first values are given by the sequence $F(n)$ in Table~\ref{table:Fn}.
   We remark that the second row is bounded above by the last row. So the property holds for $n \leq 13$.

   As for $n \geq 14$, the induction is straightforward by Lemma~\ref{lem:upper_bound_golden}. Indeed, if item~\ref{item_1_lem_up} is satisfied, then the number of \golden\ queries is bounded by $1 + \left(\log_\phi\left(\frac n \phi\right) + 1 \right) = 1 + \log_\phi(n)$. If item~\ref{item_2_lem_up} is satisfied, then it is also bounded by $2 + \left(\log_\phi\left(\frac n {\phi^2}\right) + 1 \right) = 1 + \log_\phi(n)$.
\end{proof}

\subsection{Fibonacci trees}

In order to prove the sharpness of the constant $\frac{1}{\log_2(\phi)}$ from Corollary~\ref{cor:approx_golden}, we define a new family of digraphs:
the \emph{Fibonacci trees}.

\begin{definition}[Fibonacci trees]
   For $i \geq 1$, the $i$-th Fibonacci tree $F_i$ is defined as followed.

   $F_1$ is a single vertex, $F_2$ is an arc between two vertices, and for $i\ge 2$, $F_{i+1}$ is a sink with two parents, one being the sink of a tree $F_i$ and the other the sink of a tree $F_{i-1}$.
\end{definition}

Figure~\ref{fig:first_fibonacci_trees} shows the six first Fibonacci trees.
Noting $|F_i|$ the number of vertices of the $i$-th Fibonacci tree $F_i$, we have by construction, that
\[|F_i| = |F_{i-1}| + |F_{i-2}| + 1, \quad |F_1| = 1 \quad \textrm{ and } \quad |F_2| = 2.\]
This recurrence has for solution $|F_i| = fib_{i+2} - 1$, where $fib_i$ is the $i$-th Fibonacci number.

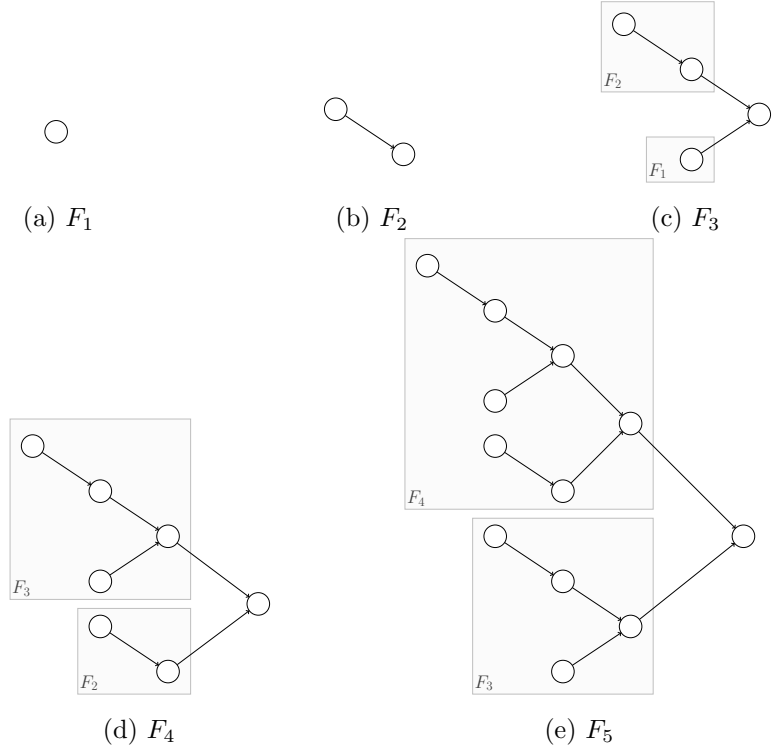
\begin{figure}[ht]
   \centering
   \begin{subfigure}[t]{.3\textwidth}
      \centering
      \scalebox{0.5}{
         \begin{tikzpicture}
    [
        scale=.6,
        auto=left,
        every node/.style={circle,draw,scale=0.8},
        minimum size=0.75cm,
        text opacity=0,
    ]

    \node (51) at (-2,10) {};

    \node[draw=none] (empty2) at (-2,8) {};

\end{tikzpicture}
      }
      \caption{$F_1$}
   \end{subfigure}
   \begin{subfigure}[t]{.3\textwidth}
      \centering
      \scalebox{0.5}{
         \begin{tikzpicture}
    [
        scale=.6,
        auto=left,
        every node/.style={circle,draw,scale=0.8,fill=white},
        minimum size=0.75cm,
        text opacity=0,
    ]

    \node (51) at (-2,10) {};
    \node (52) at (1,8) {};

    \path[->]
        (51) edge (52)
    ;

    \node[draw=none] (empty3) at (-2,7) {};

\end{tikzpicture}
      }
      \caption{$F_2$}
   \end{subfigure}
   \begin{subfigure}[t]{.3\textwidth}
      \centering
      \scalebox{0.5}{
         \begin{tikzpicture}
    [
        scale=.6,
        auto=left,
        every node/.style={circle,draw,scale=0.8,fill=white},
        minimum size=0.75cm,
        text opacity=0,
    ]

    \draw[draw=black,color=lightgray,fill,opacity=0.06,text opacity=1,draw opacity=1] (-3,7) rectangle ++(5,4);
    \draw[draw=black,color=lightgray,fill,opacity=0.06,text opacity=1,draw opacity=1] (-1,3) rectangle ++(3,2);

    \node (51) at (-2,10) {};
    \node (52) at (1,8) {};
    \node (53) at (4,6) {};
    \node (54) at (1,4) {};

    \path[->]
        (51) edge (52)
        (52) edge (53)
        (54) edge (53)
    ;

    \node[draw=none,fill=none,text opacity=1,color=darkgray] (text1) at (-2.5, 7.5) {\LARGE$F_2$};
    \node[draw=none,fill=none,text opacity=1,color=darkgray] (text2) at (-0.5, 3.5) {\LARGE$F_1$};

\end{tikzpicture}
      }
      \caption{$F_3$}
   \end{subfigure}

   \begin{subfigure}[t]{.4\textwidth}
      \centering
      \scalebox{0.5}{
         \begin{tikzpicture}
    [
        scale=.6,
        auto=left,
        every node/.style={circle,draw,scale=0.8,fill=white},
        minimum size=0.75cm,
        text opacity=0,
    ]

    \draw[draw=black,color=lightgray,fill,opacity=0.06,text opacity=1,draw opacity=1] (-3,3.2) rectangle ++(8,8);
    \draw[draw=black,color=lightgray,fill,opacity=0.06,text opacity=1,draw opacity=1] (0,-1) rectangle ++(5,3.8);

    \node (51) at (-2,10) {};
    \node (52) at (1,8) {};
    \node (53) at (4,6) {};
    \node (54) at (1,4) {};
    \node (55) at (8,3) {};
    \node (56) at (4,0) {};
    \node (57) at (1,2) {};

    \path[->]
        (51) edge (52)
        (52) edge (53)
        (54) edge (53)
        (53) edge (55)
        (56) edge (55)
        (57) edge (56)
    ;

    \node[draw=none,fill=none,text opacity=1,color=darkgray] (text1) at (-2.5, 3.7) {\LARGE$F_3$};
    \node[draw=none,fill=none,text opacity=1,color=darkgray] (text2) at (0.5, -0.5) {\LARGE$F_2$};

\end{tikzpicture}
      }
      \caption{$F_4$}
   \end{subfigure}
   \begin{subfigure}[t]{.45\textwidth}
      \centering
      \scalebox{0.5}{
         \begin{tikzpicture}
    [
        scale=.6,
        auto=left,
        every node/.style={circle,draw,scale=0.8,fill=white},
        minimum size=0.75cm,
        text opacity=0,
    ]

    \draw[draw=black,color=lightgray,fill,opacity=0.06,text opacity=1,draw opacity=1] (-3,-0.8) rectangle ++(11,12);
    \draw[draw=black,color=lightgray,fill,opacity=0.06,text opacity=1,draw opacity=1] (0,-9) rectangle ++(8,7.8);

    \node (51) at (-2,10) {};
    \node (52) at (1,8) {};
    \node (53) at (4,6) {};
    \node (54) at (1,4) {};
    \node (55) at (7,3) {};
    \node (56) at (4,0) {};
    \node (57) at (1,2) {};
    \node (58) at (12,-2) {};
    \node (59) at (7,-6) {};
    \node (510) at (4,-4) {};
    \node (511) at (1,-2) {};
    \node (512) at (4,-8) {};

    \path[->]
        (51) edge (52)
        (52) edge (53)
        (54) edge (53)
        (53) edge (55)
        (56) edge (55)
        (57) edge (56)
        (55) edge (58)
        (59) edge (58)
        (510) edge (59)
        (511) edge (510)
        (512) edge (59)
    ;

    \node[draw=none,fill=none,text opacity=1,color=darkgray] (text1) at (-2.5, -0.3) {\LARGE$F_4$};
    \node[draw=none,fill=none,text opacity=1,color=darkgray] (text2) at (0.5, -8.5) {\LARGE$F_3$};

\end{tikzpicture}
      }
      \caption{$F_5$}
   \end{subfigure}
   \caption{First Fibonacci trees.}
      \label{fig:first_fibonacci_trees}
\end{figure}

\medskip

We can establish an optimal strategy for the Fibonacci trees.

\begin{theorem}
   For any $i \geq 1$, the optimal strategy for the $i$-th Fibonacci tree $F_i$ uses $i-1$ queries in the worst-case scenario.
   \label{theo:fibonacci}
\end{theorem}

We decompose this proof in two claims.

\begin{claim}
   For $i \geq 1$, we define $F'_i$ as a sink with one parent which is the sink of an $i$-th Fibonacci tree $F_i$.
   For any $F_i$ and $F'_{i-1}$ with $i \geq 2$, there exists a strategy which finds the \faulty\ in at most $i-1$ queries.
   \label{claim:fibonacci_tree_strategies}
\end{claim}

\begin{proof}
   We prove the claim by induction on $i \geq 2$.
   The base case $i=2$ is obvious since $F_2$ and $F'_1$ are each composed of two vertices including the marked bugged vertex, and that in this case one query is required to discover the \faulty.
   Suppose by induction we have a strategy for $F'_{i-1}$ and $F_{i}$ with $i-1$ queries.

   To find the \faulty\ in $F_{i+1}$, we query the sink of the subtree $F_{i}$.
   If this sink is bugged, then we continue with our strategy on $F_{i}$. If this sink is clean, then by removing this subtree from $F_{i+1}$, we recognize $F'_{i-1}$ and use the corresponding strategy. In both cases, we have performed no more than $1 + (i-1)$ queries.

   To find the \faulty\ in $F'_{i}$, we query the parent of the sink. If it is bugged, then we continue with our strategy on the resulting $F_{i}$. Otherwise, the sink of $F'_{i}$ is the \faulty. This strategy uses at most $1 + (i-1)$ queries, as required.
\end{proof}

\begin{claim} 
   If $T$ is a tree containing two subtrees isomorphic to Fibonacci trees $F_k$ and $F_{k+1}$ and whose sinks are non-ancestors of each other (cf Figure~\ref{fig:fibonacci_proof} top), then, for any strategy searching for the \faulty\ in $T$, there exists a vertex $\vertex v$ in $T$ such that this strategy uses at least $k+1$ queries to identify $\vertex v$ as the \faulty.
   \label{claim:opt_fibonacci}
\end{claim}

\begin{figure}[ht]
   \centering
   \includegraphics[width=0.9\textwidth]{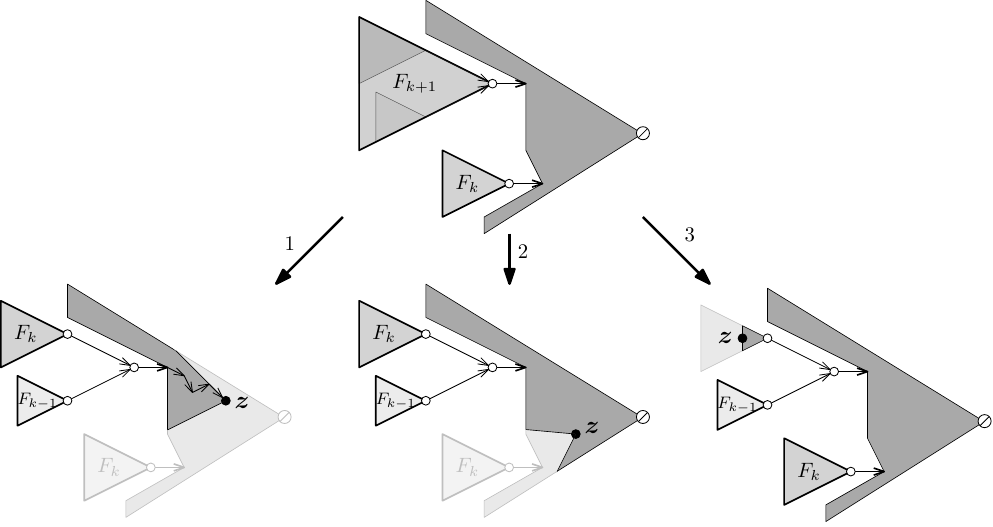}
   \caption{Illustration of the proof of Claim~\ref{claim:opt_fibonacci}.}
   \label{fig:fibonacci_proof}
\end{figure}

\begin{proof}
   We prove this claim by induction on $k$. For $k=1$, $T$ must contain at least 4 vertices, so by Proposition~\ref{prop:bounding_nb_queries}, one must query at least 2 vertices.

   Now suppose the claim statement true for a positive integer $k - 1$, and consider a strategy for the Regression Search Problem on a tree $T$ strictly containing $F_{k+1}$ and $F_{k}$. Let $\vertex z$ be the first query of this strategy. Let us investigate every possibility for $\vertex z$ (the reader can refer to Figure~\ref{fig:fibonacci_proof} for an illustration):
   \begin{enumerate}
      \item \textbf{The root of the subtree $F_{k+1}$ is an ancestor of $\vertex z$.} Then we force the \faulty\ to be an ancestor of $\vertex z$ (i.e. $\vertex z$ is bugged). Then after querying $\vertex z$, there remains all ancestors of $\vertex z$, which contains $F_{k+1}$, which, by definition of Fibonacci trees, strictly contains $F_{k}$ and $F_{k-1}$. By induction hypothesis, we need to query $k$ extra vertices to find the \faulty.
      \item \textbf{The root of the subtree $F_{k+1}$ is not an ancestor of $\vertex z$ and $\vertex z$ is not in the subtree $F_{k+1}$.} Here the \faulty\ will be a non-ancestor of $\vertex z$ (i.e. $\vertex z$ is not bugged). Like in the previous case, the remaining digraph will include $F_{k+1}$, hence copies of $F_k$ and $F_{k-1}$. We then use the induction hypothesis.
      \item  \textbf{$\vertex z$ is in the subtree $F_{k+1}$, but it is not its root. } We set $\vertex z$ to be clean so that the \faulty\ will be among the non-ancestors of $\vertex z$. The subtree $F_{k+1}$ contains two disjoint copies of $F_{k-1}$,
      one of which is in the non-ancestors of $\vertex z$. By hypothesis, $T$ includes also another copy of $F_{k}$. So the query of $\vertex z$ leads to a tree containing $F_k$ and $F_{k-1}$: the induction hypothesis indicates that we need $k$ other queries.
   \end{enumerate}
   For each of these three possibilities, the strategy uses in total $k+1$ queries, which concludes the induction.
\end{proof}

\begin{proof}[Conclusion of the proof of Theorem~\ref{theo:fibonacci}.] 
   The Fibonacci tree $F_{i}$ contains disjoint copies of $F_{i-1}$ and $F_{i-2}$. By Claim~\ref{claim:opt_fibonacci}, any strategy, in particular an optimal one, uses at least $i-1$ queries to find the \faulty\ in $F_i$ in the worst-case scenario. The optimal number of queries is then exactly $i-1$, because  by Claim~\ref{claim:fibonacci_tree_strategies}, there exists a strategy with that many queries.
\end{proof}

The first consequence of Theorem~\ref{theo:fibonacci} is that the upper bound $\lceil \log_{\phi}(n) \rceil$ from Corollary~\ref{cor:bounds_opt} is asymptotically sharp:

\begin{corollary}
   Any optimal strategy solving the Regression Search Problem for any Fibonacci tree of size $n \geq 7$ uses $\lceil \log_{\phi}(n) \rceil - 2$ queries in the worst-case scenario.
   \label{cor:sharpness_fibonacci}
\end{corollary}

The authors do not know if there exist an infinity of graphs for which solving the Regression Search Problem requires exactly $\lceil \log_{\phi}(n) \rceil$ queries.

\begin{proof}[Proof of Corollary~\ref{cor:sharpness_fibonacci}]
   \begin{sloppypar}
   Recall that $|F_i| =  fib_{i+2} - 1$.
   Using that $fib_{i+2} = \left( \phi^{i+2} - (-\phi)^{-i-2} \right)/\sqrt{5}$, we get
   \end{sloppypar}
   \[
   \log_{\phi}(|F_i|) = i + 2 - \log_{\phi}(\sqrt 5) + \varepsilon_i,
   \]
   where
   \[  \varepsilon_i = \log_\phi\left(1 - \frac{\sqrt{5}}{\phi^{i+2}} - \frac {1} {(-\phi^2)^{i+2}}\right) \]
   which increases and tends to $0$, hence  is in absolute value no more than $|\varepsilon_4| \leq 0.3$ whenever $i \geq 4$. If $i \geq 4$, we have $\lceil \log_\phi(|F_i|) \rceil = i + 2 - \lfloor \log_{\phi}(\sqrt 5) + 0.3 \rfloor = i + 1$. We conclude by Theorem~\ref{theo:fibonacci}.
\end{proof}

The previous corollary demonstrates that the Fibonacci trees are inherently flawed for the Regression Search Problem. They are the less pathological analogues of octopuses, but in the context of binary DAGs.

Finally, we show that $\frac{1}{\log_2(\phi)}$ is the good approximation ratio for \golden.

\begin{corollary}
   For $\varepsilon > 0$, \golden\ is not a $\left( \dfrac{1}{\log_2(\phi)} - \varepsilon \right)$ approximation algorithm.
\end{corollary}

\begin{proof} 
   The idea is to add a comb (see Definition~\ref{def:comb}) to the $i$-th Fibonacci tree $F_i$ to approach the $\frac{1}{\log_2(\phi)}$ ratio.

   Indeed, Theorem~\ref{theo:comb}  still holds if we replace  ``\gb" by ``\golden" since both algorithms will query first $\vertex{v_n}$, which is the only vertex with the maximal score $n$.
   Thus we apply this theorem for all $i$ such that $|F_i|$ is odd, i.e., whenever $i$ is congruent to $1$ modulo $3$.

   We deduce that $comb(F_i)$ is a binary DAG for which:
   \begin{itemize}
      \item the number of \golden\ queries is $\lceil \log_\phi(|F_i|) \rceil - 1$ (see Corollary~\ref{cor:sharpness_fibonacci}),
      \item the optimal number of queries is $\lceil \log_2(|F_i|) \rceil +1$.
   \end{itemize}
   The ratio of these two numbers makes a number tending to $\frac{1}{\log_2(\phi)}$, whenever $i$ goes to $+\infty$. This is why \golden\ cannot be a $\left( \frac{1}{\log_2(\phi)} - \varepsilon \right)$ approximation algorithm, for any $\varepsilon > 0$.
\end{proof}

\section{Is the binary case NP-complete?}
\label{sec:np_complete}

Though \gb\ and \golden\  are non-optimal algorithms, the Regression Search Problem (\textsc{Rsp} -- see Definition~\ref{def:regression_search_problem}) on binary DAGs is not proved to be NP-hard. We still do not know whether it is the case, but we show the NP-completeness of a new related problem, the Confined Regression Search Problem (\textsc{Crsp}). It is a reformulation of \textsc{Rsp} in the general case but not in the binary case.

In the Confined Regression Search Problem, we consider a DAG $D$ with additional information on some vertices. A vertex is said to be \emph{innocent} if it is not the faulty vertex, i.e., if it is not the one that introduced the bug. It is still possible to query an innocent vertex since it can be bugged or clean.
\textsc{Crsp} consists in searching the faulty vertex given a possibly empty set $I$ of innocent vertices. Conversely, it confines the faulty vertex to be in the complementary set of $I$. The DAG does not necessarily have any bugged vertex any more.

The decision version of the Confined Regression Search Problem is formally defined as follows.

\begin{definition} \textbf{Confined Regression Search Problem}  \\
   \textbf{Input.} A DAG $D$, a subset $I$ of innocent vertices, and an integer $k$.

   \noindent\textbf{Output.} Whether there is a strategy that determines in at most $k$ queries in the worst-case scenario whether $D$ has a bugged vertex, and if it is the case, which one is the faulty vertex.
   \label{def:crsp}
\end{definition}

Figure~\ref{fig:reductions_crsp_rsp} illustrates the equivalence between instances of \textsc{Crsp} and \textsc{Rsp} in the general case. The reductions work with the following transformations.
\begin{itemize}
   \item from \textsc{Crsp} to \textsc{Rsp}: create a bugged vertex $b$ and add arcs from all non-innocent vertices to $b$.
   \item from \textsc{Rsp} to \textsc{Crsp}: delete the bugged vertex $b$ and all its descendants. Set as innocent all vertices which were not ancestors of $b$.
\end{itemize}

The strategies are preserved in both reductions, but with the following change: if the DAG in \textsc{Crsp} has no bugged vertex then the marked vertex $b$ is the faulty vertex in \textsc{Rsp} and \textit{vice versa}. However, note that though the reduction from \textsc{Rsp} to \textsc{Crsp} preserves the indegree of the DAG, the reduction from \textsc{Crsp} to \textsc{Rsp} creates a vertex with a large indegree.

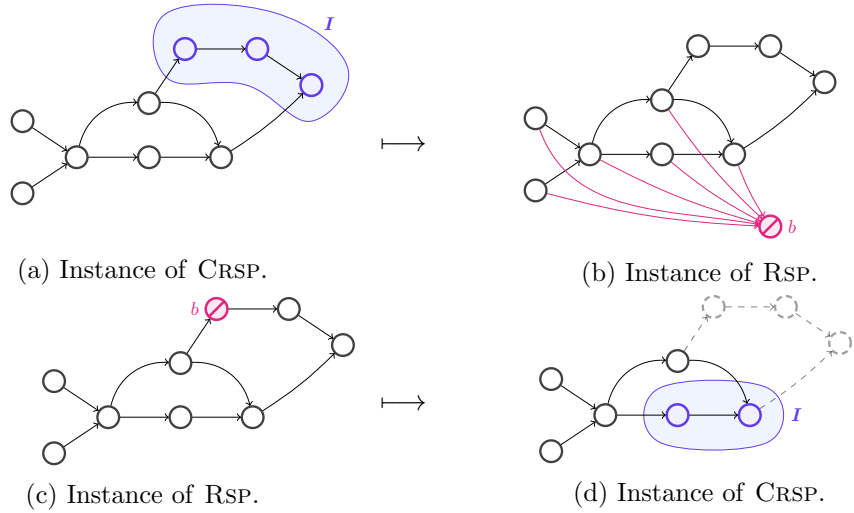
\begin{figure}[ht]
   \centering
   \begin{subfigure}[c]{.45\textwidth}
      \centering
      \hfill
      \scalebox{0.8}{
         \begin{tikzpicture}
	[
		scale=.6, every node/.style={circle,draw,scale=0.8, text=black, line width=1.2pt, opacity=.2, text opacity=1, draw opacity=1},
		minimum size=0.45cm,
		innocentNode/.style={fill=deep_purple!10, draw=deep_purple, text=deep_purple, line width=1.2pt},
		normalNode/.style={text opacity=1,opacity=0,draw opacity=1,color=darkgray,text=black, line width=1.1pt},
	]

	\draw[color=deep_purple,fill=ibm_blue!10]  plot[smooth cycle, tension=.7] coordinates {(5.5,9) (8.5,9) (10.5,7.5) (9.5,6) (7.7,7) (5.5,7.2)};
	\node[draw=none,color=deep_purple,opacity=0.2,draw opacity=0.2] at (10,8.7) {$\boldsymbol{I}$};

	\node[normalNode] (1) at (1.5, 6) {};
	\node[normalNode] (2) at (1.5, 4) {};
	\node[normalNode] (3) at (3, 5) {};
	\node[normalNode] (4) at (5, 5) {};
	\node[normalNode] (5) at (7, 5) {};
	\node[normalNode] (6) at (5, 6.5) {};
	\node[innocentNode] (7) at (6, 8) {};
	\node[innocentNode] (8) at (8, 8) {};
	\node[innocentNode] (9) at (9.5, 7) {};

	\node[draw=none] (10) at (8, 3) {};
	
	\path[->,color=black]
		(1) edge (3)
		(2) edge (3)
		(3) edge (4)
		(4) edge (5)
		(3) edge[bend left=40] (6)
		(6) edge[bend left=40] (5)
		(6) edge (7)
		(7) edge (8)
		(8) edge (9)
		(5) edge[bend right=6] (9)
	;

\end{tikzpicture}
      }
      \caption{Instance of \textsc{Crsp}.}
   \end{subfigure}
   $\longmapsto$
   \hfill
   \begin{subfigure}[c]{.45\textwidth}
      \centering
      \scalebox{0.8}{
         \begin{tikzpicture}
	[
		scale=.6, every node/.style={circle,draw,scale=0.8, text=black, line width=1.2pt, opacity=.2, text opacity=1, draw opacity=1},
		minimum size=0.45cm,
		normalNode/.style={text opacity=1,opacity=0,draw opacity=1,color=darkgray,text=black, line width=1.1pt},
		buggedNode/.style={text opacity=1,opacity=0.1,draw opacity=1,fill=ibm_magenta, color=ibm_magenta,text=black, line width=1.2pt},
	]

	\tikzset{
        halfcross/.pic = {
        \draw[rotate = 45,color=ibm_magenta] (-0.6*#1,0) -- (0.6*#1,0);
        }
    }

	\node[draw=none,color=ibm_magenta,opacity=0.2,draw opacity=0.2] at (8.6,3) {$b$};
	\draw[draw=none]  plot[smooth cycle, tension=.7] coordinates {(5.5,9) (8.5,9) (10.5,7.5) (9.5,6) (7.7,7) (5.5,7.2)};

	\node[normalNode] (1) at (1.5, 6) {};
	\node[normalNode] (2) at (1.5, 4) {};
	\node[normalNode] (3) at (3, 5) {};
	\node[normalNode] (4) at (5, 5) {};
	\node[normalNode] (5) at (7, 5) {};
	\node[normalNode] (6) at (5, 6.5) {};
	\node[normalNode] (7) at (6, 8) {};
	\node[normalNode] (8) at (8, 8) {};
	\node[normalNode] (9) at (9.5, 7) {};
	\node[buggedNode] (10) at (8, 3) {};
	
	\draw (8, 3) pic[black,opacity=1] {halfcross=11pt};

	\path[->,color=black]
		(1) edge (3)
		(2) edge (3)
		(3) edge (4)
		(4) edge (5)
		(3) edge[bend left=40] (6)
		(6) edge[bend left=40] (5)
		(6) edge (7)
		(7) edge (8)
		(8) edge (9)
		(5) edge[bend right=6] (9)
		
		(2) edge[bend right=6,color=ibm_magenta] (10)
		(3) edge[bend right=6,color=ibm_magenta] (10)
		(4) edge[bend right=6,color=ibm_magenta] (10)
		(5) edge[bend right=6,color=ibm_magenta] (10)
		(6) edge[bend right=6,color=ibm_magenta] (10)
	;
	
	\draw[->, color=ibm_magenta] (1) to[out=-70,in=170,swap] (10);

\end{tikzpicture}
      }
      \caption{Instance of \textsc{Rsp}.}
   \end{subfigure}
   \begin{subfigure}[c]{.45\textwidth}
      \centering
      \hfill
      \scalebox{0.8}{
         \begin{tikzpicture}
	[
		scale=.6, every node/.style={circle,draw,scale=0.8, text=black, line width=1.2pt, opacity=.2, text opacity=1, draw opacity=1},
		minimum size=0.45cm,
		normalNode/.style={text opacity=1,opacity=0,draw opacity=1,color=darkgray,text=black, line width=1.1pt},
		buggedNode/.style={text opacity=1,opacity=0.1,draw opacity=1,fill=ibm_magenta, color=ibm_magenta,text=black, line width=1.2pt},
	]

	\tikzset{
        halfcross/.pic = {
        \draw[rotate = 45,color=ibm_magenta] (-0.6*#1,0) -- (0.6*#1,0);
        }
    }

	\node[draw=none,color=ibm_magenta,opacity=0.2,draw opacity=0.2] at (5.4,8) {$b$};

	\node[normalNode] (1) at (1.5, 6) {};
	\node[normalNode] (2) at (1.5, 4) {};
	\node[normalNode] (3) at (3, 5) {};
	\node[normalNode] (4) at (5, 5) {};
	\node[normalNode] (5) at (7, 5) {};
	\node[normalNode] (6) at (5, 6.5) {};
	\node[buggedNode] (7) at (6, 8) {};
	\node[normalNode] (8) at (8, 8) {};
	\node[normalNode] (9) at (9.5, 7) {};
	
	\draw (6, 8) pic[black,opacity=1] {halfcross=11pt};

	\path[->,color=black]
		(1) edge (3)
		(2) edge (3)
		(3) edge (4)
		(4) edge (5)
		(3) edge[bend left=40] (6)
		(6) edge[bend left=40] (5)
		(6) edge (7)
		(7) edge (8)
		(8) edge (9)
		(5) edge[bend right=6] (9)
	;

\end{tikzpicture}
      }
      \caption{Instance of \textsc{Rsp}.}
   \end{subfigure}
   $\longmapsto$
   \hfill
   \begin{subfigure}[c]{.45\textwidth}
      \centering
      \scalebox{0.8}{
         \begin{tikzpicture}
	[
		scale=.6, every node/.style={circle,draw,scale=0.8, text=black, line width=1.2pt, opacity=.2, text opacity=1, draw opacity=1},
		minimum size=0.45cm,
		innocentNode/.style={fill=deep_purple!10, draw=deep_purple, text=deep_purple, line width=1.2pt},
		normalNode/.style={text opacity=1,opacity=0,draw opacity=1,color=darkgray,text=black, line width=1.1pt},
		transparentNode/.style={densely dashed,color=darkgray,text=black, line width=1.1pt, opacity=0.5},
	]

	\draw[color=deep_purple,fill=ibm_blue!10]  plot[smooth cycle, tension=.7] coordinates {(4.5,5.75) (7.5,5.75) (7.5,4.25) (4.5,4.25)};
	\node[draw=none,color=deep_purple,opacity=0.2,draw opacity=0.2] at (8.3,5) {$\boldsymbol{I}$};

	\node[normalNode] (1) at (1.5, 6) {};
	\node[normalNode] (2) at (1.5, 4) {};
	\node[normalNode] (3) at (3, 5) {};
	\node[innocentNode] (4) at (5, 5) {};
	\node[innocentNode] (5) at (7, 5) {};
	\node[normalNode] (6) at (5, 6.5) {};

	\node[transparentNode] (7) at (6, 8) {};
	\node[transparentNode] (8) at (8, 8) {};
	\node[transparentNode] (9) at (9.5, 7) {};

	\path[->,color=black]
		(1) edge (3)
		(2) edge (3)
		(3) edge (4)
		(4) edge (5)
		(3) edge[bend left=40] (6)
		(6) edge[bend left=40] (5)

		(6) edge[dashed,opacity=0.5] (7)
		(7) edge[dashed,opacity=0.5] (8)
		(8) edge[dashed,opacity=0.5] (9)
		(5) edge[dashed,bend right=6,opacity=0.5] (9)
	;

\end{tikzpicture}
      }
      \caption{Instance of \textsc{Crsp}.}
   \end{subfigure}
   \caption{\emph{Top.} Example of a transformation from an instance of \textsc{Crsp} to an instance of \textsc{Rsp}. \emph{Bottom.} Example of a transformation from an instance of \textsc{Rsp} to an instance of \textsc{Crsp}.}
   \label{fig:reductions_crsp_rsp}
\end{figure}

To prove that \textsc{Crsp} is NP-hard even restricted to binary DAGs,
we show that there is a polynomial reduction from the problem \textsc{Bounded (2,3)-SAT} (\textsc{Bsat}), proved to be NP-complete in \cite{TOVEY198485}.

\begin{definition} \textbf{Bounded (2,3)-SAT}  \\
   \textbf{Input.} A Boolean formula in Conjonctive Normal Form with the following restrictions:
   \begin{itemize}
      \item each clause contains 2 or 3 literals,
      \item each variable is present in at most 3 clauses.
   \end{itemize}

   \noindent\textbf{Output.} Whether there is an assignment of the variables that satisfies the formula.
   \label{def:bsat}
\end{definition}

Note that we can assume without loss of generality that each literal of any instance of \textsc{Bsat} is present in at most 2 clauses. Indeed, if a literal appears in 3 clauses, the opposite literal does not appear in the formula. Assigning the corresponding variable so that the literal is true makes the three clauses satisfied. Thus they can be removed from the formula without changing its satisfiability.

\begin{theorem}
   Confined Regression Search Problem is NP-hard even when the inputs are restricted to binary DAGs.
   \label{th:crsp_np_hard}
\end{theorem}

\begin{proof}
   We show that there exists a polynomial reduction from \textsc{Bsat} to \textsc{Crsp} with inputs restricted to binary DAGs, named \textsc{Bin-crsp}.
   The reduction algorithm takes as input a Boolean formula with $n$ variables and $m$ clauses and computes a binary DAG $D$ in polynomial time as follows.

   First, for each variable $X_i$, create two vertices $x_i$ and $\overline{x_i}$, two \emph{branching} vertices $b_i$ and $\overline{b_i}$, and one \emph{control} vertex $ct_i$. Connect $b_i$ and $ct_i$ to $x_i$, and $\overline{b_i}$ and $ct_i$ to $\overline{x_i}$ (see Figure~\ref{fig:reduction_algo_example} for an example of this \emph{variable gadget}).

   Then for each clause $C_j$, create a vertex $c_j$. For each literal $x_i$ (resp. $\overline{x_i}$) in the clause $C_j$, add an arc from $c_j$ to $b_i$ (resp. $\overline{b_i}$).
   Finally, create three isolated vertices $t_1,t_2$ and $t_3$.
   Set as innocent all vertices except $t_1,t_2,t_3$, the $(c_j)_{1\le j\le m}$ and the $(ct_i)_{1\le j\le n}$ vertices.

   Note that since every literal appears in at most two clauses, each vertex $b_i$ and $\overline{b_i}$ has indegree at most two, and the resulting DAG is binary. Moreover, the size of the Boolean formula and the size of the DAG are polynomials with respect to $(n+m)$.

   For example, Figure~\ref{fig:reduction_algo_example} shows the DAG from the reduction of the following formula in \textsc{Bsat} form:
   \[(x_1 \lor \overline{x_2}) \land (\overline{x_1} \lor \overline{x_2} \lor \overline{x_3}).\]

   \begin{figure}[ht]
      \centering
      \scalebox{0.8}{
         \begin{tikzpicture}
	[
		scale=.6, every node/.style={circle,draw,scale=0.8, text=black, line width=1.2pt, opacity=1, text opacity=1, draw opacity=1},
		innocentNode/.style={fill=deep_purple!10, draw=deep_purple, text=deep_purple, line width=1.2pt},
		normalNode/.style={text opacity=1,opacity=1,draw opacity=1,draw=darkgray,text=darkgray, line width=1.pt, fill=white},
	]

	\tikzset{
		show curve controls/.style={
			decoration={
				show path construction,
				curveto code={
					\draw [blue, dashed]
						(\tikzinputsegmentfirst) -- (\tikzinputsegmentsupporta)
						node [at end, cross out, draw, solid, red, inner sep=2pt]{};
					\draw [blue, dashed]
						(\tikzinputsegmentsupportb) -- (\tikzinputsegmentlast)
						node [at start, cross out, draw, solid, red, inner sep=2pt]{};
				}
			}, decorate
		}
	}

	\draw[color=deep_purple,fill=ibm_blue!10,opacity=1]  plot[smooth cycle, tension=.6] coordinates {(7,18.5) (11,17.5) (11,2.5) (7,1.5) (7,10)};
	\node[draw=none,color=deep_purple,opacity=0.2,draw opacity=0.2] at (11,19) {$\boldsymbol{I}$};
	
	\draw[color=ibm_orange,fill=ibm_orange!20, line width=1.5pt, dashed, text opacity=1, draw opacity=1, fill opacity=0.4]  plot[smooth cycle, tension=.8] coordinates {(5.75,12.5) (10.75,12) (10.75,8) (5.75,7.5)};
	\node[draw=none,color=ibm_orange] at (13.2,10) {\textbf{$\boldsymbol{X_2}$ gadget}};

	\node[normalNode] (1) at (5.5, 16) {$ct_1$};
	\node[innocentNode] (2) at (10, 17) {$x_1$};
	\node[innocentNode] (3) at (10, 15) {$\overline{x_1}$};
	\node[innocentNode] (16) at (8, 18) {$b_1$};
	\node[innocentNode] (17) at (8, 14) {$\overline{b_1}$};

	\node[normalNode] (4) at (6, 10) {$ct_2$};
	\node[innocentNode] (5) at (10, 11) {$x_2$};
	\node[innocentNode] (6) at (10, 9) {$\overline{x_2}$};
	\node[innocentNode] (18) at (8, 12) {$b_2$};
	\node[innocentNode] (19) at (8, 8) {$\overline{b_2}$};

	\node[normalNode] (7) at (5.5, 4) {$ct_3$};
	\node[innocentNode] (8) at (10, 5) {$x_3$};
	\node[innocentNode] (9) at (10, 3) {$\overline{x_3}$};
	\node[innocentNode] (20) at (8, 6) {$b_3$};
	\node[innocentNode] (21) at (8, 2) {$\overline{b_3}$};

	\node[normalNode] (10) at (0, 14) {$c_1$};
	\node[normalNode] (11) at (0, 6) {$c_2$};

	\node[normalNode] (12) at (0, 3) {$t_1$};
	\node[normalNode] (13) at (0, 1.5) {$t_2$};
	\node[normalNode] (14) at (0, 0) {$t_3$};

	\path[->,color=darkgray]
		(1) edge (2)
		(1) edge (3)
		(16) edge[bend left=5] (2)
		(17) edge[bend right=5] (3)

		(4) edge (5)
		(4) edge (6)
		(18) edge[bend left=5] (5)
		(19) edge[bend right=5] (6)

		(7) edge (8)
		(7) edge (9)
		(20) edge[bend left=5] (8)
		(21) edge[bend right=5] (9)

		(10) edge[bend left=25] (16)
		(10) edge[bend right=21] (19)
		
		(11) edge[bend left=21] (17)
		(11) edge[bend left=5] (19)
		(11) edge[bend right=25] (21)
	;

\end{tikzpicture}
      }
      \caption{The graph resulting from the reduction of the formula ${(x_1 \lor \overline{x_2}) \land (\overline{x_1} \lor \overline{x_2} \lor \overline{x_3})}$.}
      \label{fig:reduction_algo_example}
   \end{figure}
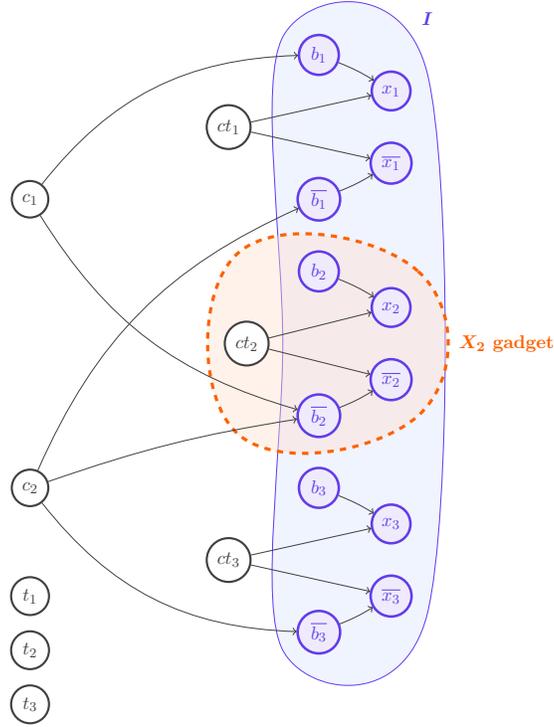

   We show that the previous transformation is a reduction: the Boolean formula has a satisfiable assignment if and only if there is a strategy using $n+3$ queries in the worst-case scenario.

   First, let us suppose that $F$ has a satisfiable assignment ${X=\{X_1,X_2,\dots,X_n\}}$ where $X_i$ is true or false and let us describe a strategy that solves \textsc{Bin-crsp}.

   For each $i$ between 1 and $n$, query $x_i$ if $X_i$ is true, $\overline{x_i}$ if $X_i$ is false. If none of the queries reveals a bugged vertex, query the three vertices $t_1,t_2,t_3$. If one is discovered as bugged then it is the faulty vertex. If they are all clean then there is no faulty vertex in the DAG.
   Otherwise, some vertex $x_i$ or $\overline{x_i}$ is found to be bugged.
   Without loss of generality, let us assume it is a vertex $x_i$ corresponding to a positive literal.
   Then, we query the parents of $b_i$ (i.e. the vertices corresponding to clauses where $x_i$ appears), which is done in no more than two queries.
   If a parent of $b_i$ is discovered as bugged, then it is the faulty vertex.
   Otherwise, $ct_i$ is the faulty vertex.

   \medskip

   Conversely, let us suppose there exists a strategy solving \textsc{Bin-crsp} in at most $n+3$ queries. We prove that the Boolean formula $F$ has a satisfiable assignment.

   Consider the scenario with no faulty vertex in the DAG. We show that exactly $n+3$ queries are required and that a solution of the formula can be deduced from which vertices have been queried.
   All the non-innocent vertices must be cleared to guaranty that the DAG has no faulty vertex, which can be done by querying the vertex itself or one of its descendants.

   Concerning the three terminal vertices $t_1, t_2$ and $t_3$, they are isolated so the strategy must query each of them to know that they are clean.
   To know that a $ct_i$ vertex is clean we must query the $ct_i$ vertex itself or one of its two descendants $x_i$ or $\overline{x_i}$. Thus we must make at least one query in each group $\{ct_i,x_i,\overline{x_i}\}$ for $1 \le i\le n$.
   By our assumption that the strategy uses at most $n+3$ queries, there must be exactly one query in each group $\{ct_i,x_i,\overline{x_i}\}$.

   For each $1 \le i\le n$ , we assign $X_i$ as $$X_i = \left\{\begin{array}{ll}\text{true} & \textrm{if } x_i \textrm{ is queried} \\ \text{false} & \textrm{if } \overline{x_i} \textrm{ is queried} \\ \text{true} \textrm{ (arbitrarily)} & \textrm{if } ct_i \textrm{ is queried}\end{array}\right.$$

   Since all clause vertices must be cleared, at least one of the descendants of each $c_j$ is queried ($c_j$ cannot be queried since we have already done $n+3$ queries). This implies that some literal of the corresponding clause $C_j$ is true in the above assignment.\mbox{\qedhere}
\end{proof}

As a corollary we have an alternative proof of the NP-completeness of \textsc{Rsp}.
\begin{corollary}
   The problems \textsc{Rsp}, \textsc{Crsp} and \textsc{Bin-crsp} are NP-complete.
   \label{cor:rsp_crsp_crspbin_npcomplete}
\end{corollary}

\begin{proof}
   On one hand, certificates for these 3 problems are strategy trees with a polynomial number of nodes. Thus \textsc{Crsp}-bin, \textsc{Crsp} and \textsc{Rsp} are in NP.
   On the other hand, \textsc{Bin-crsp} is NP-hard by Theorem~\ref{th:crsp_np_hard}. As a generalisation of \textsc{Bin-crsp}, \textsc{Crsp} is NP-hard. Since \textsc{Rsp} is equivalent to \textsc{Crsp}, as shown above, \textsc{Rsp} is also NP-hard.
\end{proof}

In the light of the above reduction, let us explain why proving the NP-hardness of the Regression Search Problem for binary graphs (abbreviated \textsc{Bin-rsp}) seems to be arduous. Observe that the reduction from \textsc{Bsat} to \textsc{Bin-crsp} (as well as the existing reductions in the literature) encodes the assignment of the $n$ variables of the SAT formula into a sequence of queries from a specific scenario. Thereby the certificate is a strategy tree of depth at least $n$. By Theorem~\ref{theo:bound_golden_binary_dag}, any binary DAG requiring this number of queries to solve \textsc{Bin-rsp} has at least $\phi^{n-1}$ vertices and thus would not be polynomial in the size of the formula.

Therefore, to find a reduction from a SAT problem to \textsc{Bin-rsp}, it is not possible to encode the assignment of the variables in a single branch of the strategy tree. One needs to consider the tree widthwise, and use the sets of queried vertices in different scenarios. We do not know of any reduction using this thorny approach yet.

\section{Conclusion}

In summary, this paper has established that \gb\ can be very inefficient on very particular digraphs, but under the reasonable hypothesis that merges must not concern more than $2$ branches each, it is proved to be a good approximation algorithm.
This study has also developed a new algorithm, \golden, which displays better theoretical results than \gb.

Notably, some open questions remain:
\begin{itemize}
   \item Is \textsc{Bin-rsp}, the Regression Search Problem for binary DAGs, still NP-complete, as discussed in the previous section?
   \item How do \gb\ and \golden\ compare in terms of average-case complexity?
   \item In \gb\ and in \golden, one never queries vertices which were eliminated from the set of candidates for the \faulty. However, we could speed up the procedure by never removing any vertex after queries. For example, consider the DAG from Figure~\ref{fig:octopus_with_comb}. If we choose $\vertex{v_7}$ as first query and it is bugged, then we remove all $\vertex{u_i}$ (the non-ancestors of  $\vertex{v_7}$). However, querying the vertices $\vertex{u_i}$ in the comb would be more efficient. Could we obtain an improved algorithm by authorising such queries?
   \item If we restrict the DAGs to be trees (oriented from the leaves to the root) with unbounded indegree, are \gb\ and \golden\ good approximation algorithms? We conjecture that they are $2-$approximation algorithms for trees. (We have found examples where the ratio is $2$.)
\end{itemize}

Finally it would be interesting to study the number of queries in the worst-case scenario when the input DAG is taken at random. Indeed, most of the examples described in this paper are not very likely to appear in reality. The notion of randomness for a digraph emanating from a VCS is therefore quite interesting and deserves to be developed. One could for example define a theoretical probabilistic model based on existing workflows.
It would be also quite useful to use random samplers for VCS repositories in order to constitute benchmarks on demand.

\section*{Aknowledgments}

This research was conducted within the project ``DynNet'' supported by the Normandy region.
The European project ``DynNet'' is funded by the European Union within the framework of the Operational Programme ERDF/ESF 2014-2020.
R.L. was also supported by the Normandy RIN project AAAA.

\bibliography{references}

\bibliographystyle{plain}

\end{document}